\documentclass{article}
\pdfoutput=1
\usepackage{arxiv}
\usepackage[utf8]{inputenc} 
\usepackage[T1]{fontenc}    
\usepackage{hyperref}       
\usepackage{url}            
\usepackage{booktabs}       
\usepackage{siunitx}
\usepackage{array}   
\newcolumntype{R}{>{$}\num{r}<{$}} 
\usepackage{amsfonts}       
\usepackage{amssymb, mathtools, xcolor,makecell,enumerate}
\usepackage{amsmath, amsthm, accents}
\usepackage{nicefrac}       
\usepackage{microtype}      
\usepackage{graphicx}
\usepackage{doi}
\usepackage{caption, float}
\usepackage[colorinlistoftodos,prependcaption,textsize=small]{todonotes}

\usepackage[noend, ruled, linesnumbered, noline, longend]{algorithm2e}
\SetNlSty{bfseries}{\color{black}}{}
\usepackage{enumerate}
\usepackage{multirow}
\usepackage{fancyhdr}
\usepackage{titlesec}
\usepackage{tcolorbox}
\usepackage{makecell}
\usepackage{xcolor}
\usepackage{cleveref}
\usepackage{tikz}
\usetikzlibrary{calc}
\usetikzlibrary{positioning}
\usepackage{url}
\pgfdeclarelayer{bg} 
\pgfsetlayers{bg,main}

\newcommand{\N}{\mathbb{N}}
\newcommand{\R}{\mathbb{R}}

\newcommand{\NQP}{\mathcal{NQP}}
\renewcommand{\P}{\mathcal{P}}

\DeclareMathAlphabet{\mymathbb}{U}{BOONDOX-ds}{m}{n}

\newcommand{\dom}{\preceq_{D}}
\newcommand{\ndom}{\npreceq_{D}}
\newcommand{\lexsmaller}{\prec_{lex}}
\newcommand{\lexsmallereq}{\preceq_{lex}}

\newcommand{\incoming}[1]{\delta^{-}(#1)}
\newcommand{\outgoing}[1]{\delta^{+}(#1)}
\newcommand{\bigoh}[1]{\mathcal{O}\left(#1\right)}
\newcommand{\lpl}{\mathtt{lastProcessedPath}}

\newcommand{\Null}{\mathtt{NULL}}

\newcommand{\mda}{\hyperref[algo:mda*]{T-MDA}}
\newcommand{\bda}{\hyperref[algo:bda*]{T-BDA}}
\newcommand{\boaClassic}{$\text{BOA}^*$}
\newcommand{\boaEnh}{$\text{BOA}^*_{enh}$}
\newcommand{\boaBidi}{$\text{BOBA}^*$}

\SetCommentSty{mycommfont}

\definecolor{mintgreen}{rgb}{0.0, 0.47, 0.44}
\definecolor{alizarin}{rgb}{0.82, 0.1, 0.26}
\definecolor{applegreen}{rgb}{0.55, 0.71, 0.0}

\newtheorem{theorem}[]{Theorem}[]

\newtheorem{lemma}[theorem]{Lemma}
\newtheorem{corollary}[theorem]{Corollary}
\newtheorem{proposition}[theorem]{Proposition}

\theoremstyle{definition}
\newtheorem{definition}[theorem]{Definition}
\newtheorem{example}[theorem]{Example}
\newtheorem{remark}{Remark}

\title{Targeted Multiobjective Dijkstra Algorithm}
  
\newcommand\blfootnote[1]{%
  \begingroup
  \renewcommand\thefootnote{}\footnote{#1}%
  \addtocounter{footnote}{-1}%
  \endgroup
}

\author{ \href{https://orcid.org/0000-0002-4197-0893}{\includegraphics[scale=0.06]{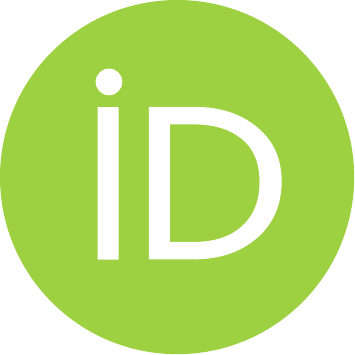}\hspace{1mm}Pedro Maristany de las Casas}\\
	Network Optimization \\
	Zuse Institute Berlin\\
	Germany, Berlin, 14195\\
	\texttt{maristany@zib.de} \\
	\And
	\href{https://orcid.org/0000-0002-5735-9415}{\includegraphics[scale=0.06]{orcid.pdf}\hspace{1mm}Luitgard Kraus} \\
	Network Optimization \\
	Zuse Institute Berlin\\
	Germany, Berlin, 14195\\
	\texttt{kraus@zib.de} \\
	\And
	\href{https://orcid.org/0000-0003-0681-4585}{\includegraphics[scale=0.06]{orcid.pdf}\hspace{1mm}Antonio Sedeño-Noda} \\
	Departamento de Matemáticas, Estadística e Investigación Operativa\\
	Universidad de La Laguna\\
	Spain, San Cristóbal de La Laguna, 38271\\
	\texttt{asedeno@ull.edu.es} \\
	\And
	\href{https://orcid.org/0000-0001-7223-9174}{\includegraphics[scale=0.06]{orcid.pdf}\hspace{1mm}Ralf Borndörfer} \\
	Network Optimization \\
	Zuse Institute Berlin\\
	Germany, Berlin, 14195\\
	\texttt{borndoerfer@zib.de} \\
}

\renewcommand{\headeright}{}
\renewcommand{\undertitle}{}
\renewcommand{\shorttitle}{Targeted Multiobjective Dijkstra Algorithm}

\hypersetup{
	pdftitle={Multiobjective Dijkstra A-Star Algorithm},
	pdfsubject={},
	pdfauthor={Pedro Maristany de las Casas, Luitgard Kraus, Antonio Sedeno-Noda, Ralf Borndoerfer},
	pdfkeywords={Multiobjective Shortest Path, Label Setting Algorithm, Heuristic Search, Dynamic Programming},
}

\begin{document}
	\maketitle
    \blfootnote{The researches affiliated to the Zuse Institute Berlin conducted this work within the Research Campus MODAL - Mathematical Optimization and Data Analysis Laboratories -, funded by the German Federal Ministry of Education and Research (BMBF) (fund number 05M20ZBM).}
	\begin{abstract}
		In this paper, we introduce the Targeted Multiobjective Dijkstra Algorithm (T-MDA), a label setting algorithm for the One-to-One Multiobjective Shortest Path (MOSP) Problem.
		The T-MDA is based on the recently published Multiobjective Dijkstra Algorithm (MDA) and equips it with A*-like techniques. The resulting speedup is comparable to the speedup that the original A* algorithm achieves for  Dijkstra's algorithm.
		Unlike other methods from the literature, which rely on special properties of the biobjective case, the T-MDA works for any dimension. To the best of our knowledge, it gives rise to the first efficient implementation that can deal with large scale instances with more than two objectives.
		A version tuned for the biobjective case, the T-BDA, outperforms state-of-the-art methods on almost every instance of a standard benchmark testbed that is not solvable in fractions of a second.
	\end{abstract}


	\section{Introduction}
	The Multiobjective Shortest Path (MOSP) Problem is an extension of the classical Shortest Path Problem in which every arc bears a multidimensional vector of cost attributes.
	The multidimensional costs of a path are then the sum of the cost vectors along its arcs. 
	If two paths $p$ and $q$ connect the same end nodes in the graph, $p$ \emph{dominates} $q$ if the cost vector of $p$ is not worse than the cost vector of $q$ in any dimension and it is better than the cost vector of $q$ in at least one dimension. 
	A path is optimal, also called \emph{efficient}, if it is not dominated by any other path connecting the same end nodes. The cost vector of an efficient path is called a \emph{non-dominated} cost vector. In our setting, MOSP problems ask for a representative efficient path for every non-dominated cost vector; the number of these representatives can be exponential.  
	Despite this principal intractability, MOSP problems arising from real world scenarios can often be solved efficiently,  
	indeed, significant advances in this respect have been made in recent years.
	Sedeño-Noda and Colebrook \cite{Sedeno19} introduced the Biobjective Dijkstra Algorithm (BDA) for the Biobjective Shortest Path (BOSP) Problem that improved the known complexity bounds for this problem and is also efficient in practice.
	Maristany et al. \cite{Casas2021a} introduced the Multiobjective Dijkstra Algorithm (MDA), which generalizes the BDA to the multidimensional case. 
	Its running time and space complexity bounds improve the state of the art in theory and practice. 
	
    This article considers the One-to-One MOSP Problem, which aims at determining a complete set of optimal paths between two designated nodes in the input graph.
	One-to-One MOSP algorithms can be equipped with different pruning techniques that discard irrelevant paths early during the search (e.g., \cite{Pulido14}, \cite{Duque15}, \cite{Sedeno19}, \cite{Casas2021a}).
	In addition, the search can be guided towards the target. 
	A first step in this direction was taken by Mandow and Pérez de la Cruz \cite{Mandow05, Mandow2010}, who extend the classical label-setting MOSP algorithm by Martins \cite{Martins84} with likewise classical A* techniques. 
	
   	Using a property known as \emph{dimensionality reduction} \cite{Pulido15}, it is possible to decide in constant time whether newly explored paths are efficient in the biobjective case.
   	The algorithms considered in this paper store explored paths in a lexicographically sorted priority queue.
    Ulloa et al. \cite{Ulloa20} use a lazy queue management in conjunction with the dimensionality reduction property in their Biobjective $\text{A}^*$ (\boaClassic{}) algorithm that is again based on Martins' algorithm. 
    In Martins' algorithm, the queue can contain exponentially many paths but all of them are guaranteed to be efficient. Paths are removed from the queue as soon as the algorithm explores a path that dominates them in what is called a \emph{merge operation}. 
    These operations are known to be costly in practice, which motivates the lazy queue management in the \boaClassic{} algorithm. There, the queue may contain dominated paths but thanks to the dimensionality reduction, these are discarded in constant time immediately after their extraction from the queue. Thus, unnecessary expansions of dominated paths are avoided.
    Combining this lazy path management with $\text{A}^*$ makes the \boaClassic{} algorithm competitive. Ahmadi et al.'s \cite{Ahmadi21} \boaEnh{} algorithm improves this method further. 
	
	The first contribution in this paper is an extension of the MDA by the A* technique to obtain the \emph{Targeted Multiobjective Dijkstra Algorithm} \mda{}. 
	We propose a variant that trades memory consumption for run time without harming the running time bounds, but improving practical performance.
	Indeed, and to the best of our knowledge, this version of the \mda{} is the first efficient implementation of a One-to-One MOSP algorithm that can deal with 
	large scale instances such as road networks. The second contribution is the \emph{Targeted Biobjective Dijkstra Algorithm} (T-BDA), a version of the \mda{} tuned for the Biobjective Shortest Path (BOSP) Problem. It outperforms state-of-the-art methods on a comprehensive benchmark set \cite{Sedeno19, Ulloa20, Ahmadi21} unless the considered instance can be solved in fractions of a second. 

	The paper is structured as follows. In \Cref{sec:o2o-mosp} we give a formal definition of the One-to-One Multiobjective Shortest Path Problem.
	We then proceed to describe the \mda{} algorithm, prove its correctness, and analyze its theoretical space and running time complexity bounds in \Cref{sec:Algorithm}.
	Furthermore, we discuss how to integrate the enhancements for BOSP problems that were discussed in \cite{Ahmadi21} and \cite{Cabrera20} into the \bda{} in Section \ref{sec:bidimensional}.
	Finally, in \Cref{sec:experiments} we report on extensive computational results on different types of graphs for two and three dimensional MOSP problems. 
	
	\section{The Targeted Multiobjective Dijkstra Algorithm}
		
	This section serves as an introduction to the \emph{Targeted Multiobjective Dijkstra Algorithm} (\mda).
	This algorithm solves the One-to-One Multiobjective Shortest Path Problem described in \Cref{sec:o2o-mosp}. 
	In \Cref{sec:Heuristics and Path Pruning}, we introduce some concepts needed for the algorithm. 
	In \Cref{sec:Algorithm}, we introduce the algorithm itself before proving its correctness in \Cref{sec:Correctness}.
	Whenever we write $x \leq y$ for vectors $x, \, y \in \R^d$, $d\in \N$, we mean $x_i \leq y_i$ for all $i \in \{1, \dots, d\}$.
			
	\subsection{One-to-One Multiobjective Shortest Path Problem}\label{sec:o2o-mosp}
	
	Consider a digraph $G = (V,A)$.
	Given multidimensional arc cost vectors $c_a \in \R^d_{\geq 0}$, $d \in \N$, for every arc $a \in A$, the costs $c(p)$ of a path $p$ are defined as the sum of its arc costs, i.e., $c(p) := \sum_{a \in p} c_a$.
	For a set $\mathcal{P}$ of paths, we use $c(\mathcal{P})$ to refer to the set of costs of all paths in $\mathcal{P}$, i.e., $c(\mathcal{P}) := \{c(p) \ |\  p \in \mathcal{P}\}\subseteq \R^d_{\geq0}$.
	The following definition formalizes the notion of optimality to the multiobjective case.
	
	\begin{definition}[Dominance and Efficiency]\label{def:dominance}
		Consider two nodes $v,\,w \in V$ and two $v$-$w$-paths $p$ and $q$ in $G$.
		Then, $p$ \emph{dominates} $q$ if $c_i(p) \leq c_i(q)$ for all $i \in \{1, \dots, d\}$ and there is a $j \in \{1, \dots, d\}$ s.t. $c_j(p) < c_j(q)$.
		The path $p$ is called \emph{efficient} if there is no 
		$v$-$w$-path in $G$ that dominates it. Otherwise, $p$ is called a \emph{dominated} path.
		If $p$ is efficient, its cost vector $c(p)$ is called \emph{non-dominated}. Otherwise it is called \emph{dominated}.
	\end{definition}
	
	Let $\mathcal{P}_{st}$ be the set of efficient $s$-$t$-paths in $G$.
	In general, for any non-dominated cost vector in $c(\mathcal{P}_{st})$, there can be multiple efficient paths in $G$.
	In this paper, we want to find a \emph{minimal complete set} of efficient $s$-$t$-paths, i.e., a set of $s$-$t$-paths such that for every non-dominated cost vector in $c(\mathcal{P}^*_{st})$, the presented algorithms outputs exactly one efficient $s$-$t$-path.
	
	\begin{definition}[One-to-One Multiobjective Shortest Path Problem]
		Given a digraph $G=~(V,A)$, a source node $s \in V$, a target node $t \in V$ and $d$-dimensional arc costs vectors $c_a \in \R^d_{\geq}$, $a \in A$, the \emph{One-to-One Multiobjective Shortest Path (MOSP) Problem} is to find a minimal complete set of efficient $s$-$t$-paths. We denote a One-to-One MOSP instance by $(G,s,t,c)$.
	\end{definition}
	
	\subsection{Heuristics and Dominance Bounds}\label{sec:Heuristics and Path Pruning}
	The following definitions and results are a generalization of those valid for the single criteria Shortest Path Problem. Throughout this section, we assume that a $d$-dimensional One-to-One MOSP instance $\mathcal{I} = (G,s,t,c)$ is given.
	
	\begin{definition}[(Admissible and Monotone) Heuristic]\label{def:monotone_heuristic}
		A multiobjective heuristic $\pi: V \mapsto \R_\geq^d$ is called \textit{an admissible and monotone heuristic for $\mathcal{I}$}, respectively, if
		\begin{description}
			\item[\textnormal{Admissibility:}]
				\begin{itemize}
				    \item[]
					\item[i.]
					    $\pi(v) \leq c(p),\, \quad \forall v \in V,\, p$ efficient $v$-$t$-path,
				\end{itemize}
			\item[\textnormal{Monotonicity:}]
    			\begin{itemize}
    			    \item[]
    				\item[i.]
    				    $\pi(t) = 0$, and
    				\item[ii.]
    				    $\pi(u) \leq c_{uv} + \pi(v), \quad \forall (u,v) \in A$. 
			\end{itemize}
		\end{description}
	\end{definition}
	
	\begin{remark}
		As in the single objective case, a monotone heuristic is always admissible. 
	\end{remark}
	
	Throughout the paper, we only consider monotone and admissible heuristics. For the sake of brevity, we refer to them just as \emph{heuristics} from now on. The notation $\pi$ to refer to the heuristics is motivated by the dual linear program of the single criteria Shortest Path Problem.
	There, the set of feasible solutions corresponds to the set of monotone heuristics, often denoted by $\pi$. 
	According to the notation from this field, the \emph{reduced costs} of an arc $(u,v) \in A$ are $\bar{c}_{uv} := c_{uv} + \pi(v) - \pi(u)$. For an $s$-$v$-path $p$, $v \in V$, we have 
	\[ 
		\sum_{a \in p} \bar{c}(a) = c(p) + \pi(v) - \pi(s)
	\]
	and since all relevant paths in this paper start at the source node $s$ of a given MOSP instance, the constant term $\pi(s)$ can be ignored.
	Thus, we define 
	\begin{equation}
		\bar{c}(p) := c(p) + \pi(v).
	\end{equation}
	\begin{lemma}
		Let $\pi$ be a heuristic for $\mathcal{I}$. For every arc $a \in A$, we have $\bar{c}_a \geq 0$.
	\end{lemma}
	\begin{proof}
		The statement follows directly from the monotonicity of $\pi$.
	\end{proof}
	
	The following is a trivial consequence of the last lemma and is used in the proofs from \Cref{sec:Correctness}.
	\begin{corollary}\label{lem:increasing_potential}
		Let $\pi$ be a heuristic for $\mathcal{I}$. Assume that node $u$ comes before node $v$ along any simple path $p$ in $G$ starting at $s$.
		Then,
		$$\bar{c}(p_{su}) \leq \bar{c}(p_{sv}).$$
	\end{corollary}

	Deciding whether partial solutions can become efficient solutions is particularly difficult at the beginning of an algorithm since little information about the set of efficient solutions is known a priori. 
	The following definition describes cost vectors that act as a threshold to prevent MOSP algorithms from processing paths that can not be extended towards the target node to produce an efficient $s$-$t$-path.

	\begin{definition}[Dominance bound]
		A \emph{dominance bound for $\mathcal{I}$} is a vector $\beta \in \R^d$ s.t. there is no efficient $s$-$t$-path that is dominated by $\beta$.
	\end{definition}

	In \Cref{sec:Algorithm} we assume that a heuristic and a dominance bound are part of the input of the described algorithm.
	In \Cref{sec:computeHeuristics} we describe how to find heuristics and dominance bounds for any One-to-One MOSP instance in a preprocessing stage.
	
	\subsection{The Algorithm.}\label{sec:Algorithm}
	
	The \mda{} is an adaptation of the label-setting Multiobjective Dijkstra Algorithm (MDA) \cite{Casas2021a}. 
	Since it is a One-to-One MOSP algorithm, the \mda{} does not need to store efficient $s$-$v$-paths for intermediate nodes $v \in V\backslash\{s,t\}$ if these paths are provably not extendable to become efficient $s$-$t$-paths. 
	Moreover, the \mda{} is designed with practical performance in mind. 
	Given the MDA's running time and space complexity, we found that in the \mda{}, we can increase the space complexity without harming the theoretical running time bound but achieving better results in practice.
	The details of this trade-off are discussed in \Cref{sec:nclListsComplexity}.
	
	The input of the \mda{} is a $d$-dimensional MOSP instance $(G,s,t,c)$, a heuristic $\pi$, and dominance bound $\beta_t$.
	The paths in the priority queue $Q$ used in the algorithm are sorted in lexicographically (lex.) non-decreasing order w.r.t. $\bar{c}$. 
	
	\begin{definition}[$\lexsmaller$]
		Given two cost vectors $x,\, y \in \R^d$, $x$ is lexicographically smaller than $y$ if the first cost component $j \in \{1, \dots d\}$, in which both vectors differ, fulfills $x_j < y_j$.
		In this case, we write $x \lexsmaller y$. 
		For paths $p$ and $q$, we say that $p$ is lex. smaller than $q$ if $\bar{c}(p) \lexsmaller \bar{c}(q)$. 
		We use $\lexsmallereq$ canonically.
	\end{definition}
	
	The output of the \mda{} is a minimal complete set of efficient $s$-$t$-paths.
	Let $v\in V$ and consider two efficient $s$-$v$-paths $p$ and $q$ with $\bar{c}(p) = \bar{c}(q)$.
	By \Cref{def:dominance}, $p$ and $q$ do not dominate each other.
	However, if the \mda{} algorithm already found and stored $p$, $q$ can be discarded when it is processed.
	Thus, a reflexive binary relation is needed to check dominance and/or equivalence of cost vectors.
	We use the following notation to stress the importance of these checks in our algorithms. 
	\begin{definition}[Dominance or Equivalence Operator $\dom$]
		Consider $x, y \in \R^d$ and $\mathcal{X} \subseteq  \R^d$. We write $x \dom y$ if $x \leq y$ and $\mathcal{X} \dom y$ if there exists a $z \in \mathcal{X}$ such that $z \dom y$.
	\end{definition}
	For ease of notation, from now on we say that a path $p$ dominates a path $q$ if $c(p) \dom c(q)$.
	
	\begin{algorithm}
		\small{
			\SetKwInOut{Input}{Input} \SetKwInOut{Output}{Output}
			\Input{MOSP instance $\mathcal{I} = (G, s, t, c)$, heuristic $\pi$ for $\mathcal{I}$, dominance bound $\beta_t$ for $\mathcal{I}$.}
			\BlankLine
			\Output{Minimal complete set $\mathcal{P}_{st}$ of efficient $s$-$t$-paths.}
			\BlankLine
			
			Priority Queue $Q \leftarrow{} \emptyset{}$\label{algo:mda*:initStart}\tcp*{$Q$ is sorted according to the $\bar{c}$ values of the contained paths.}
			$\forall (u,v)\in A$ -- Promising explored $s$-$v$-paths via $u \in V$ $\NQP_{uv} \leftarrow \emptyset$\label{algo:mda*:initNcl}\;
			$\forall v \in V$ -- Efficient Paths $\mathcal{P}_{sv} \leftarrow \emptyset$ \label{algo:mda*:initPerm}\;
			$p_{\text{init}} \leftarrow{} ()$\;
			$Q \leftarrow{} Q.\mathtt{insert}(p_{\text{init}})$\;
			\BlankLine
			\While{$Q \neq \emptyset$}{
				$p \leftarrow{} Q.\mathtt{extractMin}()$ \label{algo:mda*:extraction}\;
				$v \leftarrow{} \mathit{head}(p)$\tcp*[r]{\texttt{head} returns the last node of $p$. If $p = p_{\text{init}}$, it returns $s$.}
				$\mathcal{P}_{sv}  \leftarrow (\mathcal{P}_{sv}, p)$\label{algo:mda*:linePushBack2}\;
				$p^{new}_v \leftarrow{} \ref{algo:ncla*}(p, \NQP, \incoming{v}, \mathcal{P}, \beta_t)$ \label{algo:mda*:ncl}\;
				\lIf{$p^{new}_v \neq \Null$} 
					{$Q.\mathtt{insert}(p^{new}_v)$\label{algo:mda*:insert}}
				\If{$v == t$} {
					\textbf{continue}\;
				}
				\lFor{$w \in \outgoing{v}$}{
					$Q \leftarrow{} \ref{algo:propagate}(p, w, Q, \mathcal{P}, \NQP, \beta_t)$\label{algo:mda*:propagate}
				}
			}
			\Return $\P_t$;
			
			\caption[caption]{Multiobjective Dijkstra A*}\label{algo:mda*}}
	\end{algorithm}
	
	We proceed with the description of the \mda{}~(\Cref{algo:mda*}).
	It distinguishes two types of $s$-$v$-paths for any node $v \in V$. 
	The set $\mathcal{P}_{sv}$ contains already found efficient $s$-$v$-paths, which we call \emph{permanent paths}.
	Let $p$ be an $s$-$v$-path with $(u,v) \in A$ as its last arc.
	It is an \emph{explored path} if its $s$-$u$-subpath is already permanent and $p$ itself has already been considered by the algorithm without having been able to definitively decide whether $p$ is a relevant subpath of an efficient $s$-$t$-path.
	The sets of explored and permanent paths are always disjoint.
	This means that whenever an explored path is made permanent or finally discarded, it is no longer regarded as an explored path.
	
	The \mda{} uses a priority queue $Q$ to store a lex. minimal explored $s$-$v$-path for any $v \in V$. 
	
	\begin{definition}[Queue paths]\label{def:queuePath}
		An $s$-$v$-path $q$ is called \emph{queue path for $v$} if
		\begin{enumerate}[i.]
			\item 
			    $q$ is an explored path (i.e., not permanent but the subpath until its penultimate node is permanent),
			\item 
			    the cost vector $c(q)$ is not dominated by any cost vector in $c(\mathcal{P}_{sv})$,
			\item 
			    the cost vector $\bar{c}(q)$ is not dominated by any cost vector in $c(\mathcal{P}_{st})$, and
			\item 
			    $q$ is lex. minimal w.r.t.\ $\bar{c}$ among all explored $s$-$v$-paths.
		\end{enumerate}
	\end{definition}
	
	Other explored $s$-$v$-paths that are not queue paths, are stored in a list called \emph{nextQueuePath} ($\NQP$). 
	There is such a list for every arc in the graph.
	At any iteration of the \mda{}, for $(u,v) \in A$, $\NQP_{uv}$ stores already explored $s$-$v$-paths whose last arc is $(u,v)$ and that are, at the moment, not queue paths for $v$.
	The data structures of the \mda{} are initialized in the first lines:
	\begin{description}
		\item[Line \ref{algo:mda*:initStart}] 
    		The priority queue $Q$ contains at most one queue path for every node $v \in V$ during the algorithm's execution; it is initially empty.
    		The paths in $Q$ are lex. sorted w.r.t.\  $\bar{c}$.
		\item[Line \ref{algo:mda*:initNcl}] 
    		For any $(u,v) \in A$, the list $\NQP_{uv}$ stores explored $s$-$v$-paths whose last arc is $(u,v)$ and that are not the queue path for $v$.
    		These lists are initially empty.
		\item[Line \ref{algo:mda*:initPerm}] 
    		The lists of permanent $s$-$v$-paths $\mathcal{P}_{sv}$.
    		These lists are initially empty.
	\end{description}
	
	The last step before the main loop of the algorithm begins is to insert the empty path $p_{init}$ into $Q$.
	The path $p_{init}$ goes from $s$ to itself with zero costs.
	The main loop finishes when $Q$ is empty at the beginning of an iteration.
	Until then, every iteration performs three tasks: 
	\begin{enumerate}[i.]
		\item 
		    Extraction of a path $p$ from $Q$ that is lex. minimal w.r.t. $\bar{c}$. The path $p$ is immediately added to the list $\mathcal{P}_{sv}$ of permanent $s$-$v$-paths in Line \ref{algo:mda*:linePushBack2}.
		\item 
		    Assuming that $p$ is an $s$-$v$-path, the search for a new queue path for $v$ in \ref{algo:ncla*}.
		\item 
		    The propagation of $p$ along the outgoing arcs of $v$ in \ref{algo:propagate}.
	\end{enumerate}
	
	Even though \ref{algo:ncla*} appears prior to \ref{algo:propagate} in the pseudocode of the \mda, we describe \ref{algo:propagate} first, since \ref{algo:ncla*} returns an empty path (NULL) in the first iteration(s) of the \mda.
		
	\paragraph{Propagation}
	
	The extracted $s$-$v$-path $p$ is propagated along the arcs $(v,w) \in \outgoing{v}$.
	For such an arc $(v,w)$, we call the resulting path $p^{new}_w := p \circ (v,w)$.
	In case $\bar{c}(p^{new}_w)$ is dominated by $\beta_t$, by an element in $c(\mathcal{P}_{st})$, or by an element in $c(\mathcal{P}_{sw})$, it is immediately discarded.
	Otherwise, it is a new explored $s$-$w$ path and we have to distinguish whether it is a queue path for $w$ or whether it is stored in $\NQP_{vw}$. Three scenarios are possible:
	\begin{description}
		\item[Case 1:]
    	    There is an $s$-$w$-path $q \in Q$ with $\bar{c}(q) \lexsmaller \bar{c}(p^{new}_w)$. 
    	    Then, $p^{new}_w$ is appended to $\NQP_{vw}$.
    		In other words: $p^{new}_w$ is an explored path that can not yet be discarded or made permanent but it is currently not the lex. minimal, and hence most promising explored $s$-$w$-path.
		\item[Case 2:]
    		There is an $s$-$w$-path $q \in Q$ with $\bar{c}(p^{new}_w) \lexsmaller \bar{c}(q)$.
    		Then, $p^{new}_w$ is inserted via a \texttt{decreaseKey} operation into $Q$, where it replaces $q$.
    		Assume $(x, w) \in A$ is the last arc  of $q$. 
    		Subsequently, $q$ is prepended to the list $\NQP_{xw}$.
    		This means that even though $q$ is no longer in $Q$, it might re-enter it later.
		\item[Case 3:]
		    There is no $s$-$w$-path in $Q$.
		    Then, $p^{new}_w$ is always inserted into $Q$.
	\end{description}
	The procedure \ref{algo:propagate} returns the possibly updated priority queue $Q$.
	
	\begin{procedure}
		\small{
	\SetKwInOut{Input}{Input}
	\SetKwInOut{Output}{Output}
	\Input{$s$-$v$-path $p$, node $w \in \outgoing{v}$, priority queue $Q$, permanent paths $\mathcal{P}$, lists of promising explored paths $\NQP$, dominance bound $\beta_t$.}
	\Output{Updated priority queue $Q$.}
		\BlankLine
		$p_w^{new} \leftarrow p \circ (v,w)$\;
		\If{$c(\mathcal{P}_{st}) \dom \bar{c}(p_w^{new})$ \text{ or }$\beta_t \dom \bar{c}(p_w^{new})$}{\label{algo:propagate:pruneatfront}
		    \Return $Q$\;
		    }
		    \If{$c(\mathcal{P}_{sw}) \ndom c(p_w^{new})$ \label{algo:propagate:domCheck}}{
			    \uIf{\textbf{not} $Q.\mathtt{contains}(w)$}{
				    $Q.\mathtt{insert}\left( p_w^{new} \right)$\label{algo:propagate:insert}\;
			    }
			    \uElse{
			    	$q \leftarrow Q.\mathit{getPath}(w)$\;
			    	\If{$\bar{c}(p_w^{new}) \lexsmaller \bar{c}(q)$\label{algo:propagate:flexCheck}}{
				    	$Q.\mathtt{decreaseKey}(w, \, p_w^{new})$\label{algo:propagate:decrease}\;
				    	$(x,w) \leftarrow \mathit{lastArc}(q)$\;
				    	$\NQP_{xw}.\mathtt{prepend}(q)$\label{algo:propagate:pushFront}\tcp*[r]{Insert $q$ at the beginning of $\NQP_{xw}$.}
			    	}
		    		\Else{
		    			$\mathcal{NQP}_{vw}.\mathtt{append}(p_w^{new})$\label{algo:propagate:pushBack}\tcp*[r]{Insert $p_w^{new}$ at the end of $\NQP_{xw}$.}
	    			}
		    	}
		    }
		
		\Return $Q$\;
		\caption{propagate()}\label{algo:propagate}}
    \end{procedure}
	
	\paragraph{Next queue path.} 
	A key invariant in the \mda{} is that at any point in time during the execution, $Q$ contains at most one queue path for $v$ for any node $v \in V$.
	Thus, after $p$ is extracted from $Q$ in \Cref{algo:mda*:extraction} there is no $s$-$v$-path in $Q$.
	The procedure \ref{algo:ncla*} aims to find a new queue path for $v$.
	The search is performed among the paths in the lists $\NQP_{uv}$ of explored $s$-$v$-paths via $u$, $u \in \incoming{v}$. 
	These lists are sorted in lex. increasing order w.r.t. $\bar{c}$ (see \Cref{lem:orderedNcl}).
	Then, starting with the first path in each list, paths in $\NQP_{uv}$ are removed until the first path in the list fulfilling the queue path conditions (i)-(iii) is found or the list is emptied. 
	Thus, \ref{algo:ncla*} identifies at most one candidate $s$-$v$-path via each predecessor node of $u$.
	The path $p_v^{new}$ returned by the search is a lex. smallest one among these candidates (queue path condition (iv)), if there is any.
	Then, $p_v^{new}$ is a queue path for $v$ and is inserted into $Q$.
	
	\begin{procedure}
		\small{
        \SetKwInOut{Input}{Input}
        \SetKwInOut{Output}{Output}
        \Input{$s$-$v$-path $p^*$, lists of promising explored paths $\NQP$, node set $\incoming{v}$, permanent paths $\mathcal{P}$, dominance bound $\beta_t$.}
        \Output{New queue path for $v$, if one exists.}
        \BlankLine
        $p_v^{new} \leftarrow{} \Null$\tcp*[r]{Assume $\bar{c}(p_v^{new}) = [\infty]_{i=1}^d$.}
        \For{$u \in \incoming{v}$}{
            \For{$p_v \in \NQP_{uv}$}{
                \If{$c(\mathcal{P}_{st}) \dom \bar{c}(p_v)$ \text{ or }$\beta_t \dom \bar{c}(p_v)$}{\label{algo:ncla*:pruneatfront}
                    $\NQP_{uv}.\mathtt{remove}(p_v)$\label{algo:ncla*:delete}\;
                    \textbf{continue}\;
                    }
                    \If{\textbf{not} $c(\mathcal{P}_{sv}) \dom c(p_v)$ \textbf{and not} $c(p^*) \dom c(p_v)$\label{algo:ncla*:domCheck}}{
                    	\lIf{$\bar{c}(p_v) \lexsmaller \bar{c}(p_v^{new})$  \label{algo:ncla*:flex}}{
                            $p_v^{new} \leftarrow p_v$} \textbf{break}\label{algo:ncla*:break}\;
                    }
                
            }    					
        }
        \Return $p_v^{new}$;
        \caption{nextQueuePath()}\label{algo:ncla*}}
    \end{procedure}
	
	\subsection{Correctness}\label{sec:Correctness}
	
	The following remark is important for the proofs in this section. 
	\begin{remark}\label{rem:orderInvariant}
		For any $v \in V$, the lex. order of two $s$-$v$-paths $p$ and $q$ is the same w.r.t. $c$ and $\bar{c}$:
		$$ \bar{c}(p) \lexsmaller \bar{c}(q) \Leftrightarrow c(p) + \pi(v) \lexsmaller c(q) + \pi(v) \Leftrightarrow c(p) \lexsmaller c(q).$$
	\end{remark}

	The correctness of the \mda{} is proven in three steps. 
	We begin by showing that the elements of the lists $\mathcal{P}_{sv}$, $v\in V$ are sorted in lex. increasing order w.r.t. $\bar{c}$ and w.r.t. $c$ (Lemma~\ref{lem:ordered_extraction}).
	Then we proceed to prove that every path that is extracted from $Q$ is efficient (Lemma~\ref{lem:efficient_extract}). 
	Finally, the correctness follows from \Cref{thm:min_complete_set}.
	It states that the set $\mathcal{P}_{st}$ corresponds to a minimal complete set of efficient $s$-$t$-paths at the end of the algorithm.
	
	\begin{lemma}\label{lem:sortedQ}
		As long as the $\NQP$ lists of explored paths are sorted in lex. non-decreasing order w.r.t. $\bar{c}$, the \mda{} extracts paths from $Q$ in Line \ref{algo:mda*:extraction} in lex. non-decreasing order w.r.t. $\bar{c}$.
	\end{lemma}
	
	\begin{proof}
		Let $p$ be the path extracted from $Q$ in iteration $l \in \N$ of the \mda{}. Assume that in a prior iteration $k < l$, a path $q$ is extracted from $Q$ such that
		$$ \bar{c}(p) \lexsmaller \bar{c}(q).$$
		The \mda{} extracts a lex. minimal path from $Q$ in every iteration. 
		Thus, neither $p$ nor any subpath of $p$ can be in $Q$ in iteration $l$ since by \Cref{lem:increasing_potential} they are lex. smaller than $p$ w.r.t. $c$ and $\bar{c}$.
		However, there must be at least one node $j \in p$ such that an $s$-$j$-path $p'_j$ is in $Q$ in iteration $k$.
		Let $j \in V$ be the last such node along $p$.
		By \Cref{lem:increasing_potential} we know that the $s$-$j$-subpath $p_j$ of $p$ is lex. smaller than $p$ and we must have $\bar{c}(p_j) \lexsmallereq \bar{c}(p'_j)$. Otherwise, we would have
		$$  \bar{c}(p'_j) \lexsmaller \bar{c}(p_j) \lexsmaller \bar{c}(p) \lexsmaller \bar{c}(q). $$
		Since $p'_j$ and $q$ are in $Q$ in iteration $k$, this would contradict $q$'s extraction.
		
		Recall that $p_j$ is an explored $s$-$j$-path but it is not the queue path for $j$ in iteration $k$.
		We discuss the different cases that could lead to $p'_j$ being in $Q$ instead of $p_j$ and conclude that none of them can actually happen.
		\begin{description}
			\item[Case 1:]
    			The path $p_j$ is replaced by $p'_j$ in $Q$ in a \texttt{decreaseKey} operation in Line \ref{algo:propagate:decrease} of \ref{algo:propagate}. 
    			This can not happen since $p_j$ is lex. smaller than $p'_j$ and \texttt{decreaseKey} operations only happen if the lex. check in Line \ref{algo:propagate:flexCheck} succeeds.
			\item[Case 2:]
    			During the iteration wherein $p_j$ is explored for the first time, $p'_j$ is in $Q$.
    			In this situation, a \texttt{decreaseKey} operation is triggered in Line \ref{algo:propagate:decrease} of \ref{algo:propagate} to replace $p'_j$ with $p_j$ in $Q$.
    			After this happens, $p_j$ only leaves $Q$ if it is either extracted or if a newly explored $s$-$j$-path is lex. smaller (cf. Case 1.).
			Thus, $p'_j$ can not be extracted from $Q$ prior to $p_j$. 
			\item[Case 3:]
    			In some iteration prior to iteration $k$, an $s$-$j$-path $p''_j$ is extracted from $Q$ and a new queue path for $j$ is searched in \ref{algo:ncla*}.
    			Let $(i,j) \in A$ be the last arc of $p_j$ and $(i', j) \in A$ the last arc of $p'_j$. 
    			We can assume that $p_j$ and $p'_j$ are in the lists $\NQP_{ij}$ and $\NQP_{i'j}$ respectively and are considered as queue path candidates during the search.
    			Since $p'_j$ is in $Q$ during iteration $k$ and $p_j$ is extracted from $Q$ in an iteration before the $l$th, but after the $k$th (if not, $p$ would not be extracted in iteration $l$), both $s$-$j$-paths are neither dominated at $\mathcal{P}_{sj}$ w.r.t. $c$ nor at $\mathcal{P}_{st}$ w.r.t. $\bar{c}$ (Line \ref{algo:ncla*:pruneatfront} of \ref{algo:ncla*}). 
    			Otherwise, they would be permanently discarded in Line \ref{algo:ncla*:delete} of \ref{algo:ncla*}.
    			Recall that the new queue path $p_v^{new}$ for node $j$ computed by the search is lex. minimal among the queue path candidates in the $\NQP$ lists.
    			Given that the lex. ordering of the $\NQP$ lists is assumed in the statement from the lemma, it can not happen that $p_v^{new} = p'_j$ at the end of the search if $p_j$ is in $\NQP_{ij}$ as assumed.
		\end{description}
		We conclude that the existence of the explored path $p_j$ avoids $p'_j$ or any other $s$-$j$-path that is lex. greater than $p_j$  from being in $Q$ in iteration $k$.
		However, there must be an $s$-$j$-path in $Q$ during this iteration.
		This implies that $q$ would not be extracted from $Q$ since the $s$-$j$-path in $Q$ is lex. smaller.
	\end{proof}
	
	\begin{lemma}\label{lem:orderedNcl}
		If paths are extracted from $Q$ in Line \ref{algo:mda*:extraction} of the \mda{} in lex. non-decreasing order w.r.t. $\bar{c}$, the $\NQP$ lists of explored paths remain sorted in lex. non-decreasing order w.r.t. $\bar{c}$.
	\end{lemma}
	
	\begin{proof}
		Clearly, as long as the $\NQP$ lists contain at most one element, they are sorted.
		Let the current iteration be the first one wherein a second path is added to an $\NQP$ list, say $\NQP_{vw}$, $(v,w) \in A$.
		There are two ways in which this can occur.
		\begin{description}
		    \item[Case 1:]
		        A new $s$-$w$-path $q$ is prepended to $\NQP_{vw}$ in Line \ref{algo:propagate:pushFront} of \ref{algo:propagate}.
        		The path $q$ was in $Q$ until the current iteration and is now replaced by a newly explored path $p_w^{new}$ s.t. $\bar{c}(p_w^{new}) \lexsmaller \bar{c}(q)$ (Line \ref{algo:propagate:flexCheck} of \ref{algo:propagate}).
        		Let $q'$ be the $s$-$v$-path that is already in $\NQP_{vw}$. 
        		Since $q'$ is not lex. smaller than $q$, $q'$ is already in $\NQP_{vw}$ at the beginning of the current iteration.
        		This follows from the fourth queue path condition.
        		Thus, adding $q$ at the beginning of $\NQP_{vw}$ does not alter its ordering.
		    \item[Case 2:]
		        A newly explored $s$-$w$-path $p_w^{new}$ is appended to $\NQP_{vw}$ in Line \ref{algo:propagate:pushBack} of \ref{algo:propagate}.
        		By construction, $(v,w)$ is the last arc of $p_w^{new}$ and of the path $q'$, the $s$-$w$-path via $v$ that is already in $\NQP_{vw}$.
        		Note that the $s$-$v$-subpath $p$ of $p_w^{new}$ is extracted from $Q$ in the current iteration.
        		Since the $\NQP$ lists were sorted until the current iteration, \Cref{lem:sortedQ} guarantees that paths are extracted from $Q$ in lex. non decreasing order w.r.t. $\bar{c}$. 
        		As a consequence, the $s$-$v$-subpath $q'_v$ of $q'$ must have been extracted prior to $p$ and $\bar{c}(q'_v) \lexsmallereq \bar{c}(p)$. 
        		This implies $c(q'_v) \lexsmallereq c(p)$ and since both $s$-$v$-subpaths are expanded along the same arc $(v,w)$ to obtain $q'$ and $p_w^{new}$, we have $c(q') \lexsmallereq c(p_w^{new})$.
        		Thus, adding $p_w^{new}$ after $q'$ to the list $\NQP_{vw}$ does not alter the ordering of the list.
		\end{description}
		An induction over the number of elements in the $\NQP$ lists proves the lemma.
	\end{proof}
	
	The last two lemmas show that the ordering of the $\NQP$ lists and of the priority queue $Q$ influence each other.
	We can combine both statements to obtain the following key lemma for the correctness of the \mda. 
	
	\begin{lemma}\label{lem:ordered_extraction}
		Let $p$ be a path extracted in Line \ref{algo:mda*:extraction} of Algorithm \ref{algo:mda*}.
		Then, for all paths $q$ that were extracted from $Q$ prior to $p$ we have 
		$$ \bar{c}(q) \lexsmallereq \bar{c}(p). $$
	\end{lemma}
	
	\begin{proof}
		\Cref{lem:sortedQ} trivially holds until the first iteration wherein an $\NQP$ list gets a second path.
		Then, \Cref{lem:orderedNcl} guarantees that this path is inserted correctly into its $\NQP$ list, which then satisfies the condition of \Cref{lem:sortedQ} to guarantee that in the upcoming iterations, paths are still extracted in lex. non decreasing order w.r.t. $\bar{c}$.
		Repeating these arguments until the last iteration of the \mda{} proves the lemma.
	\end{proof}
	
	As a consequence of the last lemma, the permanent $s$-$v$-paths in $\mathcal{P}_{sv}$ are also sorted in lex. non-decreasing order w.r.t. $\bar{c}$ and w.r.t. $c$. To finally prove the correctness of Algorithm \ref{algo:mda*}, we still need to show that all permanent paths are efficient (Lemma \ref{lem:efficient_extract}) and that the list $\mathcal{P}_{st}$ corresponds to a minimal complete set of efficient $s$-$t$-paths at the end of the \mda{} (Theorem \ref{thm:min_complete_set}).
	
	\begin{lemma}\label{lem:efficient_extract}
		Consider a node $v \in V$. Any $s$-$v$-path $p$ extracted from $Q$ in Line \ref{algo:mda*:extraction} of the \mda{} is an efficient $s$-$v$-path. 
	\end{lemma}
	\begin{proof}
		Assume $p$ is the first path extracted from $Q$ that is not efficient, i.e., there exists an efficient $s$-$v$-path $p'$ that dominates $p$.
		Then $p'$ has to be lex. smaller than $p$.
		\begin{description}
			\item[Case 1:]
    			$p'$ is extracted from $Q$.
    			Following \Cref{lem:ordered_extraction} it is extracted prior to $p$.
    			After its extraction $p'$ is added to $\mathcal{P}_{sv}$ and there is no other $s$-$v$-path in $Q$.
    			For any new queue path candidate for $v$, the dominance checks in Line \ref{algo:propagate:domCheck} of \ref{algo:propagate} and Line \ref{algo:ncla*:domCheck} of \ref{algo:ncla*} make sure that it is not dominated by any path in $\mathcal{P}_{sv}$.
    			Thus, $p$ can not become a queue path for $v$ and consequently, it is never extracted from $Q$.
			\item[Case 2:]
    			$p'$ is not extracted from $Q$. 
    			Since $p'$ is an efficient $s$-$v$-path, all its subpaths are efficient. 
    			Hence, the only reason why $p'$ is not extracted from $Q$ is because the reduced costs $\bar{c}$ a subpath of $p'$ (it can be $p'$ itself) are dominated by the cost vector of a path in $\mathcal{P}_{st}$ s.t. the subpath is not added to $Q$.
    			However, since $p'$ dominates $p$, we have 
    			$$ \bar{c}(p') = c(p') + \pi_v \leq c(p) + \pi_v = \bar{c}(p),$$
    			which implies that $\bar{c}(p)$ is also dominated by the same cost vector in $\mathcal{P}_{st}$. 
    			This contradicts the insertion of $p$ into $Q$ and consequently its extraction.
		\end{description}
	\end{proof}
	
	\begin{theorem}\label{thm:min_complete_set}
		At the end of Algorithm \ref{algo:mda*}, the set $\mathcal{P}_{st}$ is a minimal complete set of efficient $s$-$t$-paths.
	\end{theorem}
	
	\begin{proof}
		By Lemma \ref{lem:efficient_extract} it is clear that $\mathcal{P}_{st}$ contains only efficient $s$-$t$-paths.
		Let $c_t$ be a non-dominated point. 
		
		First, we discard that $\mathcal{P}_{st}$ may contain two paths $p$ and $p'$ with $c(p) = c(p') = c_t$.
		W.l.o.g., assume that $p$ is added to $\mathcal{P}_{st}$ first.
		Then, whenever $p'$ is processed as a candidate to enter $Q$, the dominance checks in Line \ref{algo:propagate:domCheck} of \ref{algo:propagate} or in Line \ref{algo:ncla*:domCheck} of \ref{algo:ncla*} would fail since the dominance relation $\dom$ is irreflexive.
		Thus, $p'$ is not inserted into $\mathcal{P}_{st}$.
		
		Assume that at the end of the \mda{} there is no path $p \in \mathcal{P}_{st}$ such that $c(p) = c_t$.
		Since $c_t$ is a non-dominated point, there exists at least one $s$-$t$-path $p$ in $G$ with costs $c_t$. 
		Then, at the end of the \mda{}, there is exactly one arc $(i,j) \in p$ such that the $s$-$i$-subpath $p_{s i}$ of $p$ is in the corresponding set $\mathcal{P}_{si}$ and there is no $s$-$j$-path in $\mathcal{P}_{sj}$ with costs $c(p_{s j})$, where $p_{s j}$ is the $s$-$j$-subpath of $p$.
		The extreme case is $i = s$. 
		Since the heuristic $\pi$ is admissible, we have 
		$$\bar{c}(p_{sj}) = c(p_{sj}) + \pi_j = c(p_{si}) + c_{ij} + \pi_j \leq c_t.$$
		Thus, since $c_t$ is a non-dominated point, $c(\mathcal{P}_{st}) \dom (c(p_{si}) + c_{ij} + \pi_j)$ is not true.
		This implies that the pruning conditions checked in Line \ref{algo:propagate:pruneatfront} of \ref{algo:propagate} and in Line \ref{algo:ncla*:pruneatfront} of \ref{algo:ncla*} do not avoid that $p_{sj} = p_{si} \circ (i,j)$ is inserted into $Q$.
		Finally, since $p$ is an efficient $s$-$t$-path, $p_{sj}$ is also an efficient $s$-$j$-path by subpath optimality.
		As a consequence of Lemma \ref{lem:efficient_extract}, the dominance checks in Line \ref{algo:propagate:domCheck} of \ref{algo:propagate} and in Line \ref{algo:ncla*:domCheck} of \ref{algo:ncla*} do not prevent $p_{sj}$ from being processed unless there is a different $s$-$j$-path in $\mathcal{P}_{sj}$ with equivalent costs.
		In any case, this is a contradiction to the existence of a pair of nodes $i$ and $j$ as described earlier.	
	\end{proof}
	
	\subsection{Complexity of the searches for new queue paths}\label{sec:nclListsComplexity}
	Let $n$ and $m$ be the number of nodes and arcs in $G$, respectively.
	Additionally, we set $N$ to be the total number of extracted paths during the \mda{} and $N_{\max} := \max_{v \in V}{|\mathcal{P}_{sv}|}$ the maximum cardinality of a list of permanent paths at the end of the algorithm.
	
	Up until now, we could ignore the mandatory question of how paths should be stored efficiently.
	If $p$ is an $s$-$v$-path and $(u,v)$ its last arc, $p$ can be uniquely encoded as a tuple that stores $p$'s end node $v$ and (a reference to) the tuple encoding the $s$-$u$-subpath of $p$.
	The possibility to store paths using these tuples, often called \emph{labels}, allows us to assume that a path's memory consumption in the \mda{} algorithm is $\mathcal{O}(1)$.
	Given a path's label, the path itself can be reconstructed recursively after the algorithm terminates.
	Since the \mda{} algorithm only outputs simple paths, this can be achieved in $\mathcal{O}(n)$ time per path.
	
	In the original exposition of the MDA, there are no $\NQP$ lists for arcs $(v,w) \in A$.
	Thus, in the \textit{propagate} routine of the MDA, a new explored $s$-$w$-path $p_w^{new}$ is ignored if there is a lex. smaller queue path for $w$ in $Q$.
	But $p_w^{new}$ is not discarded forever: the next time an $s$-$w$-path is extracted from $Q$, a subroutine called \textit{nextCandidatePath} iterates over the permanent paths of the predecessors of $w$ and expands these paths along the incoming arcs of $w$.
	During these expansions the path $p_w^{new}$ is rebuilt and reconsidered as a candidate queue path for $w$.
	
	The MDA works correctly as described in the last paragraph.
	To enhance its performance, an index $\lpl_{uv} \in [0, |\mathcal{P}_{su}|)$ for every arc $(u,v) \in A$ tells the \textit{nextCandidatePath} searches at which path in the list $\mathcal{P}_{sv}$ of permanent $s$-$v$-paths it shall start generating new candidate queue paths for $w$.
	Thus, these indices correspond to the beginning of the $\NQP_{uv}$ lists in the \mda{}.
	
	Both methods are different regarding the resulting space complexity bound.
	The storage of $\lpl_{vw}$ indices uses $\bigoh{m}$ space.
	As already mentioned, the $\NQP$ lists avoid having to repeat the extension of explored paths from their permanent predecessor paths.
	Thus, for every node $v \in V$, we do not only store its permanent paths in the list $\mathcal{P}_{sv}$: for each outgoing arc $(v,w)$ of $v$, the $\NQP_{vw}$ list stores the expansions of the permanent $s$-$v$-paths towards $w$ that are still not dominated.
	In the worst case, these are the expansions of all paths in $\mathcal{P}_{sv}$.
	Thus, if $N_{\max}$ is the maximum cardinality of a list of permanent paths, the $\NQP$ lists require an additional storage space of $\bigoh{m N_{\max}}$.
	
	Using the $\NQP$ lists has no impact on the running time bounds for the original MDA or the \mda{}. In the original MDA exposition the explored paths need to be expanded from the permanent predecessor paths in every call to \textit{nextCandidatePath}.
	The required sums to build the new cost vectors can be neglected when analyzing the running time.
	The $\NQP$ lists need to remain sorted, but as shown in \Cref{lem:orderedNcl}, their ordering can be maintained by always prepending or appending new paths to the lists.
	Thus, there is no extra computational effort required to sort the lists.
	
	In contrast to the One-To-All scenario, it is not possible to design an algorithm for the One-to-One MOSP Problem that runs in polynomial time w.r.t. to the input size and the output size of any given instance (cf. \cite{Boekler18}).
	Still, it is worth noting that the running time bounds derived in \cite{Sedeno19} and \cite{Casas2021a} for the MDA are still valid for the new \mda{} despite of the increased memory consumption caused by the $\NQP$ lists.
	In the biobjective case ($d=2$), the \mda{} runs in $\mathcal{O}(N\log n + N_{\max}m)$ (see \Cref{sec:bidimensional}). For the multiobjective case ($d>2$), the running time bound is $\mathcal{O}(d N_{\max} (n \log n + N_{\max}m))$.
	
	\subsection{Obtaining Heuristics and Dominance Bounds}\label{sec:computeHeuristics}
	
	The \mda{} requires a heuristic $\pi$ and a dominance bound $\beta_t$ as part of its input. In this section we show how to obtain both for a $d$-dimensional One-to-One MOSP instance $(G,s,t,c)$.
	
	For a set $\mathcal{X} \subseteq{\R^d}$ let $\mathcal{X}^*$ be the set of non-dominated vectors in $\mathcal{X}$. Then, the \emph{ideal point} of $\mathcal{X}$
	$$ x^* := \left(\min_{x \in \mathcal{X}}x_1, \dots, \min_{x \in \mathcal{X}}x_d\right).$$ 
	is a lower bound on the values of the non-dominated vectors in $\mathcal{X}^*$ (cf. \cite{Ehrgott2005}). For ease of exposition, we introduce the following concept.
	
	\begin{definition}[Set of Preprocessing Orders]
		A set $\Sigma \subset \N^d$ is called a \emph{set of preprocessing orders} if
		\begin{enumerate}
			\item $|\Sigma| = d$
			\item All elements in $\Sigma$ are permutations of $(1, \dots, d)$.
			\item For every $i \in \{1, \dots, d\}$ there is exactly one permutation vector $\sigma \in \Sigma$ s.t. $\sigma_1 = i$.
		\end{enumerate}
	\end{definition}
	
	For any $v \in V$, we denote the set of efficient $v$-$t$-paths in $G$ by $\mathcal{P}_{vt}$.
	Moreover, given a permutation $\sigma$ of $(1, \dots, d)$, we call a \emph{$t$-to-all $\sigma$-lexicographic Dijkstra query} a run of the classical one-to-all Dijkstra Algorithm \cite{Dijkstra59} on the reversed digraph of $G$ rooted at $t$ if it processes the explored paths in lex. non decreasing order w.r.t. $\sigma$.
	The query finds a reversed shortest path tree $T_{\sigma}$.
	We denote the unique (reversed) $v$-$t$-path in $T_{\sigma}$ by $T_{\sigma}(v)$ for every $v \in V$.
	$T_{\sigma}(v)$ is a shortest $v$-$t$-path w.r.t. the $\sigma_1$th cost component of the arc cost vectors $c_{a}$, $a \in A$. 
	The following statement is clear:
	
	\begin{theorem}
		The ideal points $c^*_{v,t}$ of all sets $c(\mathcal{P}_{vt})$ can be computed by running a $\sigma$-lexicographic $t$-to-all Dijkstra query for every $\sigma \in \Sigma$. 
	\end{theorem}
	
	
	\begin{example}
		\Cref{tab:ideal} shows how the ideal point for $c(\mathcal{P}_{vt})$ is built after running $d=3$ lexicographic Dijkstra queries, one for each element of a set of preprocessing orders $\Sigma$.
		\begin{table}
			\centering
			\caption{Computing an ideal point out of the results of the $\sigma$-lexicographic $t$-to-all Dijkstra queries dictated by a set of preprocessing orders $\Sigma$.}\label{tab:ideal}
			\begin{tabular}{l l}
				$\Sigma$ & $c(T_{\sigma}(v))$ \\
				\midrule
				$\sigma = (\textcolor{red}{1}, 2, 3)$ & $(\textcolor{red}{4}, 7, 3)$ \\
				$\sigma = (\textcolor{red}{2}, 1, 3)$ & $(8, \textcolor{red}{2}, 4)$ \\
				$\sigma = (\textcolor{red}{3}, 1, 2)$ & $(5, 3, \textcolor{red}{2})$ \\
				\midrule
				\multicolumn{2}{c}{$c^*(v,t) = (4, 2, 4)$}
			\end{tabular}
		\end{table}
	\end{example}
	
	The following statement makes clear why we need the ideal points of the sets $c(\mathcal{P}_{vt})$.
	
	\begin{proposition}
		The node potential $\pi(v) := c^*_{v,t}$, $v\in V$, is a heuristic.
	\end{proposition}
	
	
	Furthermore, we can use the results of the lexicographic Dijkstra queries to construct a dominance bound. The following theorem holds for:
	$$ \beta_t := \left( \max_{\sigma \in \Sigma} c_1(T_{\sigma}(s)) + \varepsilon, \dots, \max_{\sigma \in \Sigma} c_d(T_{\sigma}(s)) + \varepsilon \right)$$
	and $\varepsilon > 0$.

	\begin{theorem}\label{thm:pruningUb}
		$\beta_t$ is a dominance bound, i.e., there is no efficient $s$-$t$-path $p$ that is dominated by $\beta_t$.
	\end{theorem}

	\begin{proof}
		Assume the statement is false and let $p$ be an efficient $s$-$t$-path dominated by $\beta_t$. Then, for any $i \in \{1,\dots,d\}$
		$$c_i(p) > \beta_{t,i} = \max_{\sigma \in \Sigma} c_i(T_{\sigma}(s))$$
		and, for fixed $\omega \in \Sigma$,
		$$\max_{\sigma \in \Sigma} c_i(T_{\sigma}(s)) > c_i(T_{\omega}(s)).$$
		The construction of $\beta_t$ does not imply that there is an $s$-$t$-path in $G$ with costs $\beta_t$. However, $T_{\omega}(s)$ is an $s$-$t$-path in $G$ and
		$$c_i(p) > c_i(T_{\omega}(s))$$
		implies that $p$ is a dominated path, which contradicts the assumption.
	\end{proof}
	
	The following example explains why we need to add $\varepsilon$ to every dimension of $\beta_t$ to obtain a feasible bound.
	
	\begin{example}
		Suppose the graph $G$ only contains the nodes $s$ and $t$ and an arc connecting both.
		Regardless of the chosen set of preprocessing orders $\Sigma$, we have $\max_{\sigma \in \Sigma} c(T_{\sigma}(s)) = c_{(s,t)}$ and $c_{(s,t)}$ is the cost of the only efficient $s$-$t$-path in $G$. Hence, 
			$$ \left( \max_{\sigma \in \Sigma} c_1(T_{\sigma}(s)), \dots, \max_{\sigma \in \Sigma} c_d(T_{\sigma}(s))\right) \dom c_{(s,t)}$$
		holds and would invalidate \Cref{thm:pruningUb}.
	\end{example}
	
	All in all, $d$ lexicographic Dijkstra queries are enough to compute a heuristic and a dominance bound for any One-to-One MOSP instance.
	The heuristic is used to influence the ordering of the paths in $Q$ to guide the search towards the target node.
	Moreover, $\pi$ and $\beta_t$ are used in Line \ref{algo:propagate:pruneatfront} of \ref{algo:propagate} and Line \ref{algo:ncla*:pruneatfront} of \ref{algo:ncla*} to apply pruning techniques that discard provably irrelevant paths before they are further expanded.
	For a path $p$, the dominance check $c(\mathcal{P}_{st}) \dom \bar{c}(p)$ is called \emph{pruning by target efficient set} and the check $\beta_t \dom \bar{c}(p)$ is called \emph{pruning by dominance bound}.
	
	\section{Enhancements for the biobjective case}\label{sec:bidimensional}
	The \bda{} is the biobjective version of the \mda. Even though the \mda{} solves any Biobjective Shortest Path (BOSP) instance, multiple enhancements that we discuss in this section give rise to a more efficient biobjective algorithm with respect to both theory and practice. The pseudocode of the \bda{} is in \Cref{apendix:pseudocode}.

	To build a heuristic and an upper-bound as described in \Cref{sec:Heuristics and Path Pruning}, a set of preprocessing orders $\Sigma$ is needed.
	The only possible set of preprocessing orders in the biobjective case is $\Sigma = \{(1,2), \, (2,1)\}$.
	Then, for any $v \in V$, the two $\sigma$-lexicographic $t$-to-all Dijkstra queries compute $v$-$t$-paths $T_{(1,2)}(v)$ and $T_{(2,1)}(v)$ with costs
	\begin{equation}\label{eq:shortcutCosts}
		c(T_{(1,2)}(v)) = (\pi_1(v), \beta_{v,2}) \text{ and } c(T_{(2,1)}(v)) = (\beta_{v,1}, \pi_2(v)).
	\end{equation}
	In this context, the resulting points $\pi(v)$ and $\beta_v$ are the ideal point and the \emph{nadir point} (cf. \cite{Ehrgott2005}) of $c(\mathcal{P}_{vt})$, respectively.
	
	As noted in \Cref{sec:nclListsComplexity}, the time complexity bound of the BDA, the biobjective version of the MDA is better than the bound for the general multiobjective case.
	The reason for this is the so called \emph{dimensionality reduction} (cf. \cite{Pulido15}) that allows dominance checks in constant time.
	
	\begin{theorem}\label{prop:dimensionality}
		For a node $v \in V$ let the lists $\mathcal{P}_{sv}$ of permanent $s$-$v$-paths be sorted in lex. increasing order.
		Consider an $s$-$v$-path $p$ that is lex. greater than every path in $\mathcal{P}_{sv}$.
		Then, $c(\mathcal{P}_{sv}) \dom c(p)$ can be checked in constant time in the biobjective case. 
	\end{theorem}
	\begin{proof}
		Let $q$ be the last path in $\mathcal{P}_{sv}$ and recall that $\mathcal{P}_{sv}$ contains only efficient $s$-$v$-paths.
		Since $p$ is lex. greater than $q$, we have $c(q) \dom c(p)$ if and only if $c_2(q) \leq c_2(p)$.
		Moreover, if $q$ does not dominate $p$, then no other path $q' \in \mathcal{P}_{sv}$ does, since $q'$ is lex. smaller than $q$ and thus $c_2(q') \geq c_2(q)$. 
	\end{proof}
	
	\Cref{prop:dimensionality} is ubiquitous in the remainder of this section.
	In the \bda{} pseudocode and its implementation we only set $\gamma_{sv}^* = c_2(p)$ every time an $s$-$v$-path $p$ is extracted from $Q$ (Line \ref{algo:bda*:permanent}) rather than updating $\mathcal{P}_{sv}$ as in Line \ref{algo:mda*:linePushBack2} of the \mda.
	Thus, $\gamma_{sv}^*$ is a scalar set to the second cost component of the lex. greatest $s$-$v$-path extracted so far. 
	The second cost component of new explored $s$-$v$-paths is then compared with $\gamma_{sv}^*$ to decide in constant time whether they are dominated by any efficient $s$-$v$-path that was extracted before during the \bda.

	Moreover, in order to be able to store all efficient paths that are found in the biggest instances used in the biobjective experiments from \Cref{sec:experiments}, we decided to design the \bda{} without using the $\NQP$ lists described in \Cref{sec:Algorithm}. 
	Instead, the original version of the \texttt{nextCandidatePath} routine (cf. \cite{Sedeno19, Casas2021a}) described in \Cref{sec:nclListsComplexity} is used since it is less memory consuming. 
	Consider a node $v \in V$ and an extracted $s$-$v$-path $p$. 
	In the \mda, the extensions of $p$ along the outgoing arcs of $v$ are only added to the $\NQP$ lists in \ref{algo:propagate} if they are not dominated. 
	The \bda{} mirrors this behavior by only adding an extracted path $p$ to $\mathcal{P}_{sv}$ if its expansion in \ref{algo:propagate2d} is successful along at least one outgoing arc of $v$. 
	If all expansions are dominated, then it is guaranteed that $p$ is not a subpath of any new queue path for a neighbor of $v$ and is no longer needed in the \bda{}.
	
	\begin{remark}
		A naive pseudocode and implementation of the \bda{} does not require the scalars $\gamma_{sv}^*$.
		Every extracted $s$-$v$-path can be added to $\mathcal{P}_{sv}$ directly after its extraction like in the \mda{} and dominance checks can be assumed to implicitly access the second cost component of the last path in $\mathcal{P}_{sv}$ only.
		However, we found that in practice, storing extracted paths that are not successfully propagated has a negative impact on the running time of \ref{algo:ncl2d}. Thus, we decided to present the \bda{} using the additional scalars $\gamma_{sv}^*$ and make extracted paths permanent in \Cref{algo:bda*:permanentForNcl} of the \bda{} only after a successful propagation. Then, the \ref{algo:ncl2d} searches do less iterations that only expand unpromising paths.
	\end{remark}
	
	For any node $v \in V$, every $s$-$v$-path $p$ extracted from $Q$ during the \bda{} is guaranteed to be an efficient $s$-$v$-path. However, $\bar{c}(p)$ could be dominated at $c(\mathcal{P}_{st})$. 
	In this case, it is not necessary to propagate $p$ along the outgoing arcs of $p$. 
	Including an extra dominance check after \ref{algo:ncla*} in the \mda{} did not improve the algorithm's performance. 
	However, since dominance is checked in constant time in the biobjective case, Line \ref{algo:bda*:dominanceAfterExtract} in the \bda{} aborts an iteration if the reduced costs of an extracted path are dominated at $c(\mathcal{P}_{st})$. 
	
	\subsection{Shortcuts to Target}\label{sec:shortcuts}
	The larger $\mathcal{P}_{st}$ becomes, the more likely it is that a newly explored path $p$ is discarded using the pruning by target efficient set.
	Thus, finding efficient $s$-$t$-paths in early iterations is crucial to discard irrelevant paths early.
	
	The following \emph{shortcut techniques} are discussed in \cite{Cabrera20} and in \cite{Ahmadi21}. 
	In \cite{Cabrera20} the Constrained Shortest Path problem is solved using the \emph{Pulse} Algorithm. In \cite{Ahmadi21} a label-setting algorithm for BOSP problems that is based on the classical algorithm by Martins \cite{Martins84} is presented.
	We show how shortcuts can be included in the \bda{} and explain the theoretical drawbacks of this technique in a general multiobjective scenario.
	
	For any $v \in V$, let $p$ be an $s$-$v$-path extracted from $Q$. $p$ can be concatenated with the path $T_{(1,2)}(v)$ computed during the $(1,2)$-lexicographic $t$-to-all Dijkstra query in the preprocessing stage to obtain an $s$-$t$-path $p_{st}$.
	Using \eqref{eq:shortcutCosts}, we have
	\begin{equation}\label{eq:costs}
		c(p_{st}) = \left(c_1(p) + \pi_1(v), \  c_2(p) + \beta_{v,2}\right) = \left(\bar{c}_1(p), \  c_2(p) + \beta_{v,2}\right).
	\end{equation}
	
	If $p_{st}$ is not dominated by any path in $\mathcal{P}_{st}$, it can be added to $\mathcal{P}_{st}$ (see Lines \ref{algo:bda*:initShortcut}-\ref{algo:bda*:endShortcut} of \Cref{algo:bda*}).
	
	\begin{remark}[Overtaking via shortcut]\label{rem:overtaking}
		\Cref{lem:ordered_extraction} guarantees that paths are extracted from $Q$ in lex. non decreasing order w.r.t. $\bar{c}$.
		The first cost component of the path $p_{st}$ is consistent with this ordering but the second cost component is not since we add $\beta_{v,2}$ rather than $\pi_2(v)$ to $c_2(p)$ (see \eqref{eq:costs}).
		Thus, an explored path $q$ with $c_1(p_{st}) = c_1(q)$ and $c_2(p_{st}) > c_2(q)$ could exist when $p_{st}$ is made permanent.
		Concerning the order in which paths are made permanent dictated by \Cref{lem:ordered_extraction}, $p_{st}$ has overtaken $q$.
	\end{remark}
	
	\begin{example}[Overtaken $s$-$t$-paths can be dominated]\label{ex:shortcutError}
	\begin{figure}
		\centering
		\begin{tikzpicture}[scale=1.25]
		\tikzstyle{arc}=[->,thick]
		\tikzstyle{nodo}=[]
		\node (s) at (-2,0) [nodo] {$s$};
		\node (v) at (2.5,1.5) [nodo] {$v$};
		\node (w) at (2.5,-1.5) [nodo] {$w$};
		\node (t) at (0.5,0) [nodo] {$t$};
		
		\draw (s) -- (v) [arc] node [midway,fill=white] {\footnotesize{$(1,1)$}};
		\draw (s) -- (t) [arc] node [midway,fill=white] {\footnotesize{$(1,10)$}};
		\draw (s) -- (w) [arc] node [midway,fill=white] {\footnotesize{$(2,2)$}};
		\draw (v) -- (t) [arc] node [midway,fill=white] {\footnotesize{$(2,4)$}};
		\draw (w) -- (t) [arc] node [midway,fill=white] {\footnotesize{$(1,2)$}};
		\draw (v) -- (w) [arc] node [midway,fill=white] {\footnotesize{$(2,0)$}};
		\end{tikzpicture}
		\caption{Biobjective Shortest Path Instance.}\label{fig:shortcutError}
	\end{figure}
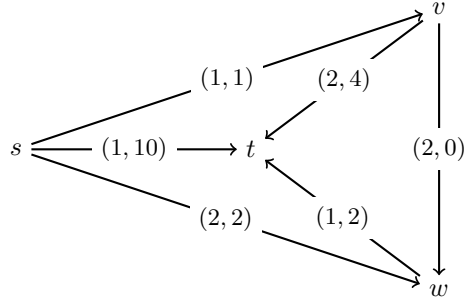
	Consider the BOSP instance in \Cref{fig:shortcutError}.
	We have $\pi(v) = (2, 2)$, $\beta_{v} = (3, 4)$, and $\pi(w) = (1,2)$. 
	The path $q = (s, w)$ has costs $c(q) = (2,2)$ and $\bar{c}(q) = (3,4)$. 
	The path $p = (s,v)$ has costs $c(p) = (1,1)$ and $\bar{c}(p) = (3,3)$. $p$ and $q$ are added to $Q$ during the first iteration of the \bda.
	Since $\bar{c}(p) \lexsmaller \bar{c}(q)$, $p$ is extracted prior to $q$. 
	The corresponding shortcut $p_{st} = p \circ T_{(1,2)}(v)$ has costs $(3,5)$ and is added to $\mathcal{P}_{st}$ since the only path therein has costs $(1,10)$.
	In the iteration after $p$'s extraction, the path $q$ is extracted.
	The path $q \circ T_{(1,2)}(w)$ has costs $(3, 4)$ and dominates $p_{st}$.
	Thus, $q \circ T_{(1,2)}(w)$ has to replace $p_{st}$ in $\mathcal{P}_{st}$.
	\end{example}
	
	Correcting the erroneous addition of shortcut paths to $\mathcal{P}_{st}$ as described in \Cref{ex:shortcutError} is done in constant time in the biobjective scenario.
	A shortcut path $p_{st} = p \circ T_{(1,2)(v)}$ is only added to $P_{st}$ if it is not dominated by any path therein.
	Thus, the addition of $p_{st}$ can only be unwanted if in a later iteration an $s$-$t$-path that dominates $p_{st}$ is found.
	Since $p_{st}$'s first cost component is consistent with the order in which the \bda{} extracts explored paths from $Q$, $p_{st}$ can only be dominated by an $s$-$t$-path $q$ with $c_1(q) = c_1(p_{st})$ and $c_2(q) < c_2(p_{st})$. 
	Thus, before adding an $s$-$t$-path $q$ to $\mathcal{P}_{st}$, the \bda{} checks if $q$ dominates the last path in $\mathcal{P}_{st}$.
	If so, the last path in $\mathcal{P}_{st}$ is replaced by $q$. Otherwise, $q$ is appended to $\mathcal{P}_{st}$ (see Lines \ref{algo:bda*:initShortcut}-\ref{algo:bda*:endShortcut} of \Cref{algo:bda*}). 
	Note that since $T_{(1,2)}(t) = (t)$ and $c(T_{(1,2)}(t)) = (0,0)$, an $s$-$t$-path extracted from $Q$ in the \bda{} can be concatenated with $T_{(1,2)}(t)$ without altering the path and can thus be treated as a shortcut path to make it permanent in Lines \ref{algo:bda*:initShortcut}-\ref{algo:bda*:endShortcut}.
	
	Further iterations of the \bda{} can be avoided if the first cost components of the ideal point $\pi(v)$ and the nadir point $\beta_{v}$ of $c(\mathcal{P}_{vt})$ coincide.
	In this case, if $p \circ T_{(1,2)}(v)$ is an efficient $s$-$t$-path, then it is guaranteed that no further efficient $s$-$t$-path can be built via $v$. 
	Hence, $p$ does not need to be further expanded along the outgoing arcs of $v$ (cf. Line \ref{algo:bda*:abortIteration} of \Cref{algo:bda*}).
 
 	\begin{remark}[Shortcuts in a multiobjective scenario]
 		In a general multiobjective scenario, the wrong addition of overtaking paths to $\mathcal{P}_{st}$ comes at a higher costs.
 		In the worst case, the newly extracted path $q'$ has to be compared to all paths $q \in \mathcal{P}_{st}$ such that $c_1(q) = c_1(q')$. 
 		There is no known a priori bound for the number of such paths in $\mathcal{P}_{st}$ and thus, the complexity of the deletion of dominated paths in $\mathcal{P}_{st}$ is in $\mathcal{O}(dN_{\max})$.
 	\end{remark}
	
	\subsection{Bidirectional Search in Parallel}
	
	For a BOSP instance $(G,s,t,c)$, the instance introduced in the following definition swaps the roles of $s$ and $t$ and the roles of the first and second arc cost components.
	
	\begin{definition}[Reversed BOSP instance]
		Consider a BOSP instance $\mathcal{I} = (G, s, t, c)$. Let $\overleftarrow{G}$ be the reversed digraph of $G$ and for any arc $a \in A$ let $\overleftarrow{c}_a := (c_{a,2}, c_{a,1})$ be the cost vector of $a$ with swapped cost components. 
		The BOSP instance $\overleftarrow{\mathcal{I}} := (\overleftarrow{G}, t, s, \overleftarrow{c})$ is called the \emph{reversed BOSP instance} of $\mathcal{I}$.
	\end{definition}

	Recall that the output $\mathcal{P}_{st}$ of the \bda{} is a minimal complete set of efficient $s$-$t$-paths and that the paths therein are sorted in lex. increasing order. The following statement is clear.

	\begin{proposition}\label{prop:reverseSearch}
		The lists $c(\mathcal{P}_{st})$ and $c(\mathcal{P}_{ts})$ put out by the \bda{} after solving $\mathcal{I}$ and $\overleftarrow{\mathcal{I}}$, respectively, contain the same cost vectors but in reversed order.
	\end{proposition}
	
	
	Similar to how it is done in \cite{Ahmadi21}, we propose a bidirectional framework to solve a BOSP instance $\mathcal{I}$. It is 
	This framework is run in parallel: two \bda{} searches run simultaneously to partially solve $\mathcal{I}$ and $\overleftarrow{\mathcal{I}}$.
	The \bda{} solving $\mathcal{I}$ is called the \emph{forward search} and the \bda{} solving $\overleftarrow{\mathcal{I}}$ is called the \emph{backward search}.
	To better understand this method, consider the paths' costs vectors of both $\mathcal{I}$ as well as $\overleftarrow{\mathcal{I}}$ in a single $c_1$-$c_2$-cartesian coordinate system.
    Then, the forward search finds efficient $s$-$t$-paths from the upper left to the lower right corner and the backward search finds efficient reversed $s$-$t$-paths from the lower right to the upper left corner.
	W.l.o.g., suppose that the \bda{} solving $\mathcal{I}$ makes an $s$-$t$-path $p$ with costs $c(p)$ permanent. 
	If the \bda{} solving $\overleftarrow{\mathcal{I}}$ already found a $t$-$s$-path with costs $\overleftarrow{c}(p)$, the whole search can stop since the union of the sets $\mathcal{P}_{st}$ and $\mathcal{P}_{ts}$ now builds a minimal complete set of efficient $s$-$t$-paths.
	The resulting parallelized bidirectional BOSP algorithm is called \emph{Bidirectional-Targeted-BDA} (BT-BDA).
	
	Besides the described stopping condition, both searches in the BT-BDA do not require further modifications other than two more lexicographic Dijkstra queries in the preprocessing phase. Moreover, the forward search and the backward search need to share $\beta_t$, $\pi$, and $\overleftarrow{\pi}$ to enhance their performance.
	
	\paragraph{$s$-to-all Dijkstra queries during preprocessing}
	Since the roles of $s$ and $t$ are exchanged and the ordering of the cost components in $\mathcal{I}$ and $\overleftarrow{\mathcal{I}}$ is reversed, two $s$-to-all lexicographic Dijkstra queries have to be performed in the preprocessing phase to determine the heuristic $\overleftarrow{\pi}$ as well as the shortcuts for the backward search.
	The dominance bound computed in these queries coincides with $\beta_t$.
	
	\paragraph{Sharing bounds between searches}
	The forward and the backward searches in the BT-BDA need to share information to check the stopping condition.
	This communication can also be beneficial when it comes to updating bounds.
	To check the stopping condition, both searches need to access $\beta_t$ which at the beginning of the algorithm is the nadir point of $c(\mathcal{P}_{st})$ and $c(\mathcal{P}_{ts})$.
	Then, every time the forward search finds a new $s$-$t$-path, it updates $\beta_{t,2}$ and every time the backward search finds a new $t$-$s$-path, it updates $\beta_{t,1}$.
	In Line \ref{algo:bda*:breakCondition} of the \bda{}, each search compares the reduced costs $\bar{c}$ of the extracted path with the updated bound from the opposite search to decide whether the whole bidirectional search can be stopped. 
	If the \bda{} did not use shortcut paths as described in \Cref{sec:shortcuts}, every node $v \in V$ reached by the search would be reached for the first time along a lex. smallest $s$-$v$-path. 
	However, using shortcut paths, the first extracted $s$-$v$-path $p$ is possibly not lex. minimal. 
	In \Cref{algo:bda*:updateHeuristic} of the \bda{}, $\overleftarrow{\pi}_1(v)$ is updated to $c_1(p)$.
	Recall that the heuristic $\overleftarrow{\pi}(v)$ for the reversed search stores the costs of shortest $s$-$v$-paths w.r.t. both cost components. 
	Moreover, since the reversed search processes explored paths in lex. order w.r.t. $\overleftarrow{c}$, newly explored $t$-$v$-paths $q$ are dominated if
	\begin{align*}
		\beta_{t,1} &\leq \bar{\overleftarrow{c}}_1(q)\\
		\Leftrightarrow \beta_{t,1} &\leq \overleftarrow{c}_1(q) + \overleftarrow{\pi}_1(v).
	\end{align*}
 	Thus, increasing the value of $\overleftarrow{\pi}_1(v)$ in the forward search if $v$ is reached for the first time by a path that is not lex. minimal, increases the likelihood for newly explored $t$-$v$-path to be discarded immediately because they are dominated.
 	This argument is symmetric concerning the updates in the backward search of the $\pi_2(v)$ values for the forward search.
 	
	\section{Experiments}\label{sec:experiments}
	
	In \Cref{sec:mdaComparison}, we compare the \mda{} with the Multiobjective Dijkstra Algorithm (MDA) introduced in \cite{Casas2021a} on instances with $3$ dimensional cost vectors.
	The MDA is originally a one-to-all MOSP algorithm.
	Using the pruning by target efficient set and the pruning by dominance bound, it is adapted to perform well in the One-to-One scenario.
	Both algorithms are implemented using the $\NQP$ lists introduced in this paper to manage explored paths that are not queue paths.
	It turns out that this implementation is faster in practice and has negligible memory overhead.
	
	In \Cref{sec:biobjExperiments}, we report on the results of our biobjective experiments. 
	It appears clear that the \boaEnh algorithm first presented in \cite{Ahmadi21} as well as its parallelized bidirectional version, the \boaBidi{} algorithm are state-of-the-art BOSP algorithms as they outperformed the BOA* algorithm from \cite{Ulloa20} and the BDA algorithm from \cite{Sedeno19}. 
	Thus, we compare the new \bda{} and the \bda{} with the \boaEnh{} algorithm and the \boaBidi{} algorithm, respectively.
	We run the experiments using two different types of priority queues: binary heaps and so called buckets \cite{Denardo79}.
	The latter can only be used for integer costs and in theory, an exponential number of buckets can be non-empty at the same time. 
	However, they perform better in our experiments than ordinary binary heaps.
	
	In all experiments, we focus on the comparison of algorithmic ideas rather than on the comparison of different heuristics or different priority queues.
	Thus, whenever two algorithms are compared against each other, they use the same heuristic and the same type of priority queue.
	Moreover, we always benchmark the time needed to find a minimal complete set only, i.e., no preprocessing times are benchmarked in our experiments.
	This is reasonable since for interesting instances that are not solved within a few seconds, the running time of the lexicographic Dijkstra queries that are run in the preprocessing stage can be neglected.
	
	All algorithms are implemented in C++. We use the publicly available implementation of the \boaEnh{} and \boaBidi{} algorithms from \cite{Ahmadi21} for the BOSP experiments.
	We used the gcc $7.5.0$ compiler with the optimization flag O3 to build the binaries.
	The MOSP experiments were performed on a computer with an Intel Xeon Gold $5122$ \@ $3.60$GHz processor and $96$ GB memory. 
	The time limit for MOSP instances was set to $5$h.
	The BOSP experiments were performed on a computer with an Intel Xeon CPU $E5-2660$ v$3$ \@ $2.60$GHz processor and $128$ GB memory.
	The time limit for BOSP instances was set to $2$h.
	
	\subsection{Instance Description}
	
	\paragraph{NetMaker Instances $d = 3$} 
	Netmaker graphs are synthetic graphs with $5000$ to $30000$ nodes and $29591$ to $688398$ arcs.
	These graphs have been used in multiple publications \cite{Skriver00, Raith09, Raith18, Casas2021a}.
	In every such graph, all nodes are connected via a Hamiltonian cycle to ensure connectivity.
	Then, arcs between the nodes along the cycle are added randomly. The arc costs are $3$ dimensional and each cost component lies between $1$ and $1000$. 
	For any arc there is a cost component between $1$ and $333$, a cost component between $334$ and $666$, and a cost component between $667$ and $1000$. Which cost component lies in which interval and the actual value of the cost component was randomly set when the graphs were first created.
	We group the NetMaker graphs into groups called \emph{NetMaker-$n$}.
	Each group is defined by the number $n$ of nodes of the graphs in it.
	Within each groups, graphs have varying number of arcs.
	We use the same $s$-$t$-pairs used in \cite{Casas2021a}: for each graph, $20$ randomly generated $s$-$t$-pairs are considered.
	
	\paragraph{Road Networks $d = 2$ and $d = 3$} We use the well known road networks of (parts of) the United States available from the DIMACS Implementation Challenge on Shortest Paths \cite{Dimacs09}. 
	The original instances are two-dimensional: each arc in the network represents a street and the costs of the arc are the street's length and traversal time. 
	For the three dimensional experiments, we add a third cost component that is equal to one on every arc. 
	In \Cref{tab:roadInstances} we list the sizes of the used networks.
	For the experiments we use the same $s$-$t$-pairs ($100$ per road network) that were used in \cite{Ahmadi21}.
	
	\begin{table}
		\small
		\centering
		\caption{Overview of the size of the used road networks.}\label{tab:roadInstances}
		\begin{tabular}{@{}lrrl@{}}
			Instance Name & Nodes   & Arcs     & Considered $s$-$t$-pairs \\
			\midrule
			NY            & $264,346$  & $733,846$   & \multirow{9}{*}{20}     \\
			BAY           & $321,270$  & $800,172$   &                          \\
			COL           & $435,666$  & $1,057,066$  &                          \\
			FLA           & $1,070,376$ & $2,712,798$  &                          \\
			NE            & $1,524,453$ & $3,897,636$ &                          \\
			LKS           & $2,758,119$ & $6,885,658$  &                          \\
			E             & $3,598,623$ & $8,778,114$ &                          \\
			W             & $6,262,104$ & $15,248,146$ &                          \\
			CTR             & $14,081,816$ & $34,292,496$ &                          \\
			\bottomrule
		\end{tabular}
	\end{table}

	\paragraph{Grid Graphs $d = 2$} We consider grid graphs of width $x \in \N$ and height $y \in \N$ where $x$ and $y$ assume the value of every multiple of $50$ between $300$ and $600$. Every such grid graph ($49$ in total) has a source node $s$ connected to all nodes in the leftmost column of the grid. Similarly, all nodes in the rightmost column of the grid are connected to a target node $t$. All other nodes are connected to all their neighbors in the grid. These instances were already used in \cite{Sedeno15, Sedeno19}.
	
	\subsection{MDA vs T-MDA}\label{sec:mdaComparison}
	The \mda{} and the MDA use the same pruning techniques to discard provably irrelevant subpaths of efficient $s$-$t$-paths early during the search. 
	To do so, the dominance bound $\beta_t$ and the front $\mathcal{P}_{st}$ of already found efficient $s$-$t$-paths are used.
	The pruning using the dominance bound in Line \ref{algo:propagate:pruneatfront} of \ref{algo:propagate} and in Line \ref{algo:ncla*:pruneatfront} of \ref{algo:ncla*} is usually not very aggressive since the dominance bound is weak in most practical instances.
	Hence, to enhance the pruning techniques, it is crucial to add paths to $\mathcal{P}_{st}$ as early as possible. 
	Any path therein is a better dominance bound than $\beta_t$. 
	This is precisely what the heuristic $\pi$ is used for: in the \mda{}, the first $s$-$t$-paths are extracted and added to $\mathcal{P}_{st}$ much earlier than in the MDA because processing paths according to their reduced costs $\bar{c}$ guides the search towards the target. 
	This general remark explains the superiority of the \mda{} when compared with the MDA in our experiments.
	
	\Cref{tab:netmaker-results} is an overview of the results obtained for the NetMaker instances. 
	$N_t$ denotes the cardinality of the solution set $\mathcal{P}_{st}$ and gives a hint regarding the difficulty of the considered instances.
	Both algorithms solved all instances within the time limit. 
	The implementation of the MDA is faster than the one used in \cite{Casas2021a}. 
	This can be seen by looking at the NetMaker-$5000$ and the NetMaker-$10000$ instances.
	In \cite{Casas2021a}, the MDA solved all instances from these groups but was slower than the current MDA implementation, even though a faster processor was used. 
	Despite this improvement, the \mda{} is $\times2.16$ to $\times3.88$ faster on average than the MDA. 
	The speedup increases with the size of the networks. On the biggest instances from the NetMaker-$25000$ and the NetMaker-$30000$ groups the maximum speedup is more than $2$ orders of magnitude. 
	The plots in \Cref{fig:3d-net-5000} to \Cref{fig:3d-net-30000} show how the running times of both algorithms in each group of NetMaker instances.
	\begin{table}
		\centering\small
		\caption{Overview of the results collected on NetMaker graphs of different sizes with $d=3$ cost components. Depending on the graph class, $200$ to $240$ $s$-$t$-pairs were randomly built.	The number depends on the number of graphs in each class. $N_t$ refers to the cardinality of the solution set, i.e., the cardinality of the set of cost vectors of efficient $s$-$t$-paths. All instances were solved by both algorithms before the time limit was reached.}\label{tab:netmaker-results}
		\begin{tabular}{ll rr rr rr r}
			\toprule
			& \begin{tabular}[c]{@{}l@{}}\mda{} sol.\\time in $[s]$\end{tabular}  & Instances & $N_t$ & \multicolumn{2}{c}{\mda} & \multicolumn{2}{c}{MDA}  & Speedup \\
			\midrule
			&    & &   &                 Inserted   &  Time     & Inserted  & Time      &           \\
			\cmidrule(lr){5-6}          \cmidrule(lr){7-8}
			\multirow{5}{*}{netMaker-5000n} &(0, 0.5] & 85 & 233 & 22242 & 0.04 & 73310 & 0.14 & 3.37 \\
			& (0.5, 5] & 62 & 789 & 248797 & 1.63 & 542056 & 3.72 & 2.29 \\
			& (5, 50] & 78 & 1528 & 780908 & 15.71 & 1437358 & 30.54 & 1.94 \\
			& (50, 500] & 15 & 3394 & 1712258 & 82.13 & 2803596 & 137.36 & 1.67 \\
			\midrule
			\multirow{5}{*}{netMaker-10000n} &(0, 0.5] & 65 & 288 & 39586 & 0.07 & 181726 & 0.34 & 5.08 \\
			& (0.5, 5] & 60 & 782 & 323734 & 1.77 & 912542 & 5.55 & 3.13 \\
			& (5, 50] & 55 & 1878 & 1134095 & 15.70 & 2505383 & 37.05 & 2.36 \\
			& (50, 500] & 38 & 3217 & 3370375 & 137.92 & 5990825 & 261.42 & 1.90 \\
			& (500, 18000] & 2 & 5200 & 7410039 & 605.85 & 11013419 & 931.79 & 1.54 \\
			\midrule
			\multirow{5}{*}{netMaker-15000n} &(0, 0.5] & 57 & 297 & 59043 & 0.10 & 321345 & 0.63 & 6.53 \\
			& (0.5, 5] & 57 & 840 & 364690 & 1.72 & 1235804 & 6.90 & 4.01 \\
			& (5, 50] & 62 & 1708 & 1246466 & 16.54 & 3231331 & 48.27 & 2.92 \\
			& (50, 500] & 55 & 2986 & 4097699 & 153.54 & 7965304 & 319.57 & 2.08 \\
			& (500, 18000] & 9 & 4876 & 9288820 & 687.00 & 15173429 & 1184.24 & 1.72 \\
			\midrule
			\multirow{5}{*}{netMaker-20000n} &(0, 0.5] & 50 & 324 & 83073 & 0.15 & 458753 & 1.03 & 6.92 \\
			& (0.5, 5] & 48 & 794 & 389028 & 1.58 & 1357132 & 6.50 & 4.12 \\
			& (5, 50] & 67 & 1656 & 1432648 & 17.81 & 3909492 & 55.89 & 3.14 \\
			& (50, 500] & 63 & 2884 & 4083009 & 129.84 & 8798571 & 300.66 & 2.32 \\
			& (500, 18000] & 12 & 5714 & 12125157 & 950.69 & 20908533 & 1690.19 & 1.78 \\
			\midrule
			\multirow{5}{*}{netMaker-25000n} &(0, 0.5] & 44 & 274 & 49470 & 0.10 & 357449 & 0.85 & 8.93 \\
			& (0.5, 5] & 43 & 725 & 416954 & 1.54 & 1710524 & 8.15 & 5.29 \\
			& (5, 50] & 38 & 1214 & 1424486 & 12.30 & 4144590 & 45.11 & 3.67 \\
			& (50, 500] & 53 & 2742 & 5079058 & 159.23 & 11233614 & 399.93 & 2.51 \\
			& (500, 18000] & 22 & 4610 & 12778316 & 827.36 & 22215895 & 1550.70 & 1.87 \\
			\midrule
			\multirow{5}{*}{netMaker-30000n} &(0, 0.5] & 39 & 281 & 57984 & 0.11 & 472935 & 1.22 & 11.08 \\
			& (0.5, 5] & 55 & 719 & 437503 & 1.61 & 2051738 & 10.08 & 6.27 \\
			& (5, 50] & 55 & 1608 & 1473406 & 15.69 & 4670018 & 58.13 & 3.70 \\
			& (50, 500] & 65 & 2911 & 4961126 & 148.45 & 11108866 & 358.13 & 2.41 \\
			& (500, 18000] & 26 & 5239 & 13930103 & 946.90 & 26853793 & 1932.81 & 2.04 \\
			\midrule

			\bottomrule
		\end{tabular}
	\end{table}

 	\begin{figure}
 		\begin{minipage}{.48\linewidth}
 			\captionsetup{type=figure}\includegraphics[width=\textwidth]{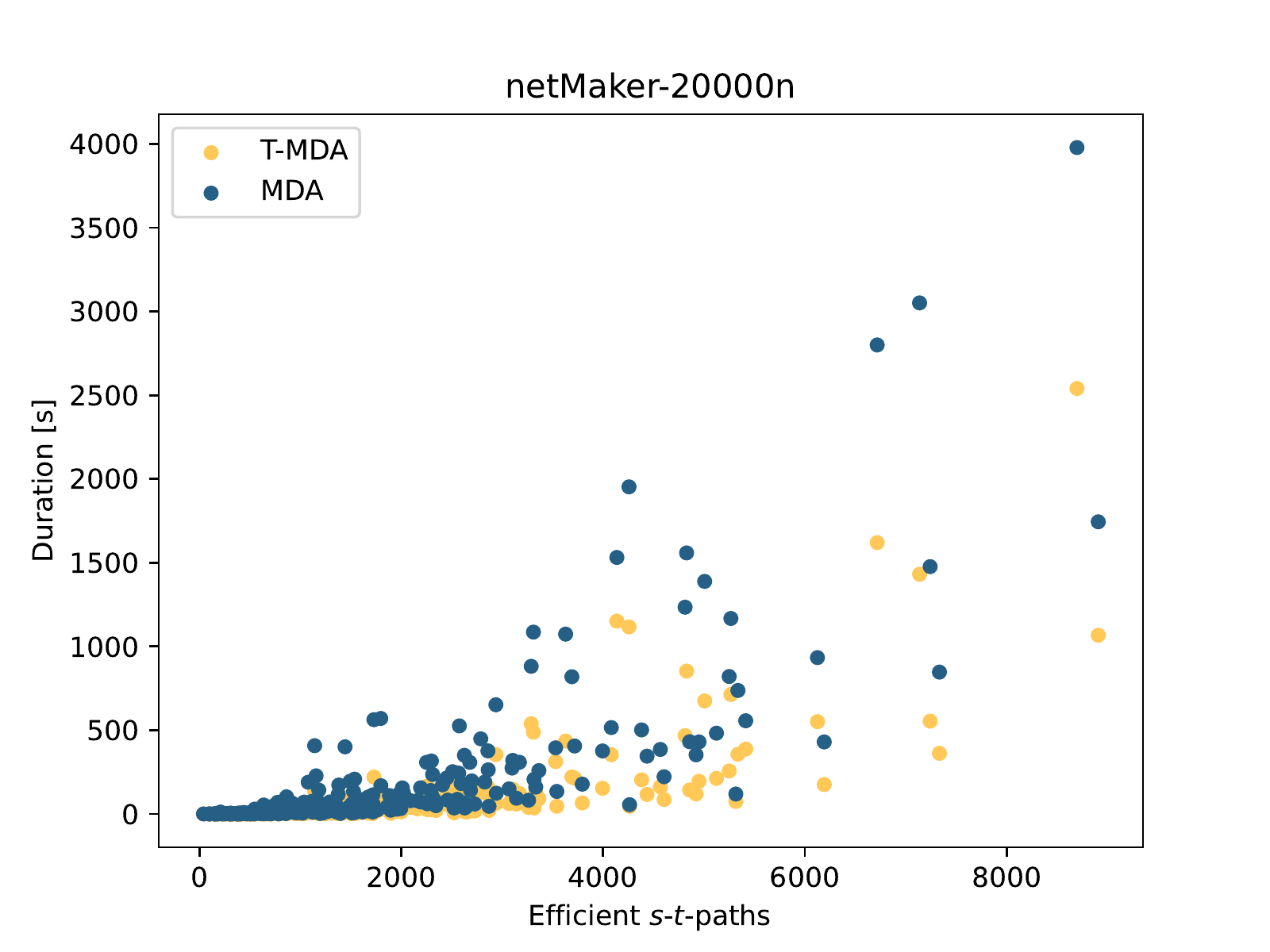}
 			\captionof{figure}{NetMaker graphs with $20,000$ nodes.}\label{fig:3d-net-20000}
 		\end{minipage}
 		\begin{minipage}{.48\linewidth}
 			\captionsetup{type=figure}\includegraphics[width=\textwidth]{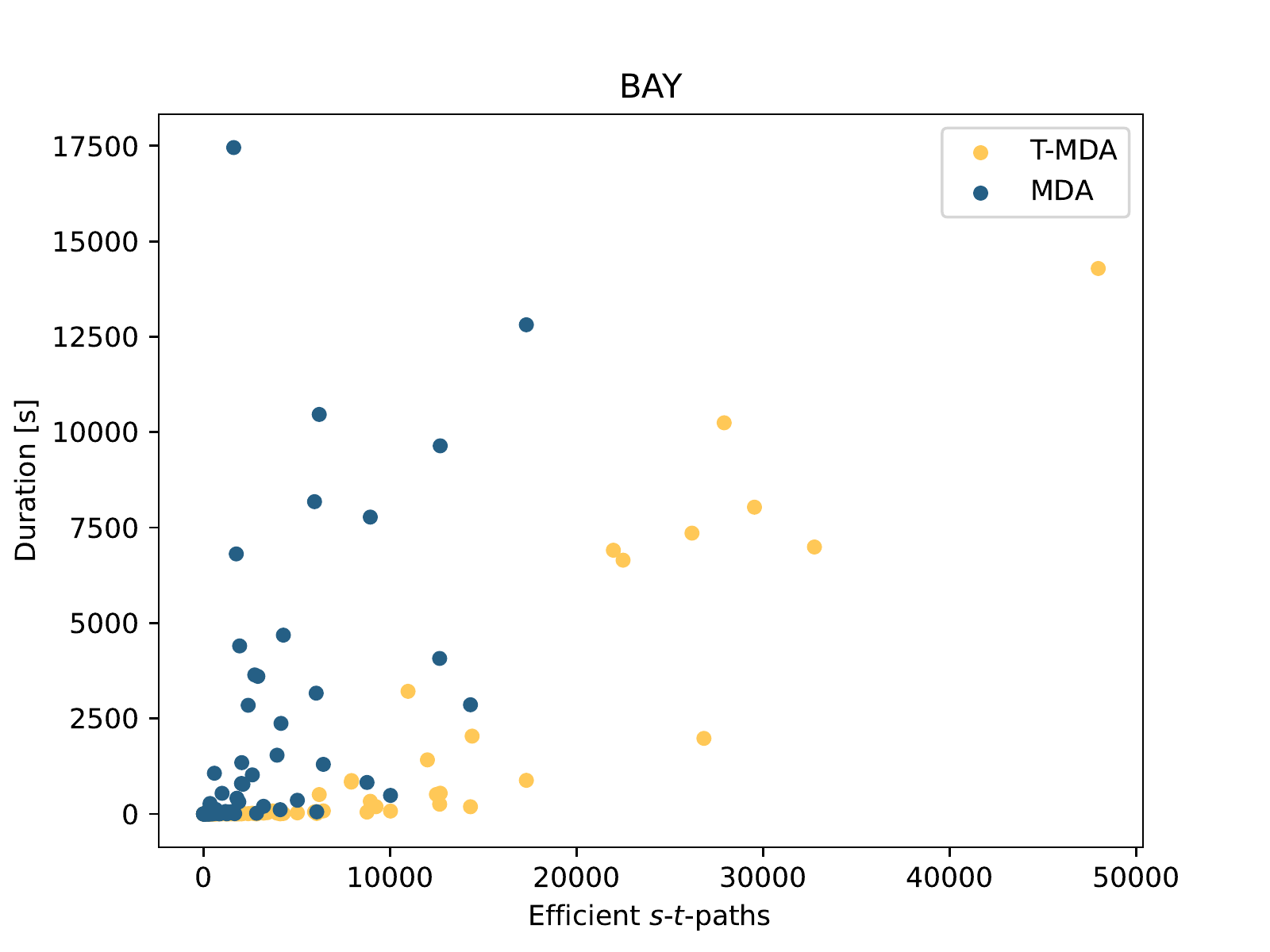}
 			\captionof{figure}{MDA vs. \mda{} on BAY network.}\label{fig:3d-road-BAY}
 		\end{minipage}
 	\end{figure}
	
	In \Cref{tab:multi-road} we summarize the comparison of the MDA and the \mda{} on road networks.
	If one of the algorithms does not solve an instance within the time limit, we set the corresponding solution time to the time limit ($5$h). 
	Moreover, we only consider instances that were solved by at least one algorithm. 
	As the size of the instances grows, the difficulties of both algorithms to solve a considerable percentage of the instances become apparent. 
	On all considered networks, the \mda{} solves more instances than the MDA. 
	On the smaller NY, BAY, and COL networks, the \mda{} solves almost every instance. 
	The average speedup w.r.t. the MDA is more than one order of magnitude. In \Cref{fig:3d-road-NY} to \Cref{fig:3d-road-COL} it becomes clear that the speedup increases as the number of computed efficient paths grows. 
	On the FLA, NE and LKS networks the \mda{} solves at most half of the considered instances but always more than the MDA.
	The average speedups are always greater than an order of magnitude, sometimes two orders of magnitude.
	As far as we know, the road networks E, W, and CTR are considered for the first time with $3$ dimensional arc costs. The \mda{} solves $28/100$ instances on the E network, $9/100$ on the W network, and $5/100$ on the CTR network. On each of these graphs, the number of solved instances at least doubles the number of instances solved by the MDA. The other instances can not be solved due to memory limitations.

\begin{table}
	\centering\small
	\caption{\mda{} vs. MDA on road networks. For every graph, the results of the $100$ instances are reported after dividing the instances into disjoint time intervals based on the time it took the T-MDA to solve them. All reported numbers are rounded geometric means. Scatter plots for each graph in \Cref{fig:3d-road-NY} - \Cref{fig:3d-road-CTR}.}\label{tab:multi-road}
	\begin{tabular}{ll r rrr rrr r}
		\toprule
		& \begin{tabular}[c]{@{}l@{}}\mda{} sol.\\time in $[s]$\end{tabular}  &   $N_t$      & \multicolumn{3}{c}{\mda} & \multicolumn{3}{c}{MDA}  & Speedup \\
		\midrule
		&       &               & Solved &    Inserted   &  Time     & Solved & Inserted  & Time      &           \\
		\cmidrule(lr){4-6}          \cmidrule(lr){7-9}
\multirow{5}{*}{BAY} & (0, 0.5] & 158 & 35/35 & 34475 & 0.05 & 35/35 & 1155634 & 1.82 & 33.66 \\
& (0.5, 5] & 1414 & 14/14 & 816555 & 1.52 & 14/14 & 23109806 & 105.44 & 69.21 \\
& (5, 50] & 3238 & 16/16 & 3277216 & 15.64 & 15/16 & 110689803 & 1652.72 & 105.69 \\
& (50, 500] & 7968 & 9/9 & 10154740 & 119.88 & 7/9 & 157039057 & 4277.69 & 35.68 \\
& (500, 18000] & 17100 & 17/26 & 62043438 & 4907.91 & 3/26 & 411635693 & 16986.79 & 3.46 \\
\midrule
\multirow{5}{*}{COL} & (0, 0.5] & 114 & 26/26 & 18029 & 0.07 & 26/26 & 607911 & 1.84 & 27.75 \\
& (0.5, 5] & 1437 & 15/15 & 809614 & 1.55 & 15/15 & 48424432 & 365.76 & 235.51 \\
& (5, 50] & 3158 & 11/11 & 3232778 & 14.77 & 9/11 & 75683391 & 1263.36 & 85.54 \\
& (50, 500] & 11252 & 13/13 & 11315066 & 194.45 & 8/13 & 271408641 & 8793.20 & 45.22 \\
& (500, 18000] & 19715 & 15/35 & 65005505 & 7650.24 & 1/35 & 326871394 & 17562.17 & 2.30 \\
\midrule
\multirow{2}{*}{CTR} & (50, 500] & 5719 & 2/2 & 12674561 & 125.40 & 1/2 & 346465533 & 12129.81 & 96.73 \\
& (500, 18000] & 15142 & 3/98 & 54537221 & 16774.62 & 0/98 & - & - & - \\
\midrule
\multirow{5}{*}{E} & (0, 0.5] & 299 & 3/3 & 53006 & 0.21 & 3/3 & 1198001 & 1.50 & 7.20 \\
& (0.5, 5] & 1188 & 5/5 & 570246 & 1.11 & 5/5 & 32064930 & 166.16 & 149.49 \\
& (5, 50] & 4050 & 4/4 & 3762525 & 21.83 & 4/4 & 82633250 & 729.67 & 33.42 \\
& (50, 500] & 5840 & 4/4 & 12873184 & 122.09 & 2/4 & 238198157 & 7979.71 & 65.36 \\
& (500, 18000] & 28354 & 12/84 & 69271421 & 14453.27 & 0/84 & - & - & - \\
\midrule
\multirow{5}{*}{FLA} & (0, 0.5] & 83 & 12/12 & 8573 & 0.08 & 12/12 & 209759 & 0.60 & 7.69 \\
& (0.5, 5] & 1032 & 4/4 & 653896 & 1.11 & 3/4 & 109357306 & 3681.38 & 3330.52 \\
& (5, 50] & 4388 & 5/5 & 3121485 & 19.13 & 4/5 & 72448102 & 1603.72 & 83.85 \\
& (50, 500] & 7416 & 16/16 & 11985005 & 158.29 & 6/16 & 315643174 & 13059.13 & 82.50 \\
& (500, 18000] & 21767 & 20/63 & 65802886 & 9914.55 & 1/63 & 229690247 & 17747.98 & 1.79 \\
\midrule
\multirow{5}{*}{LKS} & (0, 0.5] & 18 & 2/2 & 2413 & 0.15 & 2/2 & 18052 & 0.78 & 5.09 \\
& (0.5, 5] & 1667 & 5/5 & 742097 & 1.98 & 5/5 & 26358173 & 131.73 & 66.60 \\
& (5, 50] & 3421 & 1/1 & 2031939 & 8.72 & 1/1 & 43969216 & 427.01 & 48.99 \\
& (50, 500] & 8153 & 4/4 & 11299622 & 193.13 & 2/4 & 154637009 & 6614.20 & 34.25 \\
& (500, 18000] & 17755 & 8/88 & 87552362 & 15742.96 & 1/88 & 266325484 & 17857.19 & 1.13 \\
\midrule
\multirow{5}{*}{NE} & (0, 0.5] & 173 & 6/6 & 26661 & 0.11 & 6/6 & 1246047 & 2.60 & 24.21 \\
& (0.5, 5] & 1406 & 4/4 & 546287 & 1.10 & 4/4 & 16329886 & 60.16 & 54.77 \\
& (5, 50] & 3521 & 12/12 & 2776277 & 14.29 & 10/12 & 107434232 & 1808.56 & 126.58 \\
& (50, 500] & 6799 & 11/11 & 10296214 & 115.25 & 8/11 & 188332397 & 5635.33 & 48.90 \\
& (500, 18000] & 26275 & 16/67 & 81952703 & 13060.69 & 1/67 & 194569520 & 17676.86 & 1.35 \\
\midrule
\multirow{5}{*}{NY} & (0, 0.5] & 130 & 45/45 & 17848 & 0.04 & 45/45 & 392201 & 0.63 & 16.35 \\
& (0.5, 5] & 1642 & 12/12 & 740197 & 1.77 & 12/12 & 17226422 & 88.88 & 50.33 \\
& (5, 50] & 3660 & 19/19 & 3019715 & 18.32 & 19/19 & 109603246 & 2109.26 & 115.13 \\
& (50, 500] & 8992 & 13/13 & 8657858 & 138.75 & 10/13 & 158382727 & 5570.26 & 40.15 \\
& (500, 18000] & 16469 & 10/11 & 41812104 & 1962.75 & 0/11 & - & - & -\\
\midrule
\multirow{4}{*}{W} & (0.5, 5] & 2086 & 2/2 & 1149508 & 2.72 & 2/2 & 85235464 & 878.82 & 323.60 \\
& (5, 50] & 4710 & 1/1 & 2872447 & 21.56 & 1/1 & 29093544 & 211.30 & 9.80 \\
& (50, 500] & 11506 & 2/2 & 11340303 & 173.60 & 1/2 & 206900146 & 7125.17 & 41.04 \\
& (500, 18000] & 19414 & 4/95 & 75540488 & 16784.36 & 0/95 & - & - & - \\
		\bottomrule
	\end{tabular}
\end{table}
	
	\subsection{Biobjective Experiments}\label{sec:biobjExperiments}
	In the computational experiments presented in \cite{Ahmadi21} the \boaEnh{} and \boaBidi{} algorithms outperformed the \boaClassic{} algorithm from \cite{Ulloa20} and the BDA algorithm from \cite{Sedeno19}. 
	The \boaEnh{} algorithm is based on the \boaClassic algorithm introduced in \cite{Ulloa20}.
	Both algorithms can be understood as variations of the classical label-setting MOSP algorithm by Martins \cite{Martins84}. 
	They are equipped with improvements for the biobjective case that turn out to perform well in practical applications.
	Most significantly, the algorithms introduced in \cite{Ulloa20} and \cite{Ahmadi21} utilize lazy queue management.
	This means that a newly explored path $p$ is introduced into the priority queue while omitting both to check if the queue contains paths that dominate $p$ as well as if $p$ dominates existing paths in the queue.
	These dominance checks are commonly referred to as \emph{merge operations} in the literature.
	The lazy queue management approach might be surprising due to the intractability of the MOSP problem.
	But in our experiments it becomes apparent that, this leads to an increased number of extractions from the queue.
	However, since dominance is checked in constant time in the biobjective case, irrelevant paths are discarded directly after their extraction without triggering a whole iteration of the algorithms' main loop. 
	A main takeaway from \cite{Ulloa20} and \cite{Ahmadi21} is that for BOSP problems, the \emph{merge operation} is less efficient than a lazy queue management in practice. 
	
	We refer to the \bda{} and its parallelized bidirectional version as \emph{\bda-based} algorithms.
	Both algorithms are compared with the \boaEnh{} and the \boaBidi{} algorithms to which we refer as \emph{\boaClassic-based} algorithms. 
	The comparisons are done with two different types of priority queues: binary heaps and bucket queues \cite{Denardo79}.
	The latter can only be used for integer costs and in theory, an exponential number of buckets can be non-empty at the same time.
	However, as in \cite{Ahmadi21}, buckets yield better running times in our experiments.
	
	Still, binary heaps are relevant. 
	BOSP instances with rational costs or Time-Dependent BOSP instances with piecewise linear arc cost functions can not be solved using bucket queues.
	Moreover, \bda{}-based algorithm using binary heaps require less iterations to solve any BOSP instance.
	In situations wherein the evaluation of arc costs is costly, this advantage becomes relevant.
	The reason for the reduced number of iterations is that bucket queues yield better performance if the list of explored paths in each bucket is not required to be lex. sorted.
	
	For both types of queues we run the same set of experiments.
	On average, the parallel bidirectional searches are faster than their unidirectional siblings. 
	However, due to the synchronization of threads in a parallel environment, counting the number of iterations is not deterministic in these cases: small variations can happen between different runs of the same instance. 
	Thus, to compare the workload performed by the algorithms, the unidirectional versions are better suited.
	
		\begin{table}
		\centering\small
		\caption{Overview of the results collected on two dimensional grids of different sizes. $N_t$ refers to the geometric mean of the cardinality of the solution sets. All instances were solved before the time limit was reached. The subdivision of the instances into time intervals was done based on the time it took the BDA-based algorithms to solve the instances. Corresponding plots: \Cref{fig:2d-grid-bi-bucket}-\Cref{fig:2d-grid-uni-heap}.}\label{tab:grids-results}
		\begin{tabular}{ll rr rr rr r}
			\toprule
			& \begin{tabular}[c]{@{}l@{}}sol. time in $[s]$\\BDA based\end{tabular}  & Instances & $N_t$ & \multicolumn{2}{c}{BDA based} & \multicolumn{2}{c}{BOA based}  & Speedup \\
			\midrule
			&    & &   &                 Inserted   &  Time     & Inserted  & Time      &           \\
			\cmidrule(lr){5-6}          \cmidrule(lr){7-8}
			\multirow{3}{*}{Grids-Bidi-Bucket} &(0, 0.5] & 2 & 434 & 3544743 & 0.44 & 5671696 & 0.53 & 1.22 \\
			& (0.5, 5] & 30 & 599 & 11202640 & 1.73 & 17626082 & 2.24 & 1.30 \\
			& (5, 50] & 17 & 902 & 36165995 & 7.72 & 56578747 & 10.36 & 1.34 \\
			\midrule
			\multirow{2}{*}{Grids-Bidi-Heap} & (0.5, 5] & 26 & 553 & 8386782 & 1.82 & 15612581 & 3.24 & 1.78 \\
			& (5, 50] & 23 & 864 & 30789899 & 9.55 & 57730831 & 17.77 & 1.86 \\
			\midrule
			\multirow{2}{*}{Grids-Uni-Bucket} & (0.5, 5] & 32 & 587 & 7357548 & 1.89 & 10945046 & 2.06 & 1.09 \\
			& (5, 50] & 17 & 902 & 25355482 & 8.66 & 37587174 & 10.67 & 1.23 \\
			\midrule
			\multirow{2}{*}{Grids-Uni-Heap} & (0.5, 5] & 25 & 546 & 5836388 & 2.14 & 9843951 & 3.02 & 1.41 \\
			& (5, 50] & 24 & 859 & 21198599 & 10.34 & 36295558 & 16.80 & 1.63 \\
			\bottomrule
		\end{tabular}
	\end{table}
	
	In \Cref{tab:grids-results} and in Figures \ref{fig:2d-grid-bi-bucket}-\ref{fig:2d-grid-uni-heap} we see that on two dimensional grids, the BDA-based algorithms always outperform the BOA-based algorithms.
	Interestingly, the bidirectional parallel algorithms using bucket queues are the only ones that need more than $5$s to solve some instances. 
	At the same time, these scenario is the only one in which $2$ instances are solved in less than $0.5$s.
	All in all, there are no big differences between the unidirectional and the bidirectional versions of the algorithms.
	The \emph{Grids-Bidi-Heap} and \emph{Grids-Uni-Heap} rows in \Cref{tab:grids-results} give a good intuition: the set of instances solved in less and in more than $5$s are almost identical and the average  number of queue insertions performed by the bidirectional search is much higher than in the unidirectional scenario. 
	This happens because the calculated heuristic does not effectively guide the searches towards their target nodes and this in turn is due to the randomly generated arc costs in the grid. 
	
	On road networks the situation changes.
	The calculated heuristics capture the metrical structure of the arc costs (road length and traversal time) and effectively guide the searches.
	On average, the bidirectional algorithms using bucket queues outperform their unidirectional siblings and the versions of the algorithms using binary heaps. 
	We focus on a head to head comparison of the different versions of the BDA-based and BOA-based algorithms. 
	In each comparison, both considered algorithms are either bidirectional or unidirectional and they always use the same type of priority queue. 
	\Cref{tab:overview2d} gives an overview about where to find the relevant information and plots for every comparison. 
	
	\begin{table}
		\centering \small
		\caption{Tables and figures reporting the results of the experiments on bidimensional road networks.}\label{tab:overview2d}
		\begin{tabular}{c c c c c c}
			&& \multicolumn{2}{c}{Bucket Heap} & \multicolumn{2}{c}{Binary Heap}\\
			&& Overview & Plots & Overview & Plots\\
			\midrule
			\multirow{2}{*}{Road Networks} & Bidirectional 	& \Cref{tab:bi-bi-bucket}	& \Cref{fig:2d-bbda-boba-bucket-BAY} - \Cref{fig:2d-bbda-boba-bucket-CTR} & \Cref{tab:bi-bi-heap} & \Cref{fig:2d-bbda-boba-heap-BAY} - \Cref{fig:2d-bbda-boba-heap-CTR}\\
			& Unidirectional	& \Cref{tab:bi-uni-bucket}	& \Cref{fig:2d-bda-boa-bucket-BAY} - \Cref{fig:2d-bda-boa-bucket-CTR} & \Cref{tab:bi-uni-heap} & \Cref{fig:2d-bda-boa-heap-BAY} - \Cref{fig:2d-bda-boa-heap-CTR}\\
			\bottomrule
		\end{tabular}
	\end{table}
	
	\paragraph{Instances solved by \bda{}-based algorithms in $(0, 0.5]s$} These instances always favor the \boaEnh{} and the \boaBidi{} algorithm. 
	This is because on small instances the percentage of needless insertions into the queue that are performed due to the lazy heap management is small.
	
	\paragraph{Instances solved by \bda{}-based algorithms in $(0.5, 5]s$}
	In the bidirectional scenario using buckets, the TB-BDA slightly outperforms the \boaBidi{} algorithm on all road networks except the CTR network, where the average speedup is $\times 0.74$.
	In the unidirectional scenario using buckets, the \boaBidi{} algorithm is better than the TB-BDA algorithm on all road networks. Using binary heaps, the \bda{}-based algorithms always outperform the BOA* based algorithms.
	The greatest speedup of $\times 1.42$ is achieved on the FLA network in the unidirectional scenario.
	
	\paragraph{Instances solved by \bda{}-based algorithms in $(5, 50]s$} 
	The \bda{}-based algorithms almost always outperform the BOA* based algorithms by $\times 1.07$ (CTR) to $\times 1.15$ (W) in the bidirectional scenario and $\times 0.97$ (COL) to $\times 1.07$ (FLA) in the unidirectional scenario.
	Using binary heaps, the speedups increase: $\times 1.41$ (COL) to $\times 1.65$ (LKS) in the bidirectional scenario and $\times 1.43$ (COL) to $\times 2.11$ (E) in the unidirectional scenario.
	
	\paragraph{Instances solved by \bda{}-based algorithms in $(50, 500]s$}
	The results favor the \bda{}-based algorithms. 
	Using buckets, the greatest speedup of $\times 1.18$ in the bidirectional scenario is reached on the FLA instances. 
	Using binary heaps, the speedups are again greater than for buckets.
	On the E road network, the speedup on  instances that the \bda{}solves between $50s$ and $500s$ reaches $\times 2.42$. 
	
	\paragraph{Instances solved by \bda{}-based algorithms in $(500, 7200]s$} 
	We observe a slight decrease of the speedup factor but still the \bda{}-based algorithms outperform the BOA* based algorithms in every scenario. 
	The implementation of the BOA* based algorithms in \cite{Ahmadi21} relies on a so called \emph{memory pool}.
	It allocates big chunks of memory a priori to reduce reallocation of memory during the algorithm and also to avoid memory fragmentation. 
	Moreover, labels that are discarded can be reused.
	These advantages become more useful as the workload of the algorithm increases.
	We tried a similar memory management technique for \bda{}-based algorithms.
	The running time gains were limited because these algorithms are designed to have a possibly small set of explored active paths at any point in time (at most one path per node in the priority queue $Q$).
	Since we wanted to deliver an implementation of the \bda{} algorithms that follows the theoretical description in this paper and the derived bounds on the space consumption, we decided to report the results without using memory pools in the \bda{}-based algorithms.
	
	As already mentioned, the number of extractions made by each algorithm is better assessed by looking at the unidirectional results (\Cref{tab:bi-uni-bucket} and \Cref{tab:bi-uni-heap}). 
	The \bda{}-based algorithms always perform less iterations than the BOA based algorithms.
	This is to be expected because of the lazy heap management that is used in the BOA based algorithms.
	Moreover, the extra effort done to maintain binary heaps sorted (single buckets are not sorted in a bucket queue) pays off when it comes to the number of iterations. 
	All in all, the \bda{}-based algorithms using binary heaps need a minimal number of iterations to solve every instance.
	
	\section{Conclusion}
	In this paper, we introduced the \emph{Targeted Multiobjective Dijkstra Algorithm} (\mda) for the One-to-One Multiobjective Shortest Path (MOSP) Problem. 
	The \mda{} is a variant of the Multiobjective Dijkstra Algorithm (MDA) that guides the search towards the target node using a heuristic.
	The underlying idea is similar to the one used in the design of the $\text{A}^*$ algorithm for the classical One-to-One Shortest Path Problem and achieves an analogous effect. 
	Making use of the dimensionality reduction that enables dominance checks in constant time in the biobjective case, we present a tuned version of the \mda{}, the Targeted Biobjective Dijkstra Algorithm (\bda{}) for the Biobjective Shortest Path (BOSP) Problem. 
	The resulting algorithms are benchmarked against state of the art algorithms on different types of graphs from the literature. 
	To the best of our knowledge, this paper is the first to present the solution of large scale MOSP instances, e.g., road networks, with more than $d=2$ objectives by $\text{A}^*$-like algorithms. 
	Our experiments show that the \mda{} clearly outperforms the state of the art MDA on all considered instances.
	In the experiments on BOSP instances, we benchmarked the \bda{}-based algorithms (unidirectional and parallelized bidirectional) against recently published state of the art BOSP algorithms. 
	These algorithms are based on the classical label setting algorithm for MOSP by Martins and use a lazy queue management instead of the time consuming merge operations required in the original version of the algorithm. On average, the BDA based algorithms achieved better running times on all instances except those solvable within fractions of a second.
	
	\bibliographystyle{plain}
	\bibliography{literature}
	
	\appendix
	\newpage
	\section{Pseudocode of the Targeted Biobjective Dijkstra Algorithm (\bda)}\label{apendix:pseudocode}
		\begin{algorithm}[H]
		\small{
			\SetKwInOut{Input}{Input} \SetKwInOut{Output}{Output}
			\Input{MOSP instance $(G, s, t, c)$,\\
				heuristics $\pi$ and $\overleftarrow{\pi}$ ($\overleftarrow{\pi}$ only for BT-BDA),\\
				nadir point $\beta_t$ of $c(\mathcal{P}_{st})$,\\
				$(1,2)$-lexicographic $t$-to-all shortest path tree $T_{(1,2)}$.}
			\BlankLine
			\Output{Minimal complete set $\mathcal{P}_{st}$ of efficient $s$-$t$-paths. In the bidirectional BT-BDA, it is not complete and has to be combined with the set $\mathcal{P}_{ts}$ computed by the backward search.}
			\BlankLine
			
			Priority Queue: $Q \leftarrow{} \emptyset{}$\label{algo:bda*:initStart}\tcp*{$Q$ is sorted according to the $\bar{c}$ values of the contained paths.}
			Efficient $s$-$v$-paths: $\mathcal{P}_{sv} \leftarrow \emptyset$, $\forall v \in V$\label{algo:bda*:initPerm}\;
			\tcp{Due to the dimensionality reduction (\Cref{prop:dimensionality}), $\beta_{v,2}$ is the only relevant information for the dominance checks at $\mathcal{P}_{sv}$ until a new $s$-$v$-path is extracted.}
			Second cost component of the latest $s$-$v$-path extracted from $Q$: $\gamma_{sv}^* \leftarrow \infty$, $\forall v \in V\backslash\{t\}$\;
			$p_{\text{init}} \leftarrow{} ()$\;
			$Q \leftarrow{} Q.\mathtt{insert}(p_{\text{init}})$\;
			\BlankLine
			\While{$Q \neq \emptyset$}{
				$p \leftarrow{} Q.\mathtt{extractMin}()$ \label{algo:bda*:extraction}\;
				$v \leftarrow{} \mathit{head}(p)$\tcp*[]{\texttt{head} returns the last node of $p^*_v$. If $p^*_v = p_{\text{init}}$, it returns $s$.}
				\BlankLine
				\tcp{Stopping condition for bidirectional search. Opposite search lowers $\beta_{t,1}$ when it finds an efficient $t$-$s$-path. When this condition holds,
				the current search will not find any $s$-$t$-path with an efficient cost vector that is unknown to the opposite search.}
				\lIf{$\bar{c}_1(p) \geq \beta_{t,1}$}{\textbf{break}\label{algo:bda*:breakCondition}}
				$\gamma_{sv}^* \leftarrow c_2(p)$\label{algo:bda*:permanent}\;
				\BlankLine
				$p^{new}_v \leftarrow{} \ref{algo:ncl2d}(p, \incoming{v}, \mathcal{P}, \beta_t)$ \label{algo:bda*:ncl}\;
				\lIf{$p^{new}_v \neq \Null$} 
				{$Q.\mathtt{insert}(p^{new}_v)$\label{algo:bda*:insert}}
				
				\BlankLine
				\If{$\beta_{t,2} \leq \bar{c}_2(p)$}{\textbf{continue}\label{algo:bda*:dominanceAfterExtract}}
				\Else{
					\lIf{$\mathcal{P}_{sv} == \emptyset$}{$\overleftarrow{\pi}_1(v) \leftarrow c_1(p)$\label{algo:bda*:updateHeuristic}}
				}

				\BlankLine
				\tcp{Lines \ref{algo:bda*:initShortcut}-\ref{algo:bda*:endShortcut} possibly add shortcuts to $\mathcal{P}_{st}$. See \Cref{sec:shortcuts}.}
				$p_{st} \leftarrow p \circ T_{(1,2)}(v)$\label{algo:bda*:initShortcut}\;
				
				\If(\tcp*[f]{Check if $p_{st}$ is dominated at $\mathcal{P}_{st}$.}){$\beta_{t,2} > c_2(p_{st})$} {
					$q \leftarrow$ last path in $\mathcal{P}_{st}$\;
					\lIf{$c_1(q) == c_1(p_{st})$ and $c_2(q) > c_2(p_{st})$} {
						$\mathcal{P}_{st} \leftarrow ((\mathcal{P}_{st} \backslash q), p_{st})$ 
					}
					\lElse{
						$\mathcal{P}_{st} \leftarrow (\mathcal{P}_{st}, p_{st})$
					}
					$\beta_{t,2} \leftarrow c_2(p_{st})$\label{algo:bda*:improveUb}\tcp*{Strengthen upper bound.}
					\lIf(\tcp*[f]{No further efficient $s$-$t$-path via $v$ exists.}){$\pi_1(v) == \beta_{v,1}$}{\textbf{continue}\label{algo:bda*:abortIteration}\label{algo:bda*:endShortcut}}
				}
				\BlankLine
				$success \leftarrow{} \textbf{False}$\;
				\lFor{$w \in \outgoing{v}$}{
					$(Q, success) \leftarrow{} \ref{algo:propagate2d}(p, w, Q, \mathcal{P}, \gamma_{sw}^*, \beta_t)$
				}
				\lIf{$success == \textbf{True}$}{$\mathcal{P}_{sv}  \leftarrow (\mathcal{P}_{sv}, p)$\label{algo:bda*:permanentForNcl}}
			}
			\Return $\P_{st}$;
			
			\caption[caption]{Targeted Biobjective Dijkstra Algorithm}\label{algo:bda*}}
		\end{algorithm}
	
		\begin{procedure}
		\small{
			\SetKwInOut{Input}{Input}
			\SetKwInOut{Output}{Output}
			\Input{$s$-$v$-path $p$, \\
					Node $w \in \outgoing{v}$,\\
					Priority Queue $Q$, \\
					Permanent paths $\mathcal{P}$, \\
					$2$nd cost component of last extracted $s$-$w$-path $\gamma_{sw}^*$, \\
					upper bound $\beta_t$.}
			\Output{Updated Priority Queue $Q$, Boolean $success$ indicating whether $p$ was successfully propagated along $(v,w)$.}
			\BlankLine
			$p_w^{new} \leftarrow p \circ (v,w)$\;
			$success \leftarrow \textbf{False}$\;
			\tcp{These checks are the dominance checks of $p_w^{new}$ at the sets of efficient paths at $t$ and at $w$.}
			\If{$\beta_{t,2} \leq \bar{c}_2(p_w^{new}) \text{ ot } \beta_{w,2} \leq c_2(p_w^{new})$}{\label{algo:propagate2d:pruneatfront}
				\Return $(Q, \,success)$\;
			}
			\Else{
				$success \leftarrow \textbf{True}$\;
				\uIf{\textbf{not} $Q.\mathtt{contains}(w)$}{
					$Q.\mathtt{insert}\left( p_w^{new} \right)$\label{algo:propagate2d:insert}\;
				}
				\uElse{
					$q \leftarrow Q.\mathtt{getPath}(w)$\;
					\If{$\bar{c}(p_w^{new}) \lexsmaller \bar{c}(q)$\label{algo:propagate2d:flexCheck}}{
						$Q.\mathtt{decreaseKey}(w, \, p_w^{new})$\label{algo:propagate2d:decrease}\;
					}
				}
			}
			
			\Return $(Q, \,success)$\;
			\caption{propagate2d()}\label{algo:propagate2d}}
	\end{procedure}
	
	\begin{procedure}
		\small{
			\SetKwInOut{Input}{Input}
			\SetKwInOut{Output}{Output}
			\Input{$s$-$v$-path $p$, \\
				Neighborhood $\incoming{v}$,\\
				$\forall (u,v) \in A$, $\lpl_{uv}$: first path in $\mathcal{P}_{su}$ that can be a queue path for $v$ after expansion along $(u,v)$, \\
				Permanent paths $\mathcal{P}$,\\
				Upper bound $\beta_t$.}
			\Output{New queue path for $v$, if one exists.}
			\BlankLine
			$p_v^{new} \leftarrow{} \Null$\tcp*[r]{Assume $\bar{c}(p_v^{new}) = [\infty]_{i=1}^d$.}
			\For{$u \in \incoming{v}$}{
				\For{$p_u \in (\mathcal{P}_{su}[\lpl_{uv}], \, |\mathcal{P}_{su}|)$}{
					$p_v \leftarrow p_u \circ (u,v)$\;
					\tcp{First condition is a dominance check of $\bar{c}(p_v)$ at $c(\mathcal{P}_{st})$. The second condition ensures that $p$ is lex. smaller than $p_v$ but does not dominate it. As a consequence, $p_v$ would not be dominated at $\mathcal{P}_{sv}$.}
					\If{$\beta_{t,2} \leq \bar{c}_2(p_v)$ or \textbf{not} $(c_1(p_v) > c_1(p) \text{ and } c_2(p_v) < c_2(p))$}{\label{algo:ncla*:pruneatfront}
						$\lpl_{uv} \leftarrow \lpl_{uv} + 1$\;
						\textbf{continue}\;
					}
					\lIf{$\bar{c}(p_v) \lexsmaller \bar{c}(p_v^{new})$  \label{algo:ncla*:flex:2d}}{
						$p_v^{new} \leftarrow p_v$} 
					\textbf{break}\tcp*[r]{Even if the lex. check fails, we stop and do not increase $\lpl_{uv}$.}\label{algo:ncla*:break}
					
				}    					
			}
			\Return $p_v^{new}$;
			\caption{nextCandidatePath2d()}\label{algo:ncl2d}}
	\end{procedure}
	
	\newpage
	\section{Scatter plots MDA vs. \mda{}}
		\subsection{NetMaker instances}
		\begin{figure}[H]
			\begin{minipage}{.48\linewidth}
				\captionsetup{type=figure}\includegraphics[width=\textwidth]{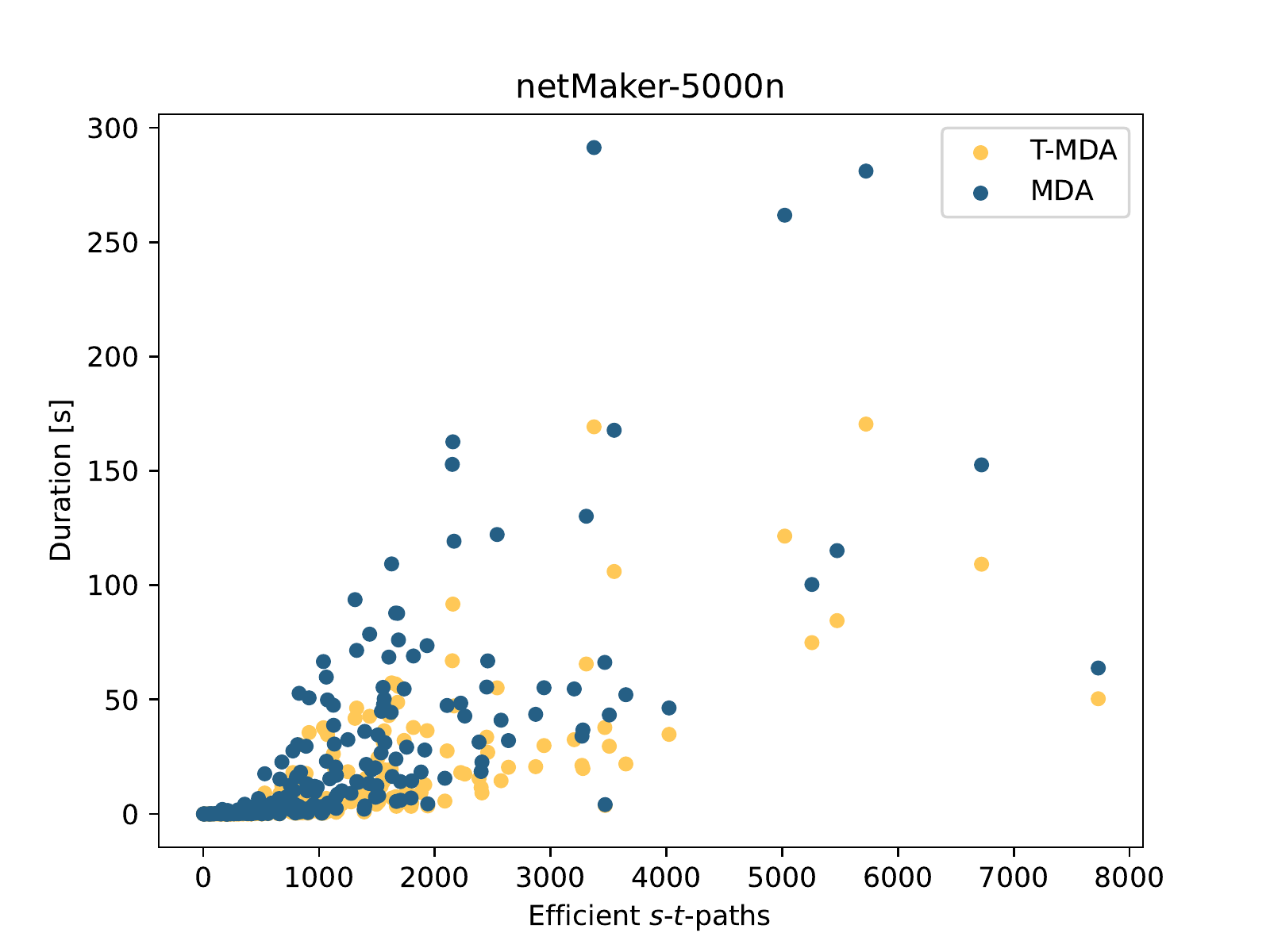}
				\captionof{figure}{NetMaker graphs with $5,000$ nodes.}\label{fig:3d-net-5000}
			\end{minipage}
			\begin{minipage}{.48\linewidth}
				\captionsetup{type=figure}\includegraphics[width=\textwidth]{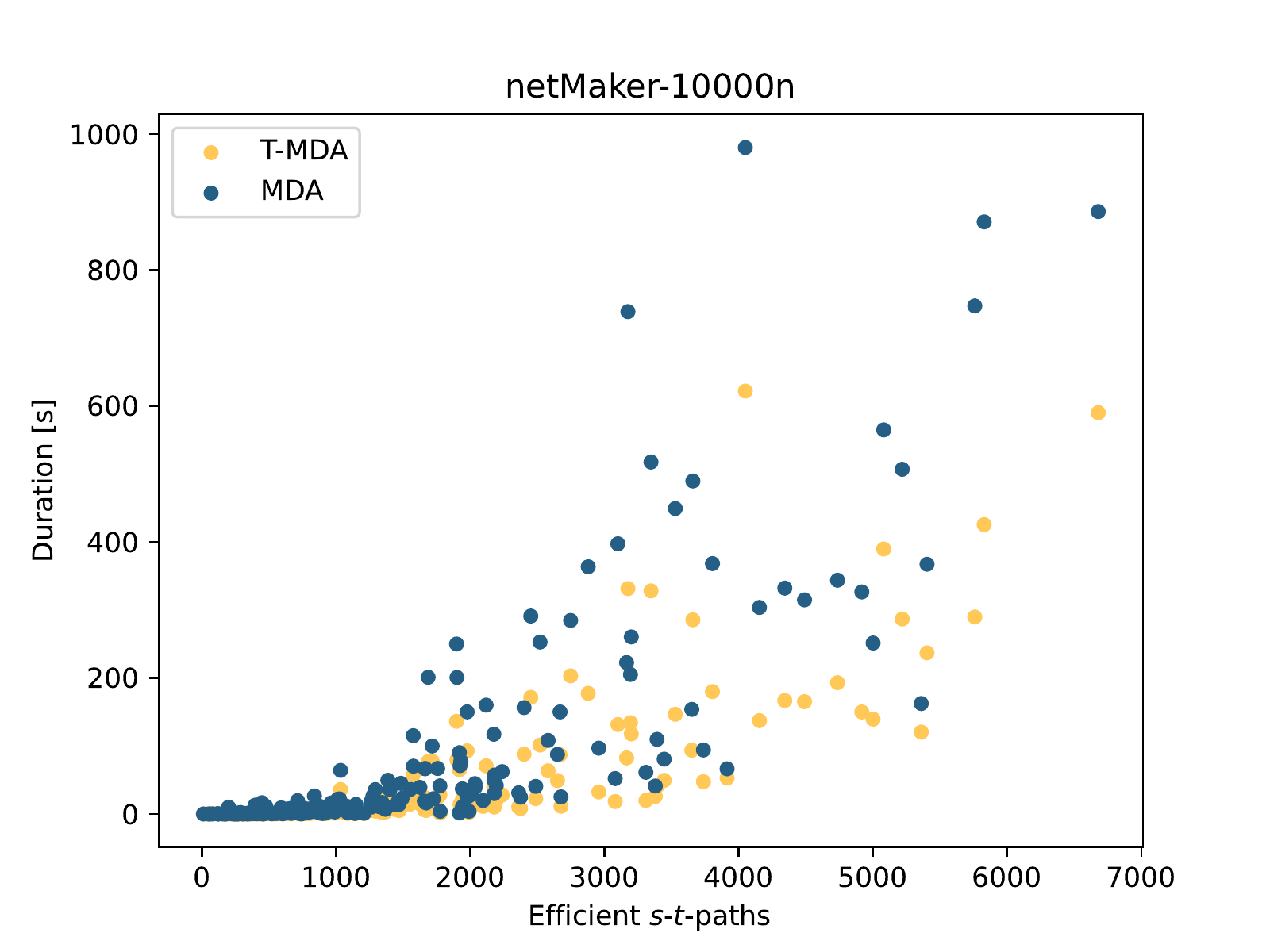}
				\captionof{figure}{NetMaker graphs with $10,000$ nodes.}\label{fig:3d-net-10000}
			\end{minipage}
		\end{figure}
		\begin{figure}[H]
			\begin{minipage}{.48\linewidth}
				\captionsetup{type=figure}\includegraphics[width=\textwidth]{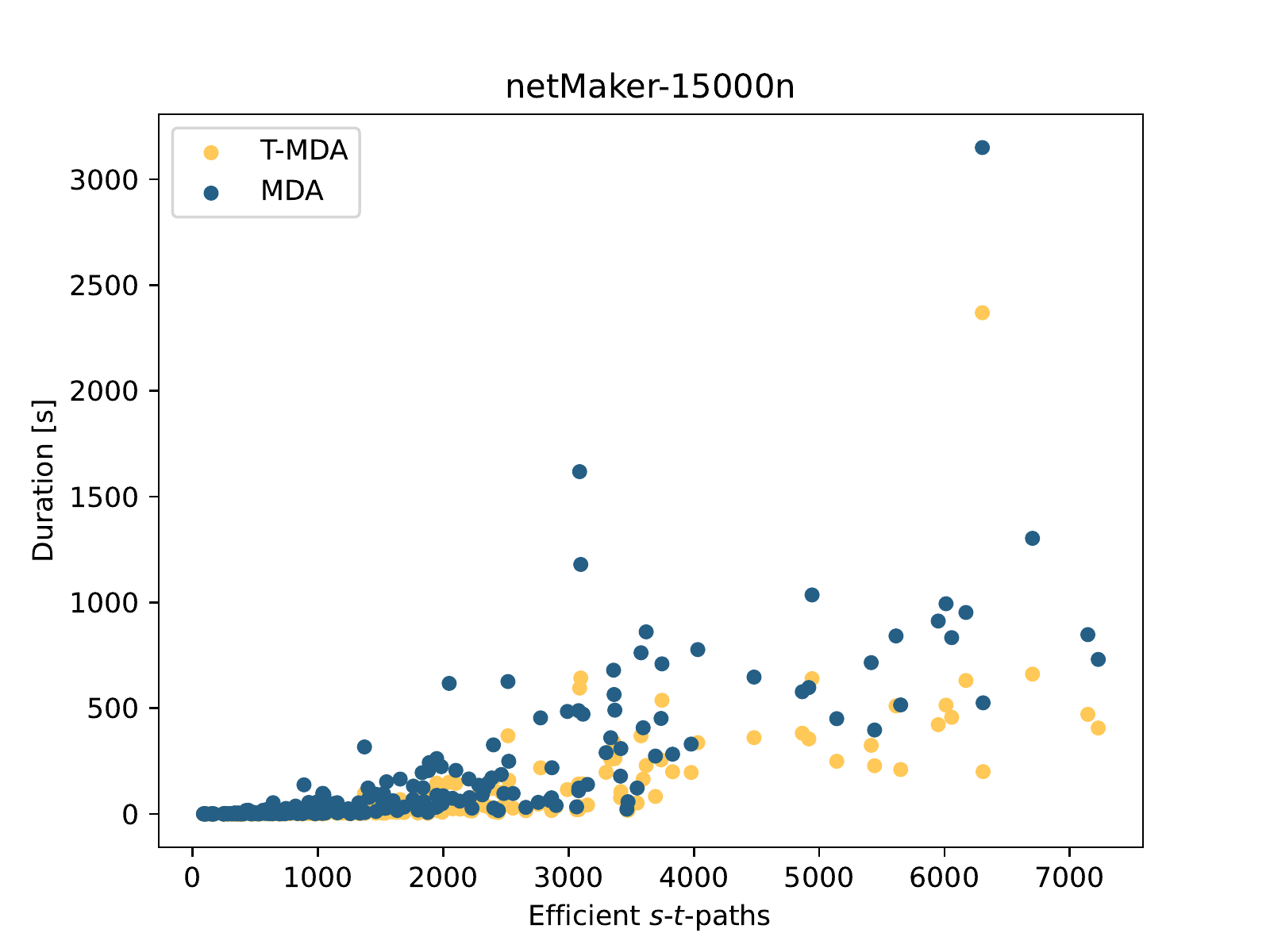}
				\captionof{figure}{NetMaker graphs with $15,000$ nodes.}\label{fig:3d-net-15000}
			\end{minipage}
			\begin{minipage}{.48\linewidth}
				\captionsetup{type=figure}\includegraphics[width=\textwidth]{results/3d-netMaker/NetMaker20000.pdf}
				\captionof{figure}{NetMaker graphs with $20,000$ nodes.}\label{fig:3d-net-20000}
			\end{minipage}
		\end{figure}
		\begin{figure}[H]
			\begin{minipage}{.48\linewidth}
				\captionsetup{type=figure}\includegraphics[width=\textwidth]{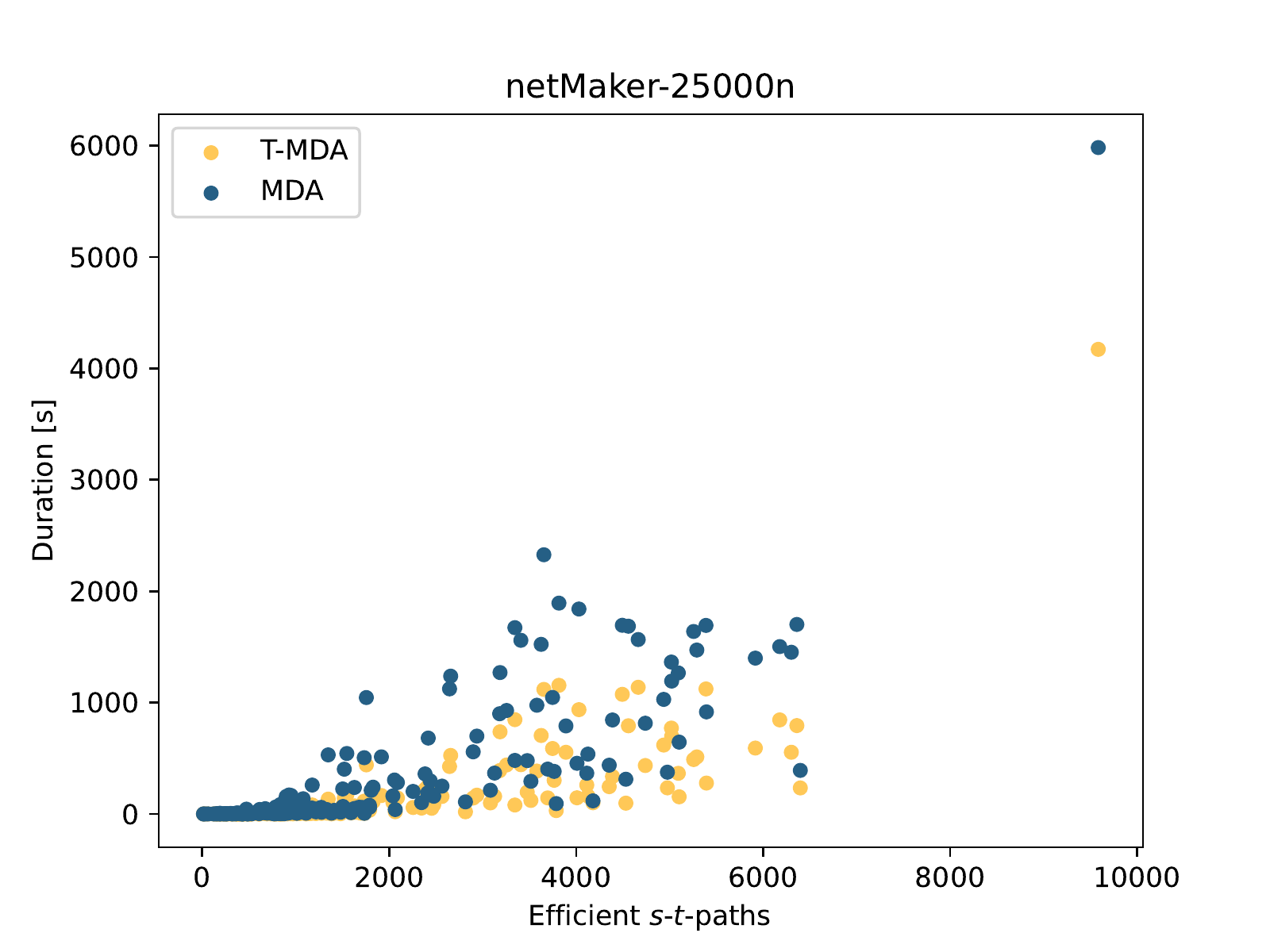}
				\captionof{figure}{NetMaker graphs with $25,000$ nodes.}\label{fig:3d-net-25000}
			\end{minipage}
			\begin{minipage}{.48\linewidth}
				\captionsetup{type=figure}\includegraphics[width=\textwidth]{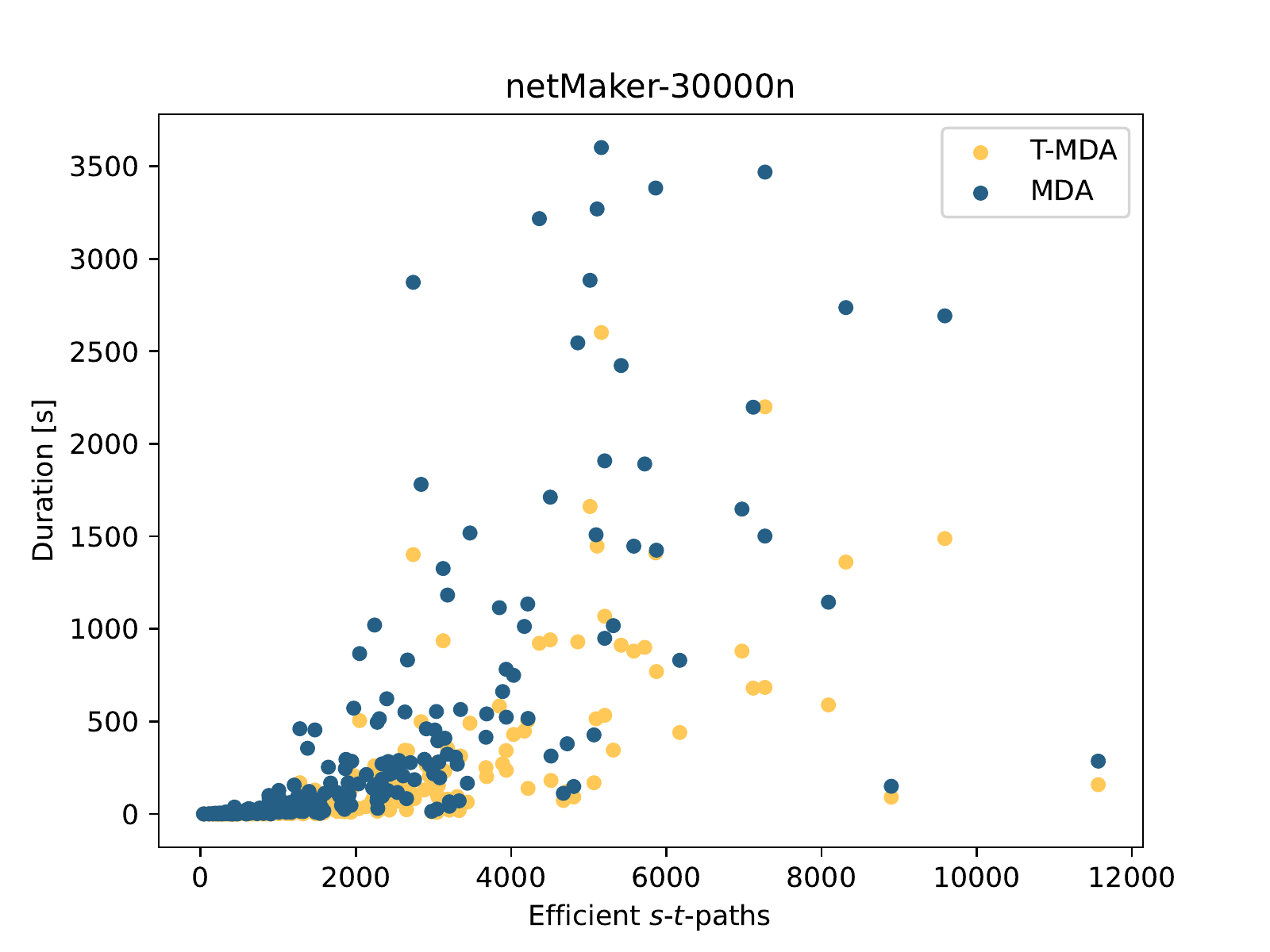}
				\captionof{figure}{NetMaker graphs with $30,000$ nodes.}\label{fig:3d-net-30000}
			\end{minipage}
		\end{figure}
	
		\subsection{Road networks $3$d}
			\begin{figure}[H]
			\begin{minipage}{.48\linewidth}
				\captionsetup{type=figure}\includegraphics[width=\textwidth]{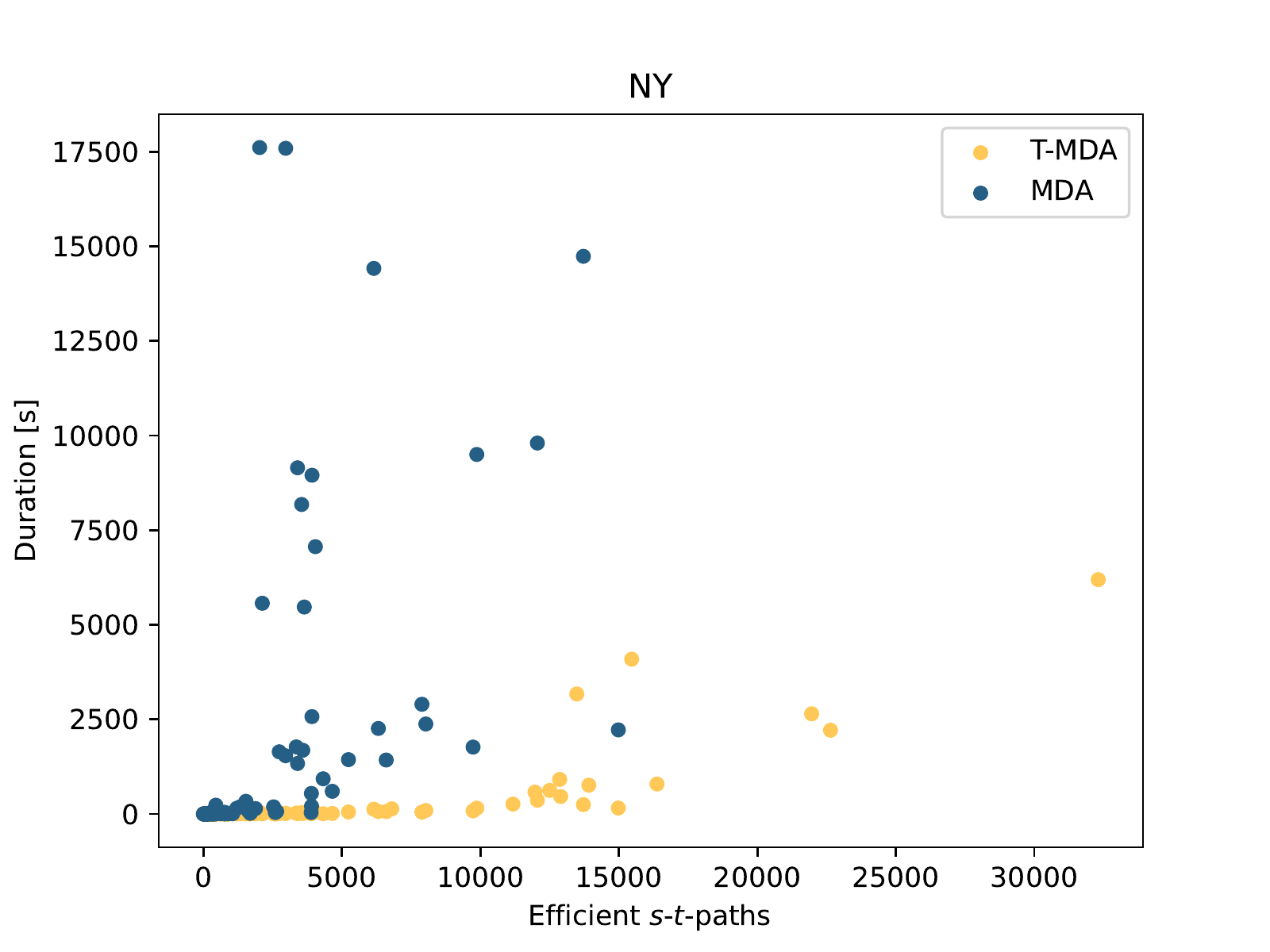}
				\captionof{figure}{MDA vs. \mda{} on NY network.}\label{fig:3d-road-NY}
			\end{minipage}
			\begin{minipage}{.48\linewidth}
				\captionsetup{type=figure}\includegraphics[width=\textwidth]{results/3d-roads/BAY.pdf}
				\captionof{figure}{MDA vs. \mda{} on BAY network.}\label{fig:3d-road-BAY}
			\end{minipage}
		\end{figure}
		\begin{figure}[H]
			\begin{minipage}{.48\linewidth}
				\captionsetup{type=figure}\includegraphics[width=\textwidth]{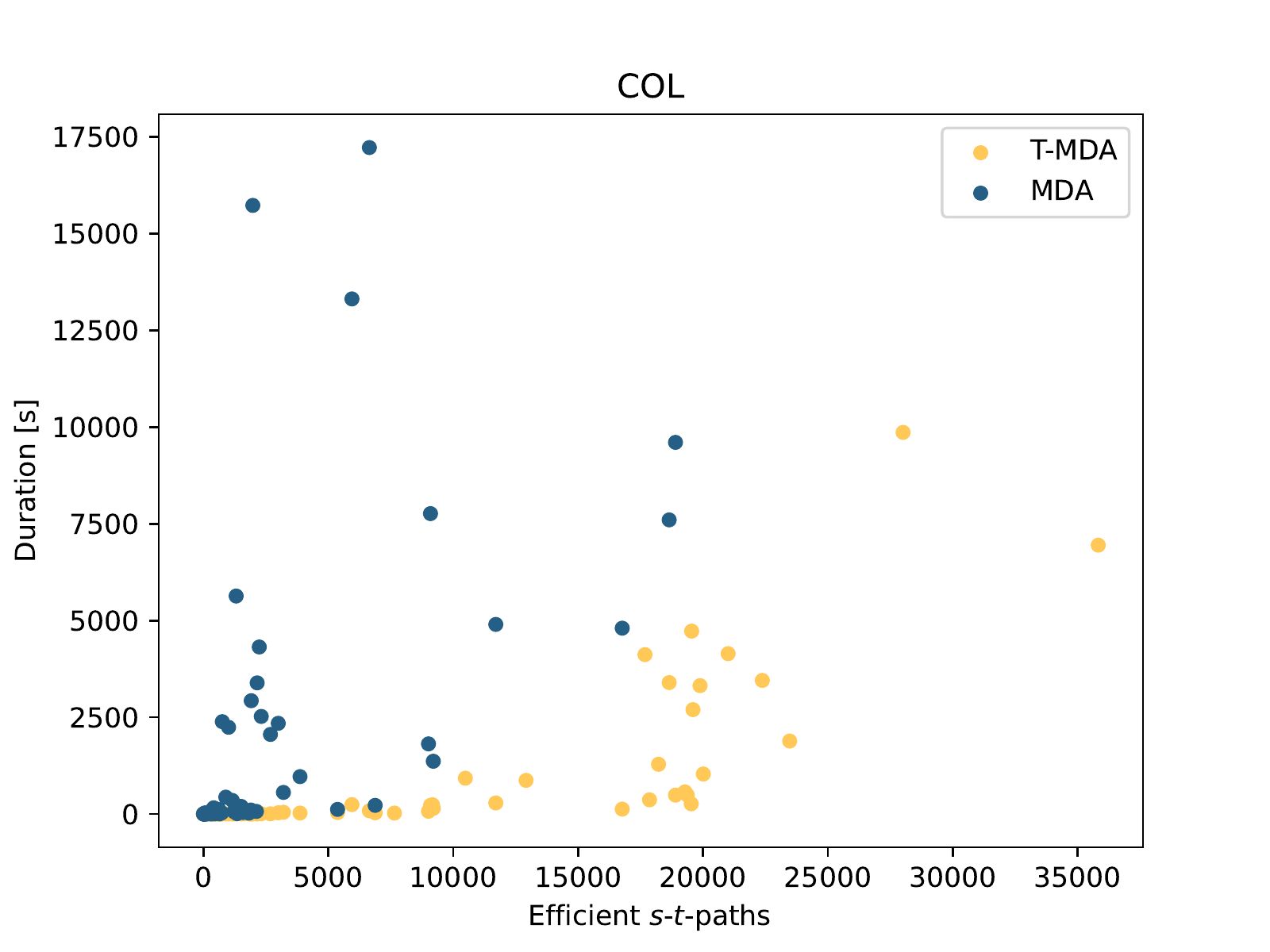}
				\captionof{figure}{MDA vs. \mda{} on COL network.}\label{fig:3d-road-COL}
			\end{minipage}
			\begin{minipage}{.48\linewidth}
				\captionsetup{type=figure}\includegraphics[width=\textwidth]{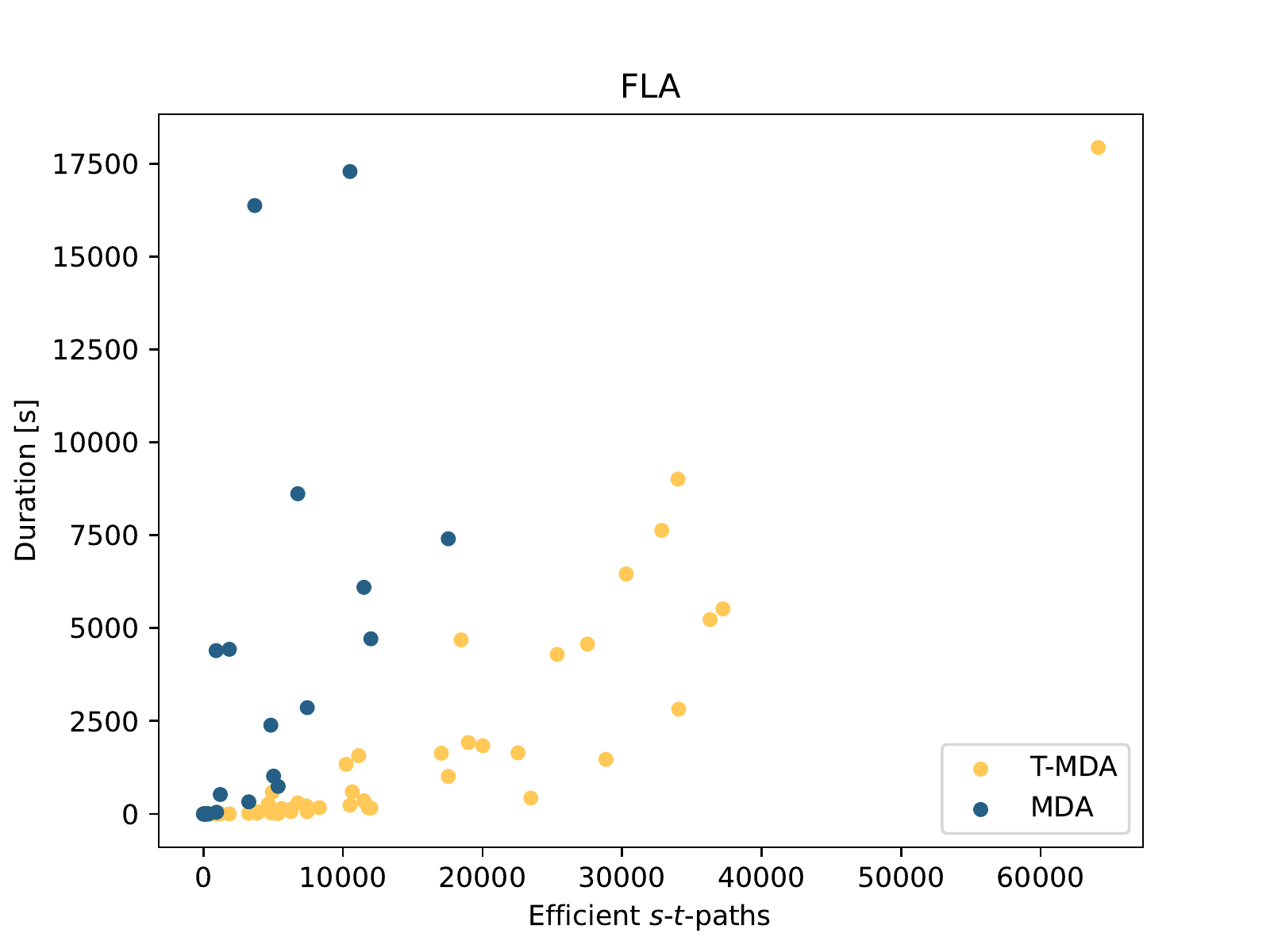}
				\captionof{figure}{MDA vs. \mda{} on FLA network.}\label{fig:3d-road-FLA}
			\end{minipage}
		\end{figure}
		\begin{figure}[H]
			\begin{minipage}{.48\linewidth}
				\captionsetup{type=figure}\includegraphics[width=\textwidth]{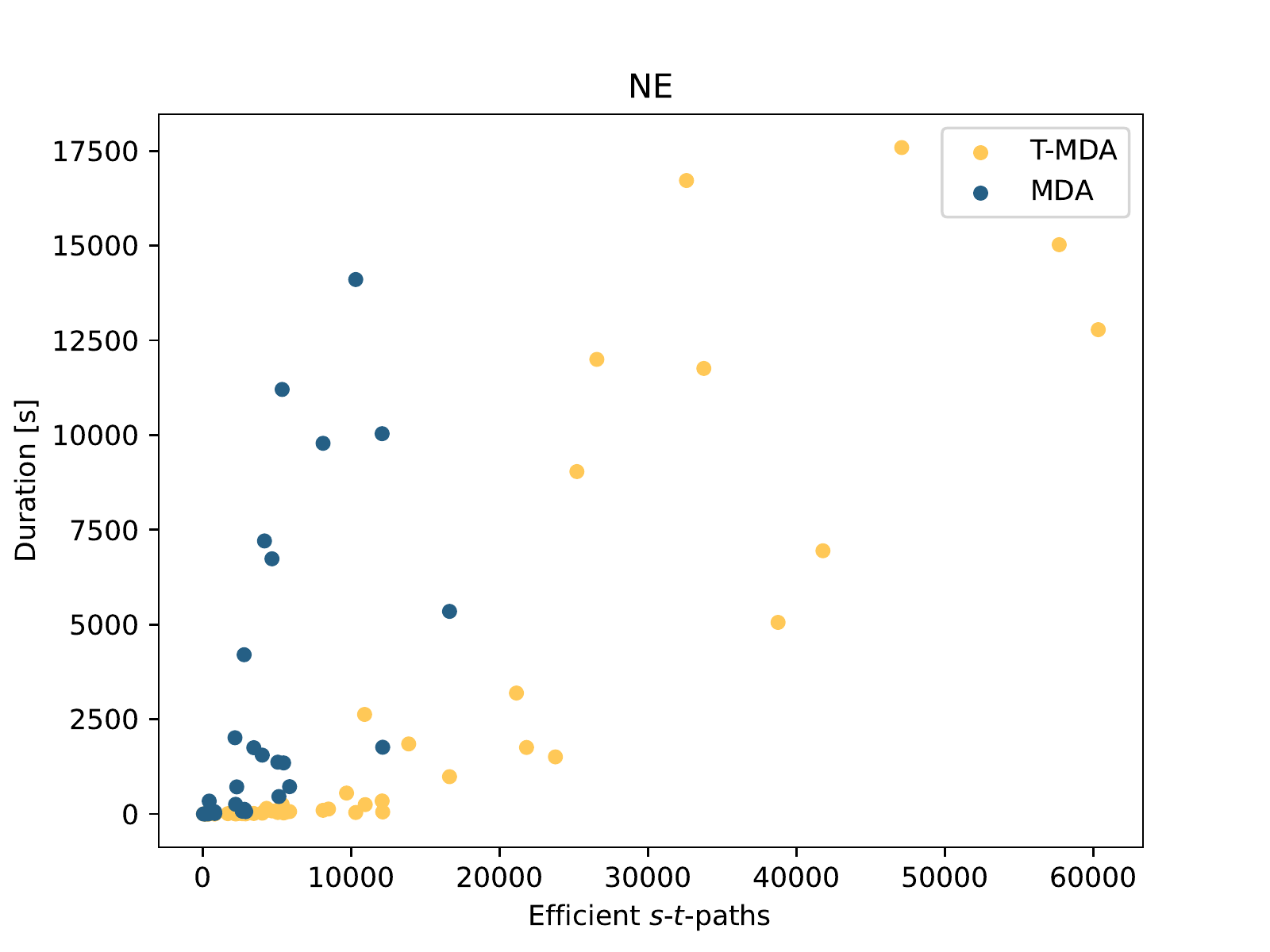}
				\captionof{figure}{MDA vs. \mda{} on NE network.}\label{fig:3d-road-NE}
			\end{minipage}
			\begin{minipage}{.48\linewidth}
				\captionsetup{type=figure}\includegraphics[width=\textwidth]{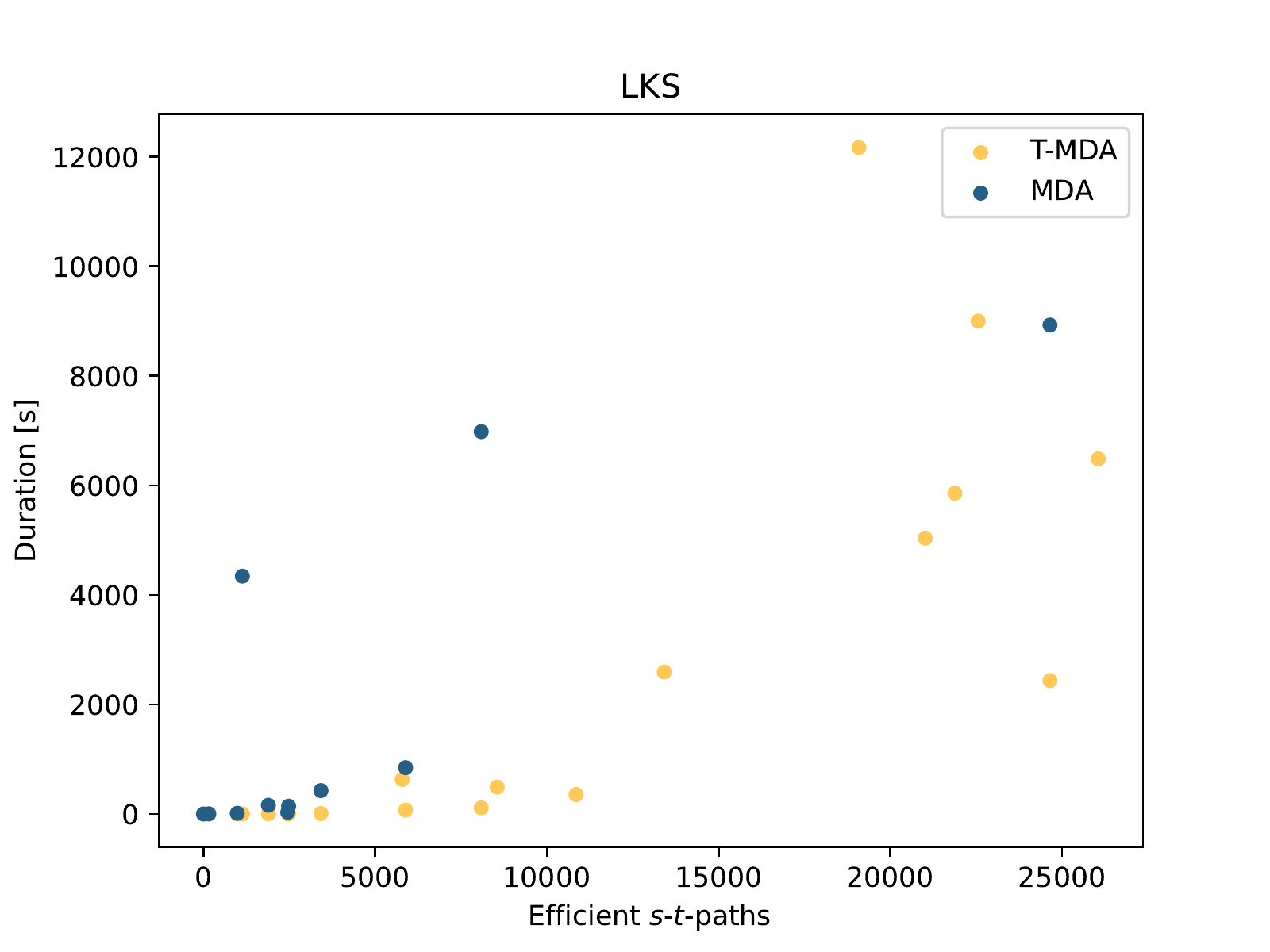}
				\captionof{figure}{MDA vs. \mda{} on LKS network.}\label{fig:3d-road-LKS}
			\end{minipage}
		\end{figure}
		\begin{figure}[H]
			\begin{minipage}{.48\linewidth}
				\captionsetup{type=figure}\includegraphics[width=\textwidth]{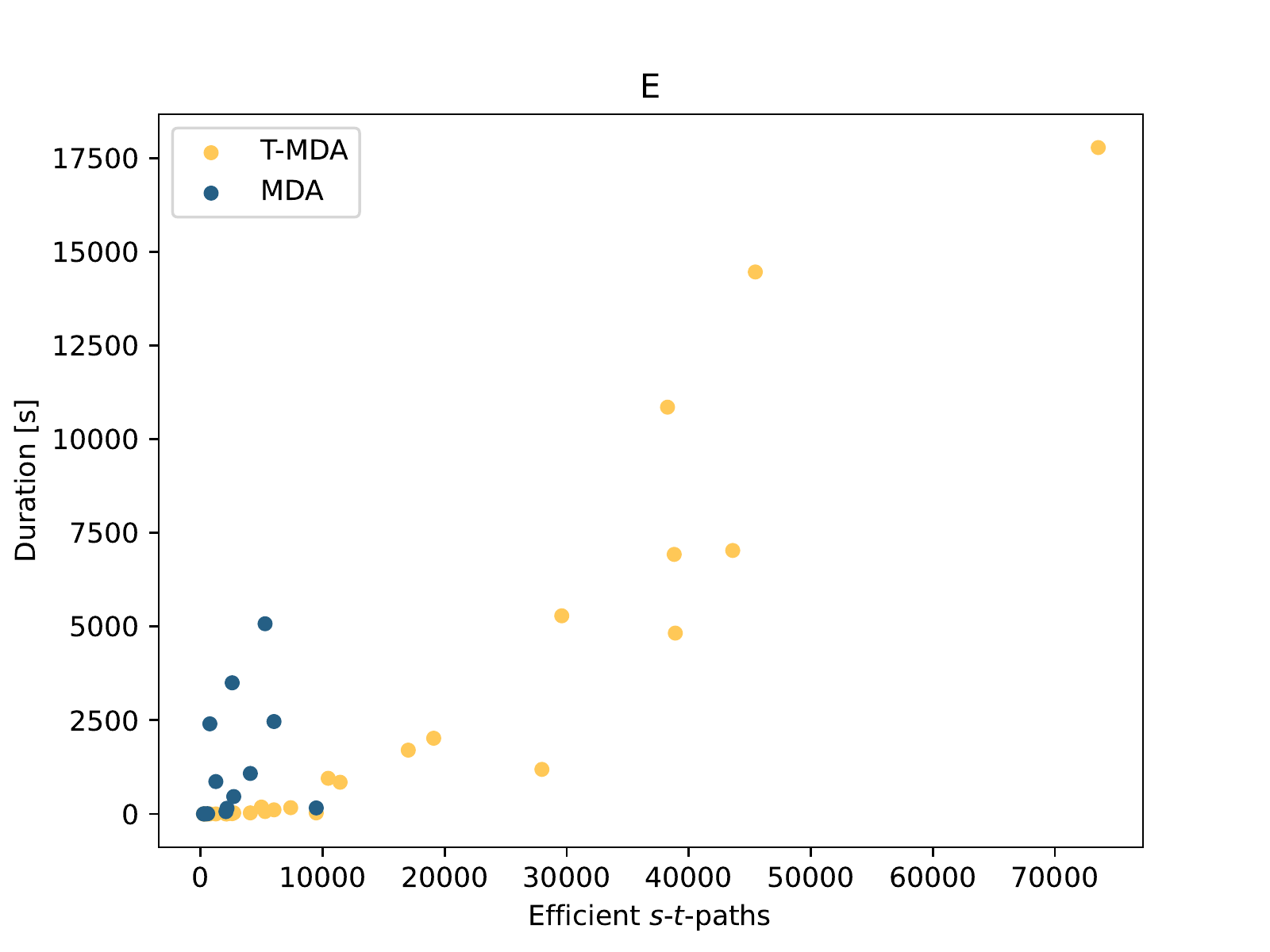}
				\captionof{figure}{MDA vs. \mda{} on E network.}\label{fig:3d-road-E}
			\end{minipage}
			\begin{minipage}{.48\linewidth}
				\captionsetup{type=figure}\includegraphics[width=\textwidth]{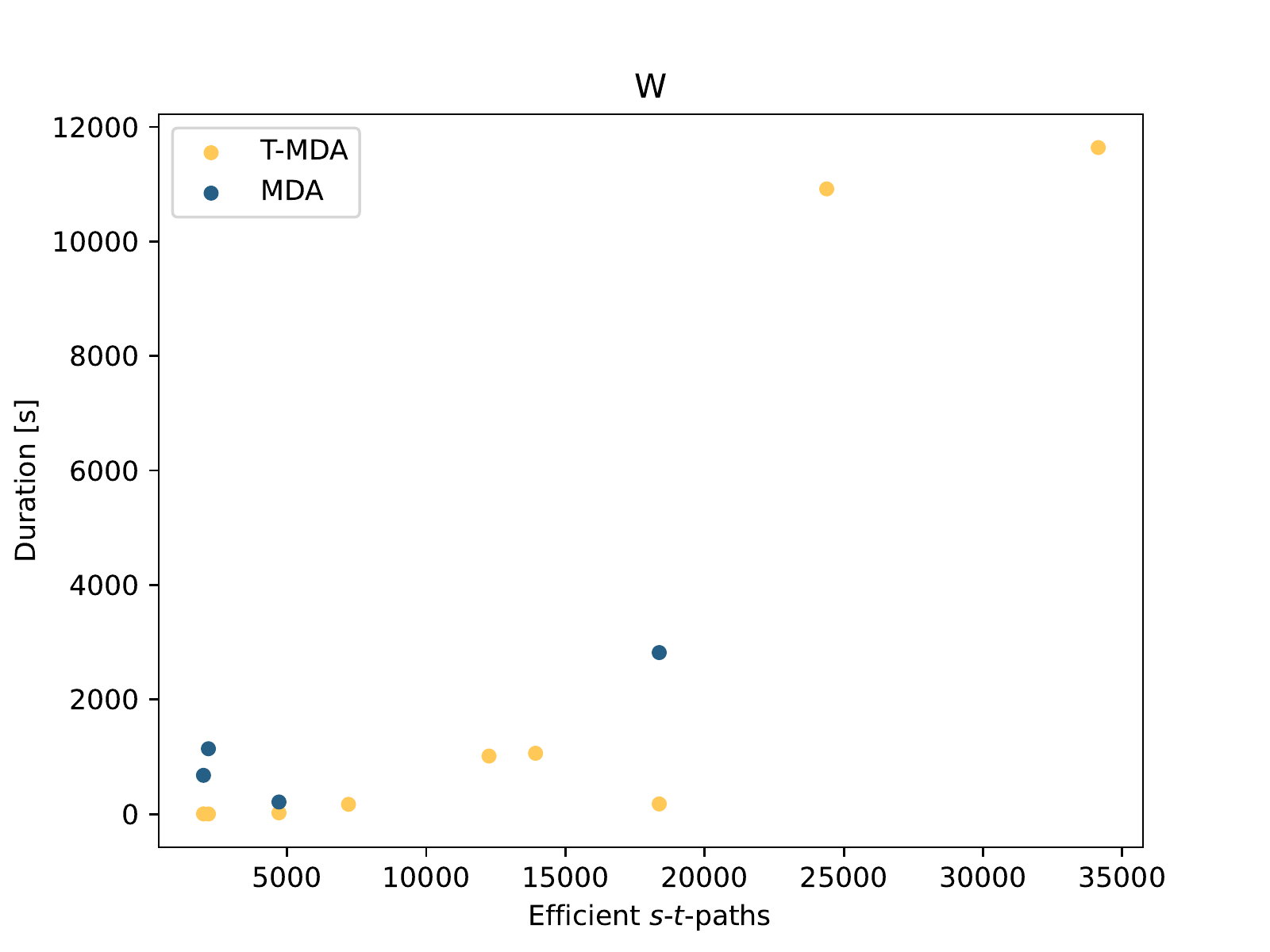}
				\captionof{figure}{MDA vs. \mda{} on W network.}\label{fig:3d-road-W}
			\end{minipage}
		\end{figure}
		\begin{figure}[H]
			\begin{minipage}{.48\linewidth}
				\captionsetup{type=figure}\includegraphics[width=\textwidth]{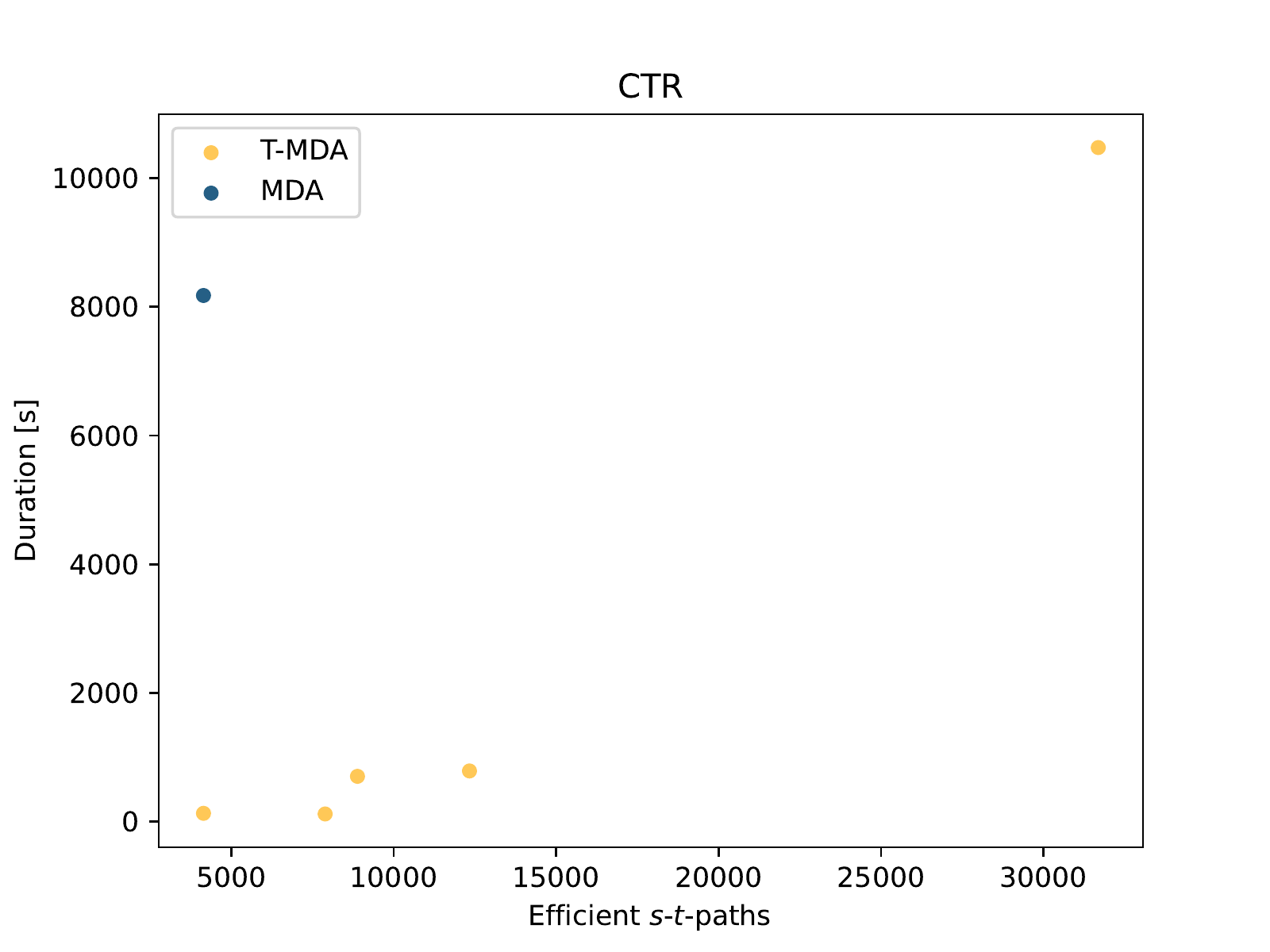}
				\captionof{figure}{MDA vs. \mda{} on CTR network.}\label{fig:3d-road-CTR}
			\end{minipage}
		\end{figure}
	\newpage
	\section{Biobjective Experiments}
	\subsection{Tables}
\begin{table}[H]
	\centering\small
	\caption{Results BT-BDA vs. \boaBidi{} using \textcolor{blue}{bucket queues} on road networks. For every graph, the results of the $100$ instances are reported after dividing the instances into disjoint time intervals based on the time it took the BT-BDA to solve them. We use rounded geometric means.}\label{tab:bi-bi-bucket}
	\begin{tabular}{ll rr rr rr r}
		\toprule
		& \begin{tabular}[c]{@{}l@{}}BT-BDA sol.\\time in $[s]$\end{tabular}  & \begin{tabular}[c]{@{}l@{}}Instance\\count\end{tabular} &   $N_t$      & \multicolumn{2}{c}{BT-BDA} & \multicolumn{2}{c}{\boaBidi{}}  & Speedup \\
		\midrule
		&       &    &           & Inserted   &  Time     & Inserted  & Time      &           \\
		\cmidrule(lr){5-6}          \cmidrule(lr){7-8}
\multirow{2}{*}{BAY} &(0, 0.5] & 99 & 62 & 40625 & 0.02 & 49943 & 0.02 & 0.75 \\
& (0.5, 5] & 1 & 655 & 4618916 & 0.54 & 5351638 & 0.56 & 1.04 \\
\midrule
\multirow{2}{*}{COL} &(0, 0.5] & 91 & 83 & 47464 & 0.03 & 53794 & 0.02 & 0.73 \\
& (0.5, 5] & 9 & 1526 & 10694947 & 1.29 & 12840144 & 1.43 & 1.11 \\
\midrule
\multirow{5}{*}{CTR} &(0, 0.5] & 4 & 105 & 103864 & 0.42 & 107033 & 0.03 & 0.07 \\
& (0.5, 5] & 21 & 1361 & 6528726 & 1.56 & 7664263 & 1.13 & 0.72 \\
& (5, 50] & 23 & 3438 & 73107725 & 17.64 & 85795476 & 18.46 & 1.05 \\
& (50, 500] & 31 & 8648 & 454823235 & 146.81 & 522860049 & 153.02 & 1.04 \\
& (500, 7200] & 21 & 21483 & 2820469600 & 1081.09 & 3254313265 & 1145.45 & 1.06 \\
\midrule
\multirow{5}{*}{E} &(0, 0.5] & 31 & 181 & 210353 & 0.16 & 252177 & 0.05 & 0.30 \\
& (0.5, 5] & 14 & 1497 & 12341423 & 1.84 & 14470630 & 1.87 & 1.02 \\
& (5, 50] & 20 & 2691 & 56652148 & 11.18 & 67838991 & 12.57 & 1.12 \\
& (50, 500] & 28 & 6518 & 427291715 & 124.83 & 515680451 & 131.48 & 1.05 \\
& (500, 7200] & 7 & 14228 & 2113666222 & 735.75 & 2521366438 & 778.47 & 1.06 \\
\midrule
\multirow{3}{*}{FLA} &(0, 0.5] & 74 & 133 & 107532 & 0.06 & 125394 & 0.03 & 0.46 \\
& (0.5, 5] & 19 & 1214 & 9175294 & 1.21 & 10818218 & 1.33 & 1.10 \\
& (5, 50] & 7 & 2582 & 53748512 & 9.80 & 63662475 & 11.53 & 1.18 \\
\midrule
\multirow{5}{*}{LKS} &(0, 0.5] & 28 & 184 & 168986 & 0.13 & 199154 & 0.05 & 0.34 \\
& (0.5, 5] & 20 & 1595 & 11259390 & 1.62 & 13396943 & 1.70 & 1.05 \\
& (5, 50] & 22 & 3067 & 83435641 & 18.22 & 101859555 & 20.11 & 1.10 \\
& (50, 500] & 25 & 9562 & 581410031 & 173.10 & 701058691 & 186.76 & 1.08 \\
& (500, 7200] & 5 & 15199 & 1927723557 & 644.00 & 2271500621 & 698.51 & 1.08 \\
\midrule
\multirow{3}{*}{NE} &(0, 0.5] & 49 & 165 & 185289 & 0.09 & 227362 & 0.04 & 0.46 \\
& (0.5, 5] & 40 & 1035 & 10336396 & 1.58 & 12779331 & 1.71 & 1.08 \\
& (5, 50] & 11 & 2138 & 58734135 & 13.27 & 72864414 & 14.23 & 1.07 \\
\midrule
\multirow{1}{*}{NY} &(0, 0.5] & 100 & 57 & 24230 & 0.02 & 30239 & 0.01 & 0.63 \\
\midrule
\multirow{5}{*}{W} &(0, 0.5] & 16 & 278 & 482099 & 0.26 & 592757 & 0.10 & 0.38 \\
& (0.5, 5] & 29 & 1632 & 11196884 & 1.66 & 13207940 & 1.69 & 1.02 \\
& (5, 50] & 25 & 4009 & 80021210 & 15.03 & 93615464 & 17.28 & 1.15 \\
& (50, 500] & 24 & 10441 & 489914196 & 128.28 & 564094695 & 138.58 & 1.08 \\
& (500, 7200] & 6 & 16522 & 2202981245 & 753.26 & 2540792665 & 808.41 & 1.07 \\
\midrule
		\bottomrule
	\end{tabular}
\end{table}

\begin{table}
	\centering\small
	\caption{Results \bda{} vs. \boaEnh{} using \textcolor{blue}{bucket queues} on road networks. For every graph, the results of the $100$ instances are reported after dividing the instances into disjoint time intervals based on the time it took the BT-BDA to solve them. We use rounded geometric means.}\label{tab:bi-uni-bucket}
	\begin{tabular}{ll rr rr rr r}
	\toprule
	& \begin{tabular}[c]{@{}l@{}}\bda{} sol.\\time in $[s]$\end{tabular}  & \begin{tabular}[c]{@{}l@{}}Instance\\count\end{tabular} &   $N_t$      & \multicolumn{2}{c}{\bda{}} & \multicolumn{2}{c}{\boaEnh{}}  & Speedup \\
	\midrule
	&       &    &           & Inserted   &  Time     & Inserted  & Time      &           \\
	\cmidrule(lr){5-6}          \cmidrule(lr){7-8}
\multirow{2}{*}{BAY} &(0, 0.5] & 95 & 57 & 27703 & 0.02 & 33868 & 0.02 & 0.69 \\
& (0.5, 5] & 5 & 642 & 4012946 & 0.83 & 4844621 & 0.75 & 0.90 \\
\midrule
\multirow{3}{*}{COL} &(0, 0.5] & 84 & 71 & 30264 & 0.03 & 36172 & 0.02 & 0.67 \\
& (0.5, 5] & 14 & 934 & 5361484 & 1.09 & 6613358 & 0.98 & 0.90 \\
& (5, 50] & 2 & 1229 & 32256753 & 7.63 & 39730169 & 7.38 & 0.97 \\
\midrule
\multirow{5}{*}{CTR} &(0, 0.5] & 4 & 105 & 67595 & 0.41 & 79430 & 0.03 & 0.08 \\
& (0.5, 5] & 18 & 1251 & 4889049 & 1.80 & 6055539 & 1.30 & 0.72 \\
& (5, 50] & 23 & 3164 & 50325472 & 18.04 & 61125858 & 17.97 & 1.00 \\
& (50, 500] & 28 & 7351 & 303654790 & 153.83 & 364040207 & 155.71 & 1.01 \\
& (500, 7200] & 27 & 19198 & 1908034860 & 1297.36 & 1988914831 & 1317.76 & 1.02 \\
\midrule
\multirow{5}{*}{E} &(0, 0.5] & 27 & 146 & 113022 & 0.15 & 135219 & 0.04 & 0.25 \\
& (0.5, 5] & 15 & 1121 & 6919429 & 1.73 & 8628146 & 1.49 & 0.86 \\
& (5, 50] & 22 & 2681 & 41868328 & 11.63 & 51221472 & 12.13 & 1.04 \\
& (50, 500] & 27 & 6054 & 316234376 & 145.49 & 398605061 & 151.19 & 1.04 \\
& (500, 7200] & 9 & 13640 & 1563191554 & 955.03 & 1956197401 & 993.52 & 1.04 \\
\midrule
\multirow{3}{*}{FLA} &(0, 0.5] & 70 & 121 & 75779 & 0.07 & 94182 & 0.03 & 0.43 \\
& (0.5, 5] & 22 & 1091 & 6908224 & 1.48 & 8411956 & 1.34 & 0.91 \\
& (5, 50] & 8 & 2383 & 45759407 & 11.65 & 55236464 & 12.45 & 1.07 \\
\midrule
\multirow{5}{*}{LKS} &(0, 0.5] & 24 & 148 & 90095 & 0.12 & 110841 & 0.04 & 0.29 \\
& (0.5, 5] & 19 & 1255 & 7579029 & 1.81 & 9435583 & 1.63 & 0.90 \\
& (5, 50] & 27 & 2826 & 53244055 & 16.58 & 66781137 & 17.39 & 1.05 \\
& (50, 500] & 21 & 9022 & 371893218 & 167.29 & 461995930 & 174.39 & 1.04 \\
& (500, 7200] & 9 & 14168 & 1452305875 & 841.46 & 1779308464 & 879.40 & 1.05 \\
\midrule
\multirow{3}{*}{NE} &(0, 0.5] & 45 & 148 & 131276 & 0.09 & 162577 & 0.04 & 0.42 \\
& (0.5, 5] & 40 & 896 & 7107872 & 1.71 & 9173581 & 1.59 & 0.93 \\
& (5, 50] & 15 & 2193 & 40280288 & 12.84 & 51531456 & 13.09 & 1.02 \\
\midrule
\multirow{2}{*}{NY} &(0, 0.5] & 98 & 55 & 18877 & 0.02 & 24046 & 0.01 & 0.58 \\
& (0.5, 5] & 2 & 573 & 2496976 & 0.58 & 3346304 & 0.47 & 0.81 \\
\midrule
\multirow{5}{*}{W} &(0, 0.5] & 14 & 247 & 338636 & 0.25 & 425815 & 0.09 & 0.37 \\
& (0.5, 5] & 27 & 1424 & 7713009 & 1.84 & 9536562 & 1.62 & 0.88 \\
& (5, 50] & 30 & 3866 & 61677780 & 16.85 & 74762048 & 17.51 & 1.04 \\
& (50, 500] & 21 & 10192 & 410743130 & 147.98 & 488473926 & 159.61 & 1.08 \\
& (500, 7200] & 8 & 16245 & 1567562508 & 842.46 & 1852468610 & 860.99 & 1.02 \\
		\bottomrule
	\end{tabular}
\end{table}

\begin{table}
	\centering\small
	\caption{Results BT-BDA vs. \boaBidi{} using \textcolor{blue}{binary heaps} on road networks. For every graph, the results of the $100$ instances are reported after dividing the instances into disjoint time intervals based on the time it took the BT-BDA to solve them. We use rounded geometric means.}\label{tab:bi-bi-heap}
	\begin{tabular}{ll r rrr rrr r}
		\toprule
		& \begin{tabular}[c]{@{}l@{}}BT-BDA sol.\\time in $[s]$\end{tabular}  &   $N_t$      & \multicolumn{3}{c}{BT-BDA} & \multicolumn{3}{c}{BOB$\text{A}^*$}  & Speedup \\
		\midrule
		&       &               & Solved &    Inserted   &  Time     & Solved & Inserted  & Time      &           \\
		\cmidrule(lr){4-6}          \cmidrule(lr){7-9}
\multirow{2}{*}{BAY} & (0, 0.5] & 60 & 97/97 & 37714 & 0.02 & 97/97 & 43735 & 0.01 & 0.54 \\
& (0.5, 5] & 632 & 3/3 & 3830848 & 0.68 & 3/3 & 4395194 & 0.74 & 1.09 \\
\midrule
\multirow{3}{*}{COL} & (0, 0.5] & 78 & 88/88 & 39187 & 0.02 & 88/88 & 43494 & 0.01 & 0.49 \\
& (0.5, 5] & 1240 & 10/10 & 6377844 & 1.12 & 10/10 & 7628297 & 1.39 & 1.24 \\
& (5, 50] & 1229 & 2/2 & 26245555 & 5.96 & 2/2 & 31313672 & 8.40 & 1.41 \\
\midrule
\multirow{5}{*}{CTR} & (0, 0.5] & 180 & 5/5 & 164244 & 0.33 & 5/5 & 213097 & 0.05 & 0.15 \\
& (0.5, 5] & 1183 & 17/17 & 5579215 & 1.62 & 17/17 & 6648589 & 1.60 & 0.99 \\
& (5, 50] & 3073 & 20/20 & 49779921 & 17.13 & 20/20 & 57364789 & 24.43 & 1.43 \\
& (50, 500] & 6869 & 31/31 & 308476360 & 147.63 & 31/31 & 355402214 & 235.18 & 1.59 \\
& (500, 7200] & 19835 & 27/27 & 2275789128 & 1509.65 & 25/27 & - & 2291.72 & 1.52 \\
\midrule
\multirow{5}{*}{E} & (0, 0.5] & 152 & 28/28 & 160197 & 0.12 & 28/28 & 185554 & 0.03 & 0.28 \\
& (0.5, 5] & 1326 & 13/13 & 7293211 & 1.49 & 13/13 & 8380494 & 1.77 & 1.18 \\
& (5, 50] & 2270 & 21/21 & 42123470 & 11.85 & 21/21 & 49626258 & 19.46 & 1.64 \\
& (50, 500] & 6037 & 28/28 & 342207810 & 157.68 & 28/28 & 410266755 & 281.83 & 1.79 \\
& (500, 7200] & 12307 & 10/10 & 1720688747 & 1052.32 & 10/10 & 2035247152 & 1721.61 & 1.64 \\
\midrule
\multirow{3}{*}{FLA} & (0, 0.5] & 123 & 71/71 & 91592 & 0.06 & 71/71 & 104436 & 0.03 & 0.44 \\
& (0.5, 5] & 1161 & 20/20 & 7495450 & 1.50 & 20/20 & 8643346 & 1.99 & 1.32 \\
& (5, 50] & 2183 & 9/9 & 41251739 & 13.62 & 9/9 & 49092614 & 19.71 & 1.45 \\
\midrule
\multirow{5}{*}{LKS} & (0, 0.5] & 184 & 28/28 & 163527 & 0.13 & 28/28 & 194267 & 0.04 & 0.34 \\
& (0.5, 5] & 1509 & 16/16 & 9374568 & 2.07 & 16/16 & 11055888 & 2.56 & 1.24 \\
& (5, 50] & 2749 & 24/24 & 62695082 & 20.54 & 24/24 & 76806976 & 33.87 & 1.65 \\
& (50, 500] & 7885 & 20/20 & 405592558 & 185.01 & 20/20 & 490492841 & 339.63 & 1.84 \\
& (500, 7200] & 14254 & 12/12 & 1398550016 & 762.38 & 12/12 & 1666695527 & 1308.45 & 1.72 \\
\midrule
\multirow{4}{*}{NE} & (0, 0.5] & 152 & 46/46 & 153034 & 0.08 & 46/46 & 184880 & 0.03 & 0.41 \\
& (0.5, 5] & 853 & 34/34 & 7738980 & 1.92 & 34/34 & 9596334 & 2.39 & 1.25 \\
& (5, 50] & 1954 & 19/19 & 34172231 & 11.90 & 19/19 & 41450095 & 18.14 & 1.52 \\
& (50, 500] & 2073 & 1/1 & 132356676 & 51.02 & 1/1 & 166036909 & 94.70 & 1.86 \\
\midrule
\multirow{1}{*}{NY} & (0, 0.5] & 57 & 100/100 & 23611 & 0.02 & 100/100 & 28374 & 0.01 & 0.54 \\
\midrule
\multirow{5}{*}{W} & (0, 0.5] & 227 & 13/13 & 333907 & 0.20 & 13/13 & 379780 & 0.06 & 0.31 \\
& (0.5, 5] & 1325 & 27/27 & 7956444 & 1.76 & 27/27 & 9406443 & 2.23 & 1.27 \\
& (5, 50] & 3584 & 28/28 & 61540531 & 17.15 & 28/28 & 71560273 & 28.09 & 1.64 \\
& (50, 500] & 9806 & 24/24 & 416197917 & 164.66 & 24/24 & 482695661 & 321.44 & 1.95 \\
& (500, 7200] & 17113 & 8/8 & 1915274474 & 1149.59 & 8/8 & 2231672678 & 1826.39 & 1.59 \\
\bottomrule

	\end{tabular}
\end{table}

\begin{table}
	\centering\small
	\caption{Results \bda{} vs. \boaEnh{} using \textcolor{blue}{binary heaps} on road networks. For every graph, the results of the $100$ instances are reported after dividing the instances into disjoint time intervals based on the time it took the \bda{} to solve them. We use rounded geometric means.}\label{tab:bi-uni-heap}
	\begin{tabular}{ll r rrr rrr r}
		\toprule
		& \begin{tabular}[c]{@{}l@{}}\bda{} sol.\\time in $[s]$\end{tabular}  &   $N_t$      & \multicolumn{3}{c}{\bda{}} & \multicolumn{3}{c}{\boaEnh{}}  & Speedup \\
		\midrule
		&       &               & Solved &    Inserted   &  Time     & Solved & Inserted  & Time      &           \\
		\cmidrule(lr){4-6}          \cmidrule(lr){7-9}
		\multirow{2}{*}{BAY} & (0, 0.5] & 53 & 92/92 & 24134 & 0.02 & 92/92 & 29612 & 0.01 & 0.57 \\
		& (0.5, 5] & 544 & 8/8 & 2984841 & 1.01 & 8/8 & 3622751 & 1.14 & 1.13 \\
		\midrule
		\multirow{3}{*}{COL} & (0, 0.5] & 66 & 82/82 & 27128 & 0.02 & 82/82 & 32586 & 0.01 & 0.53 \\
		& (0.5, 5] & 1008 & 16/16 & 4830522 & 1.52 & 16/16 & 5996713 & 1.98 & 1.30 \\
		& (5, 50] & 1229 & 2/2 & 32212779 & 12.82 & 2/2 & 39852849 & 18.30 & 1.43 \\
		\midrule
		\multirow{5}{*}{CTR} & (0, 0.5] & 105 & 4/4 & 67539 & 0.30 & 4/4 & 79856 & 0.03 & 0.11 \\
		& (0.5, 5] & 1029 & 13/13 & 3417148 & 1.61 & 13/13 & 4330385 & 1.72 & 1.06 \\
		& (5, 50] & 2566 & 21/21 & 29020798 & 15.02 & 21/21 & 35617705 & 23.36 & 1.56 \\
		& (50, 500] & 6001 & 30/30 & 200757115 & 148.93 & 30/30 & 242724359 & 281.38 & 1.89 \\
		& (500, 7200] & 17860 & 32/32 & 1632045110 & 1832.03 & 28/32 & - & 3013.59 & 1.64 \\
		\midrule
		\multirow{5}{*}{E} & (0, 0.5] & 133 & 25/25 & 91225 & 0.12 & 25/25 & 110015 & 0.05 & 0.42 \\
		& (0.5, 5] & 849 & 12/12 & 3770526 & 1.42 & 12/12 & 4770631 & 2.01 & 1.42 \\
		& (5, 50] & 2320 & 25/25 & 32160063 & 13.98 & 25/25 & 39606712 & 29.48 & 2.11 \\
		& (50, 500] & 5464 & 24/24 & 245243763 & 159.53 & 24/24 & 308196779 & 385.59 & 2.42 \\
		& (500, 7200] & 11491 & 14/14 & 1165367046 & 1083.52 & 14/14 & 1135431594 & 2208.60 & 2.04 \\
		\midrule
		\multirow{4}{*}{FLA} & (0, 0.5] & 118 & 69/69 & 71701 & 0.06 & 69/69 & 90697 & 0.03 & 0.48 \\
		& (0.5, 5] & 949 & 21/21 & 5966199 & 2.05 & 21/21 & 7392376 & 2.68 & 1.30 \\
		& (5, 50] & 2507 & 9/9 & 33198809 & 13.53 & 9/9 & 40329084 & 22.08 & 1.63 \\
		& (50, 500] & 2999 & 1/1 & 106379213 & 56.57 & 1/1 & 129750546 & 127.05 & 2.25 \\
		\midrule
		\multirow{5}{*}{LKS} & (0, 0.5] & 113 & 21/21 & 57393 & 0.09 & 21/21 & 71294 & 0.03 & 0.29 \\
		& (0.5, 5] & 1010 & 14/14 & 3914124 & 1.35 & 14/14 & 4901558 & 1.69 & 1.25 \\
		& (5, 50] & 2347 & 29/29 & 32055480 & 14.67 & 29/29 & 40380190 & 25.59 & 1.74 \\
		& (50, 500] & 6681 & 24/24 & 246563476 & 162.97 & 24/24 & 309467525 & 331.72 & 2.04 \\
		& (500, 7200] & 14272 & 12/12 & 1277655419 & 1104.40 & 12/12 & 1571844519 & 2171.67 & 1.97 \\
		\midrule
		\multirow{4}{*}{NE} & (0, 0.5] & 131 & 41/41 & 101825 & 0.07 & 41/41 & 126504 & 0.03 & 0.49 \\
		& (0.5, 5] & 769 & 35/35 & 4983924 & 1.89 & 35/35 & 6524521 & 2.89 & 1.53 \\
		& (5, 50] & 1771 & 23/23 & 26412133 & 12.16 & 23/23 & 33775376 & 21.60 & 1.78 \\
		& (50, 500] & 2073 & 1/1 & 84650339 & 54.69 & 1/1 & 111095228 & 101.84 & 1.86 \\
		\midrule
		\multirow{1}{*}{NY} & (0, 0.5] & 55 & 98/98 & 18834 & 0.02 & 98/98 & 24141 & 0.01 & 0.64 \\
		& (0.5, 5] & 573 & 2/2 & 2494212 & 0.88 & 2/2 & 3362606 & 1.09 & 1.25 \\
		\midrule
		\multirow{5}{*}{W} & (0, 0.5] & 224 & 12/12 & 265327 & 0.19 & 12/12 & 333929 & 0.08 & 0.40 \\
		& (0.5, 5] & 1152 & 22/22 & 5178536 & 1.84 & 22/22 & 6582838 & 2.40 & 1.30 \\
		& (5, 50] & 3006 & 26/26 & 31816717 & 12.67 & 26/26 & 38891925 & 21.34 & 1.68 \\
		& (50, 500] & 6923 & 28/28 & 231002964 & 131.42 & 28/28 & 278461075 & 269.83 & 2.05 \\
		& (500, 7200] & 16509 & 12/12 & 1292686430 & 1045.15 & 12/12 & 1169278518 & 2180.79 & 2.09 \\
		\bottomrule
	\end{tabular}
\end{table}

	\subsection{Graphics -- Bidimensional-Grids}
		\begin{figure}[H]
		\begin{minipage}{.48\linewidth}
			\captionsetup{type=figure}\includegraphics[width=\textwidth]{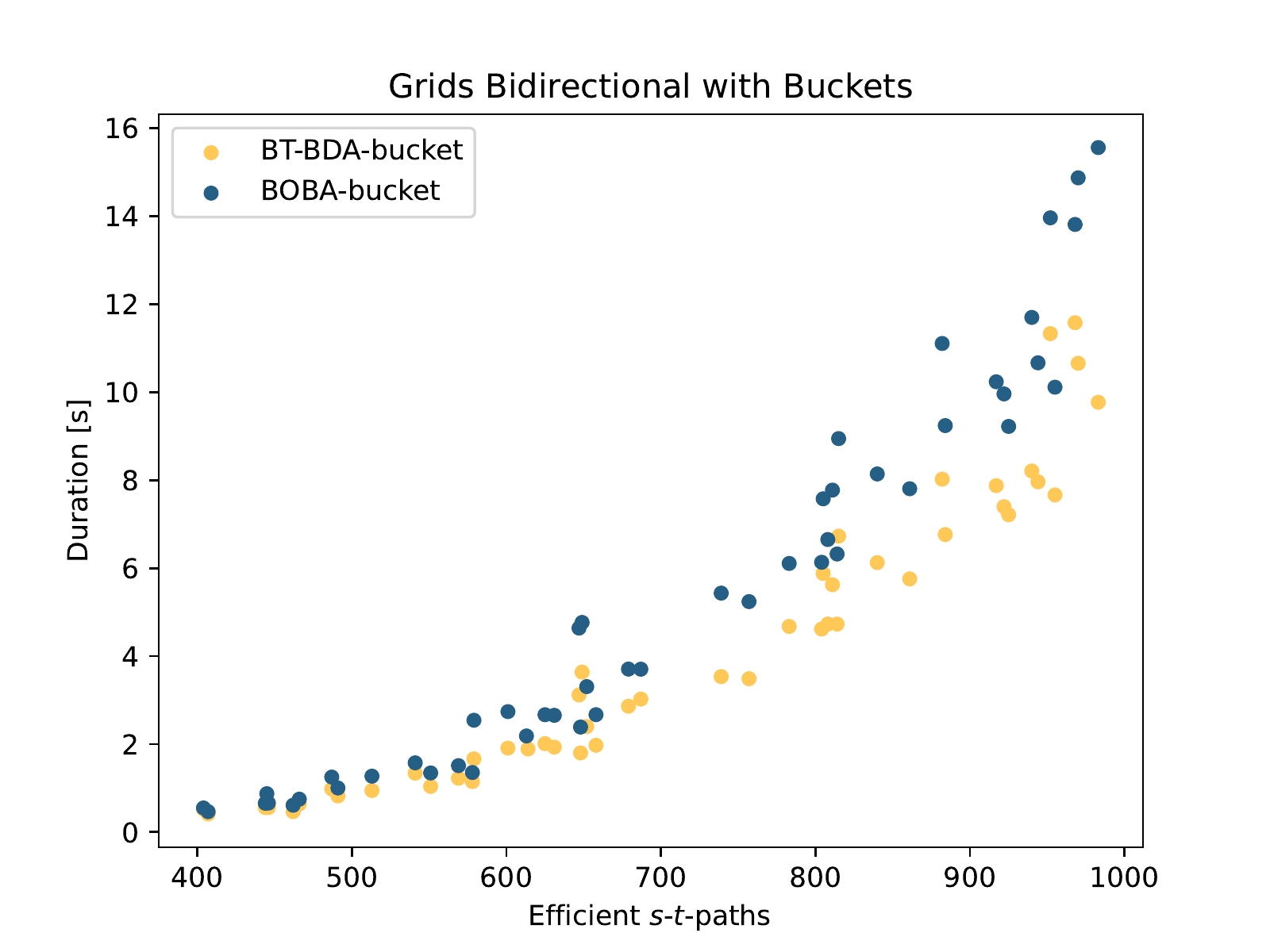}
			\captionof{figure}{}\label{fig:2d-grid-bi-bucket}
		\end{minipage}
		\begin{minipage}{.48\linewidth}
			\captionsetup{type=figure}\includegraphics[width=\textwidth]{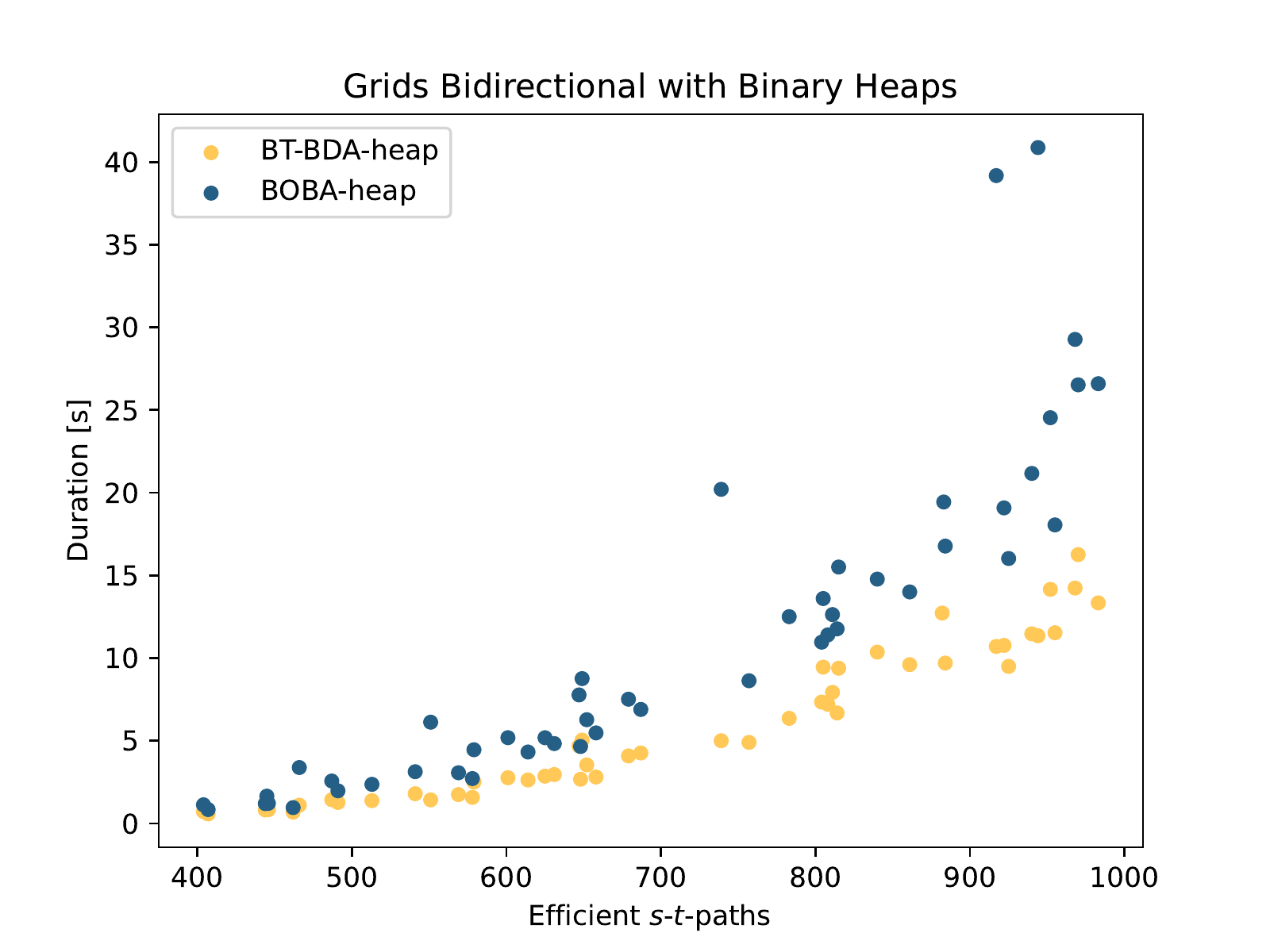}
			\captionof{figure}{}\label{fig:2d-grid-bi-heap}
		\end{minipage}
		\end{figure}
		\begin{figure}[H]
		\begin{minipage}{.48\linewidth}
			\captionsetup{type=figure}\includegraphics[width=\textwidth]{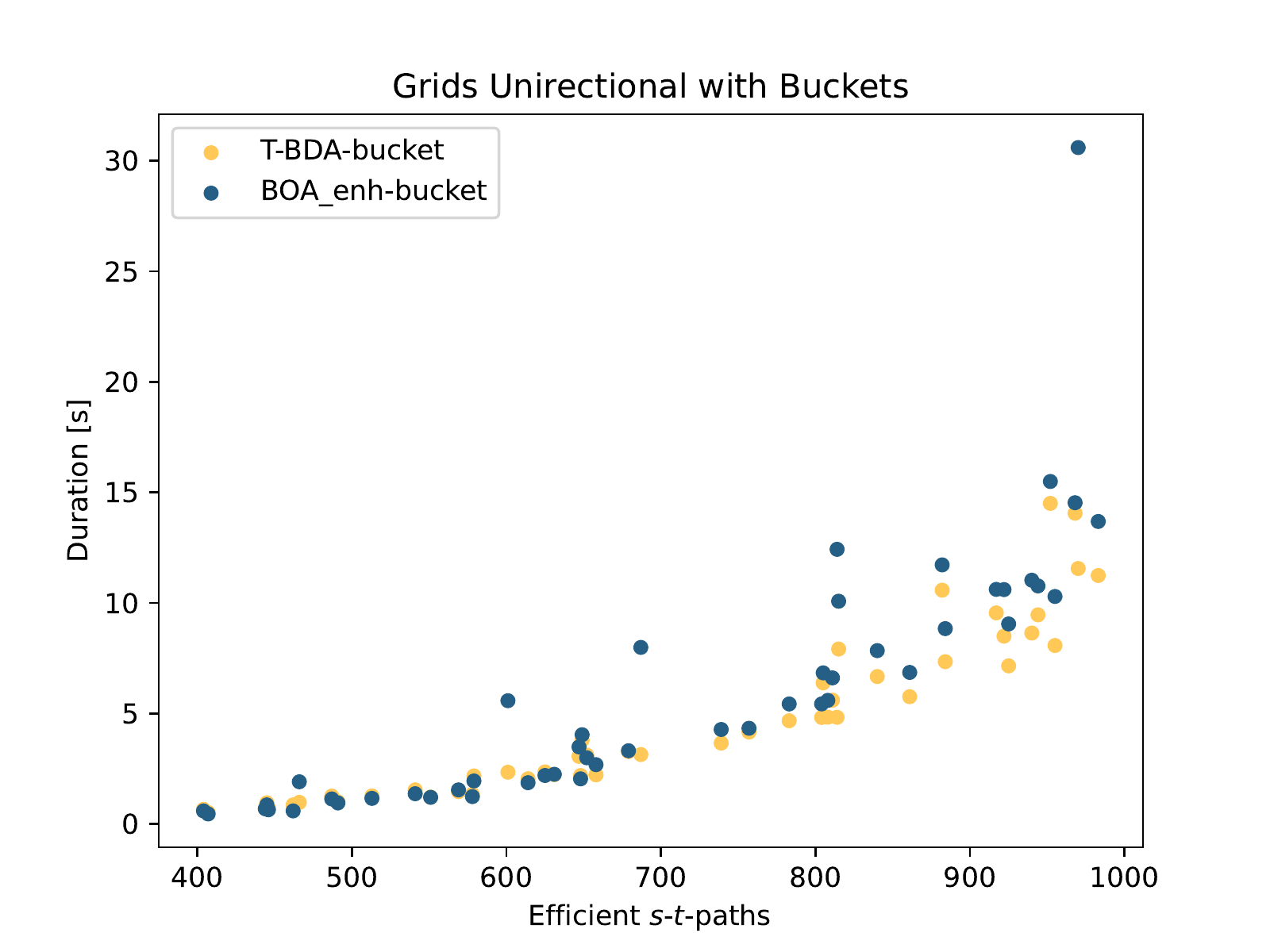}
			\captionof{figure}{}\label{fig:2d-grid-uni-bucket}
		\end{minipage}
		\begin{minipage}{.48\linewidth}
			\captionsetup{type=figure}\includegraphics[width=\textwidth]{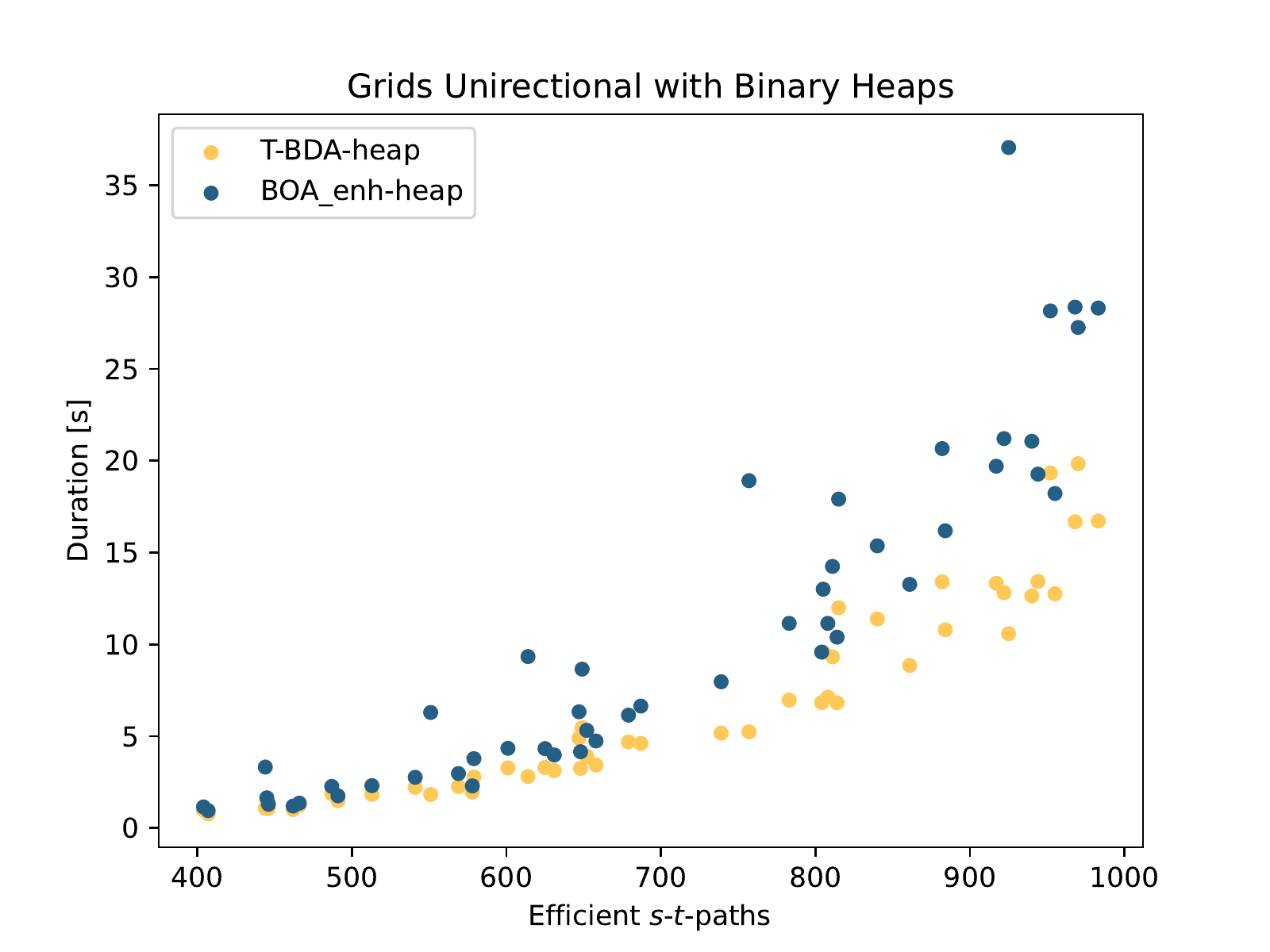}
			\captionof{figure}{}\label{fig:2d-grid-uni-heap}
		\end{minipage}
		\end{figure}

	\subsection{Graphics -- Bidimensional-Bidirectional-Buckets}
			\begin{figure}[H]
		\begin{minipage}{.48\linewidth}
			\captionsetup{type=figure}\includegraphics[width=\textwidth]{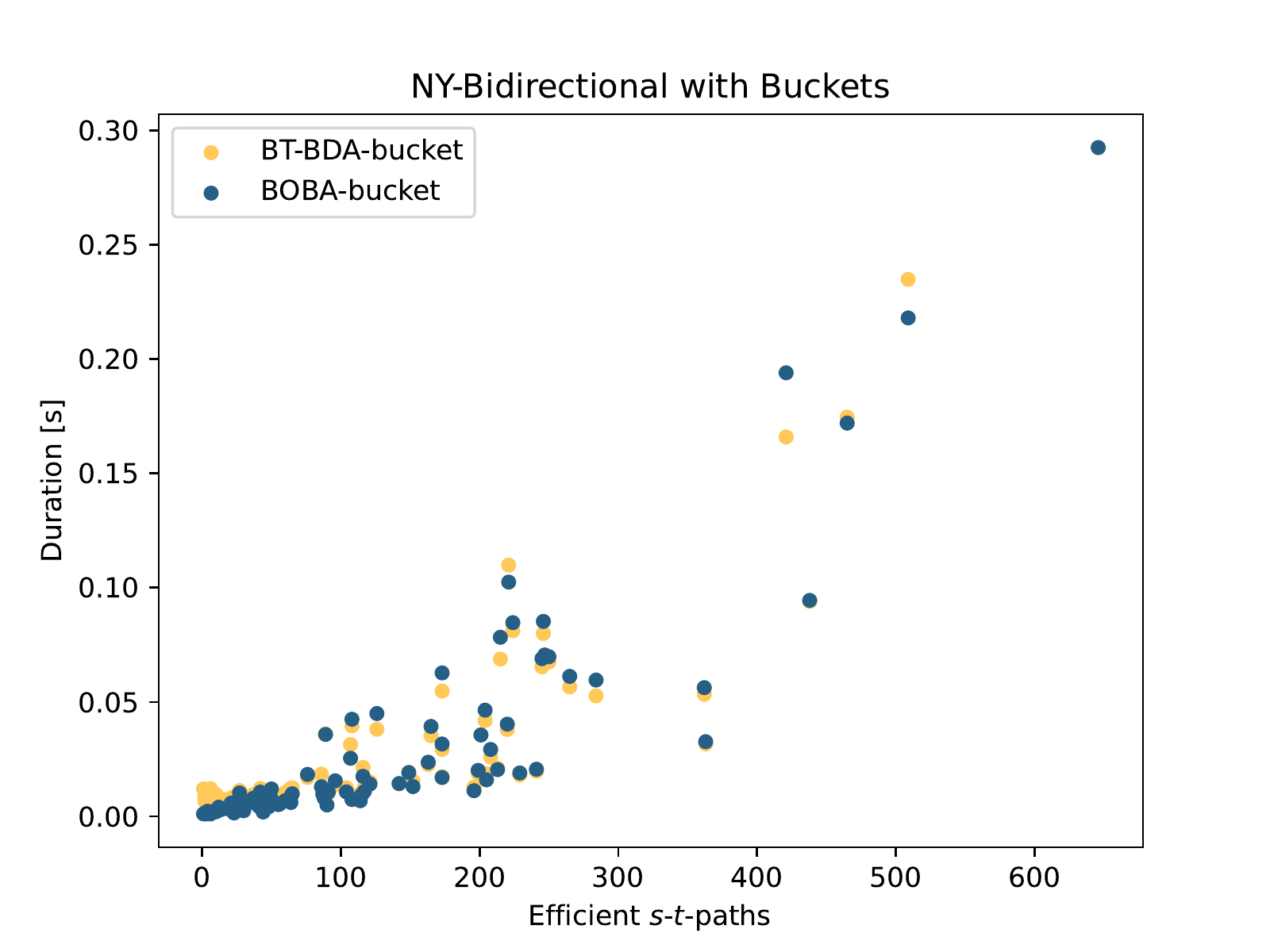}
			\captionof{figure}{}\label{fig:2d-bbda-boba-bucket-NY}
		\end{minipage}
		\begin{minipage}{.48\linewidth}
			\captionsetup{type=figure}\includegraphics[width=\textwidth]{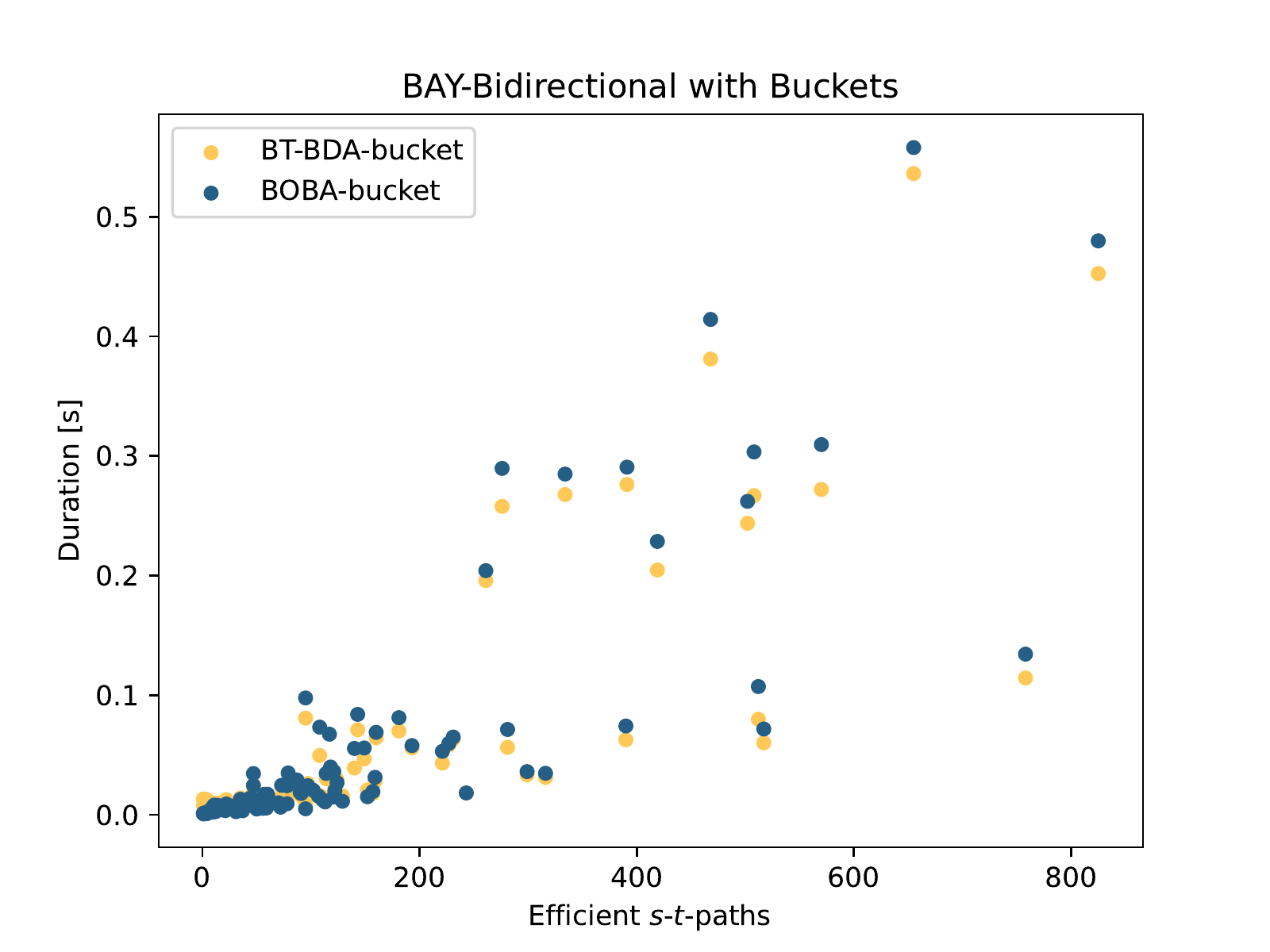}
			\captionof{figure}{}\label{fig:2d-bbda-boba-bucket-BAY}
		\end{minipage}
	\end{figure}
	\begin{figure}[H]
		\begin{minipage}{.48\linewidth}
			\captionsetup{type=figure}\includegraphics[width=\textwidth]{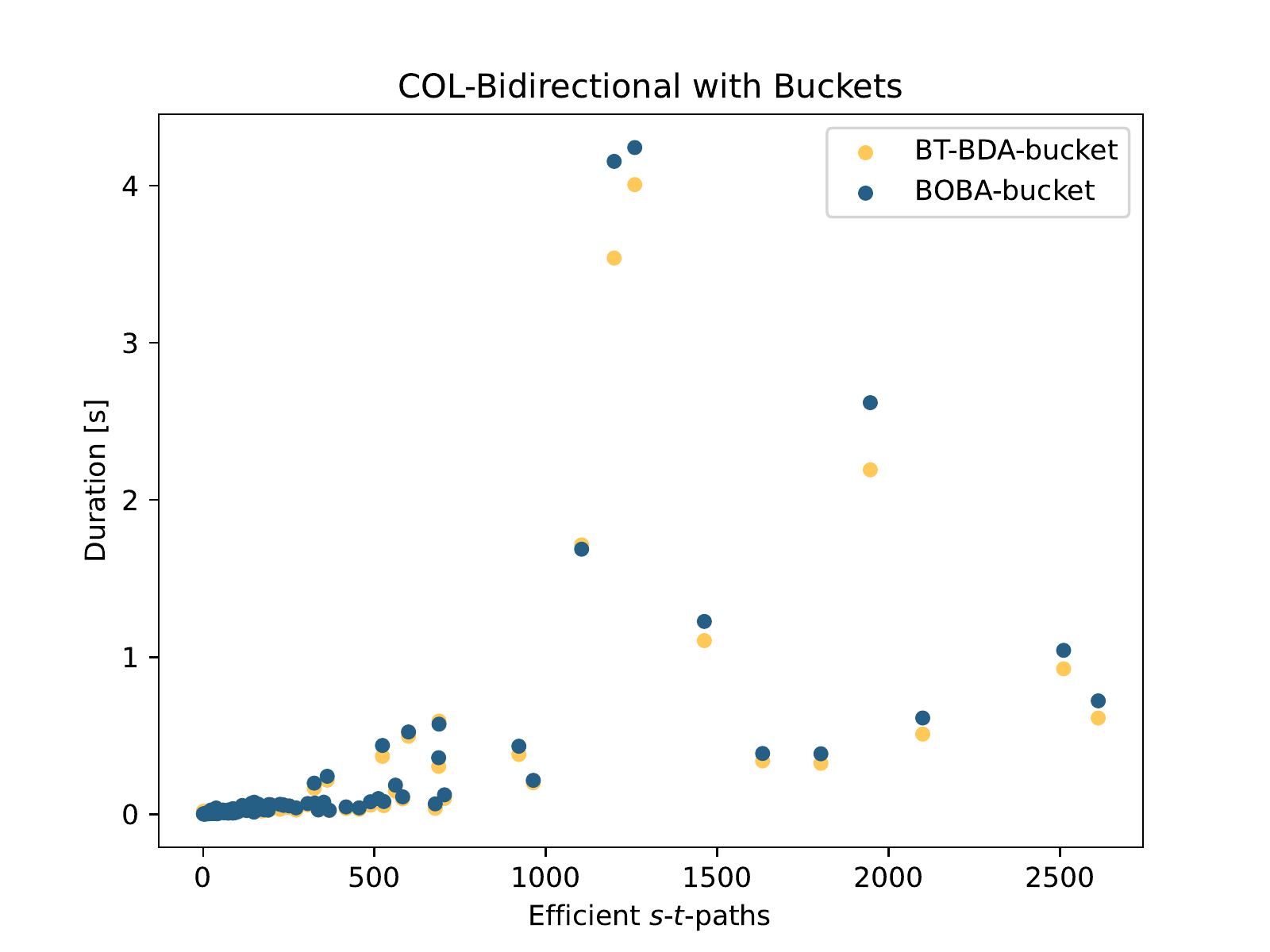}
			\captionof{figure}{}\label{fig:2d-bbda-boba-bucket-COL}
		\end{minipage}
		\begin{minipage}{.48\linewidth}
			\captionsetup{type=figure}\includegraphics[width=\textwidth]{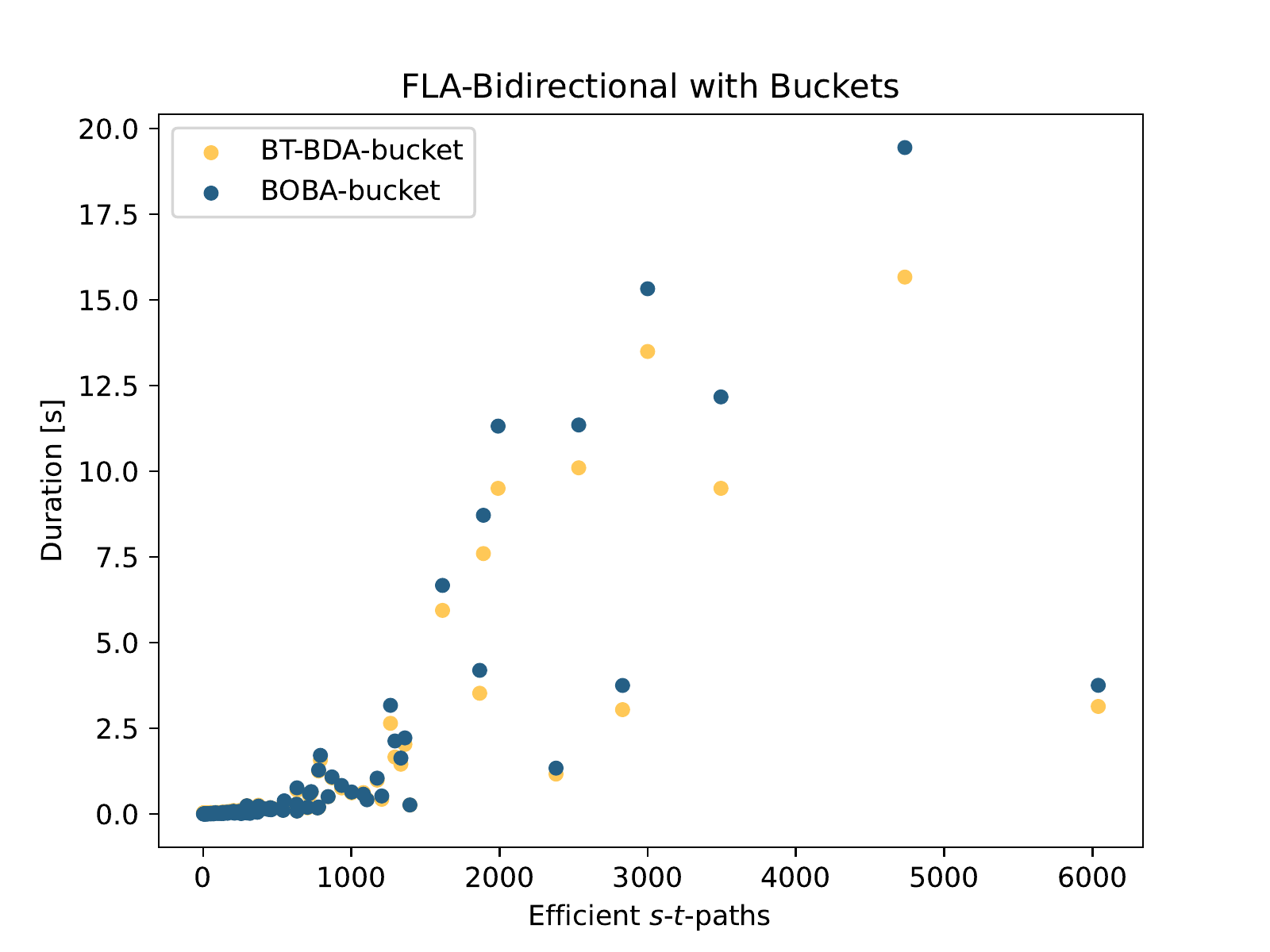}
			\captionof{figure}{}\label{fig:2d-bbda-boba-bucket-FLA}
		\end{minipage}
	\end{figure}
	\begin{figure}[H]
		\begin{minipage}{.48\linewidth}
			\captionsetup{type=figure}\includegraphics[width=\textwidth]{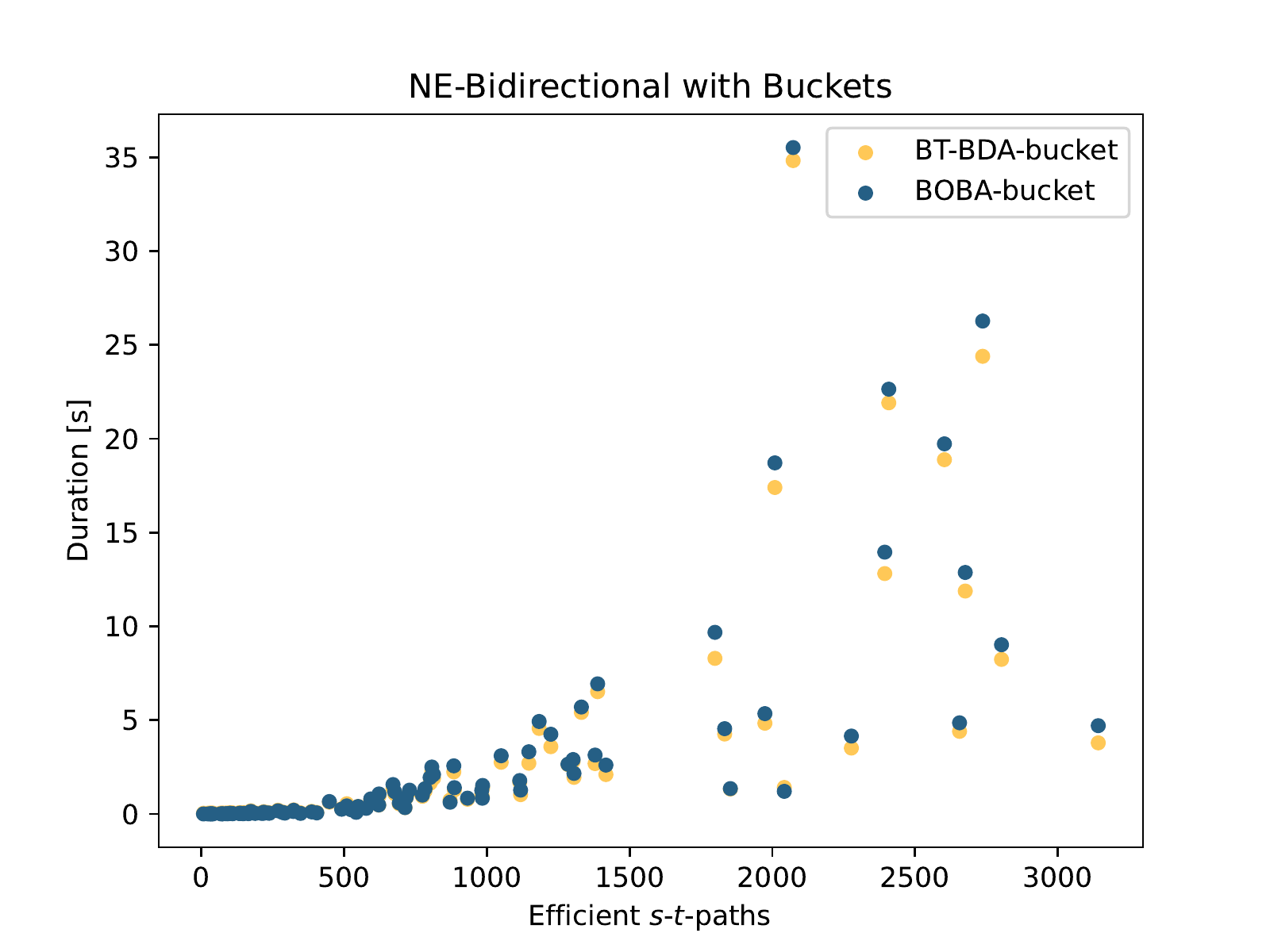}
			\captionof{figure}{}\label{fig:2d-bbda-boba-bucket-NE}
		\end{minipage}
		\begin{minipage}{.48\linewidth}
			\captionsetup{type=figure}\includegraphics[width=\textwidth]{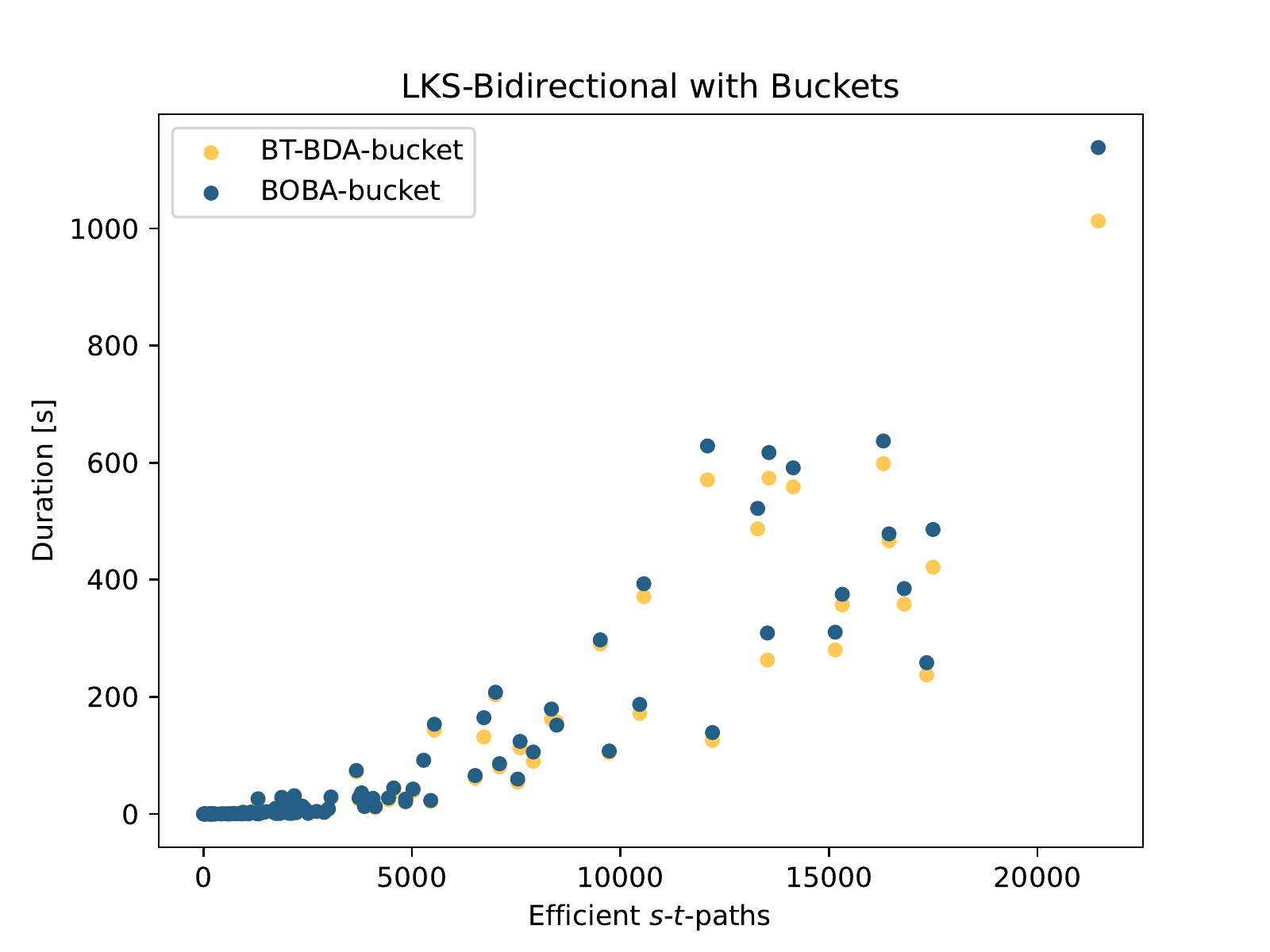}
			\captionof{figure}{}\label{fig:2d-bbda-boba-bucket-LKS}
		\end{minipage}
	\end{figure}
	\begin{figure}[H]
		\begin{minipage}{.48\linewidth}
			\captionsetup{type=figure}\includegraphics[width=\textwidth]{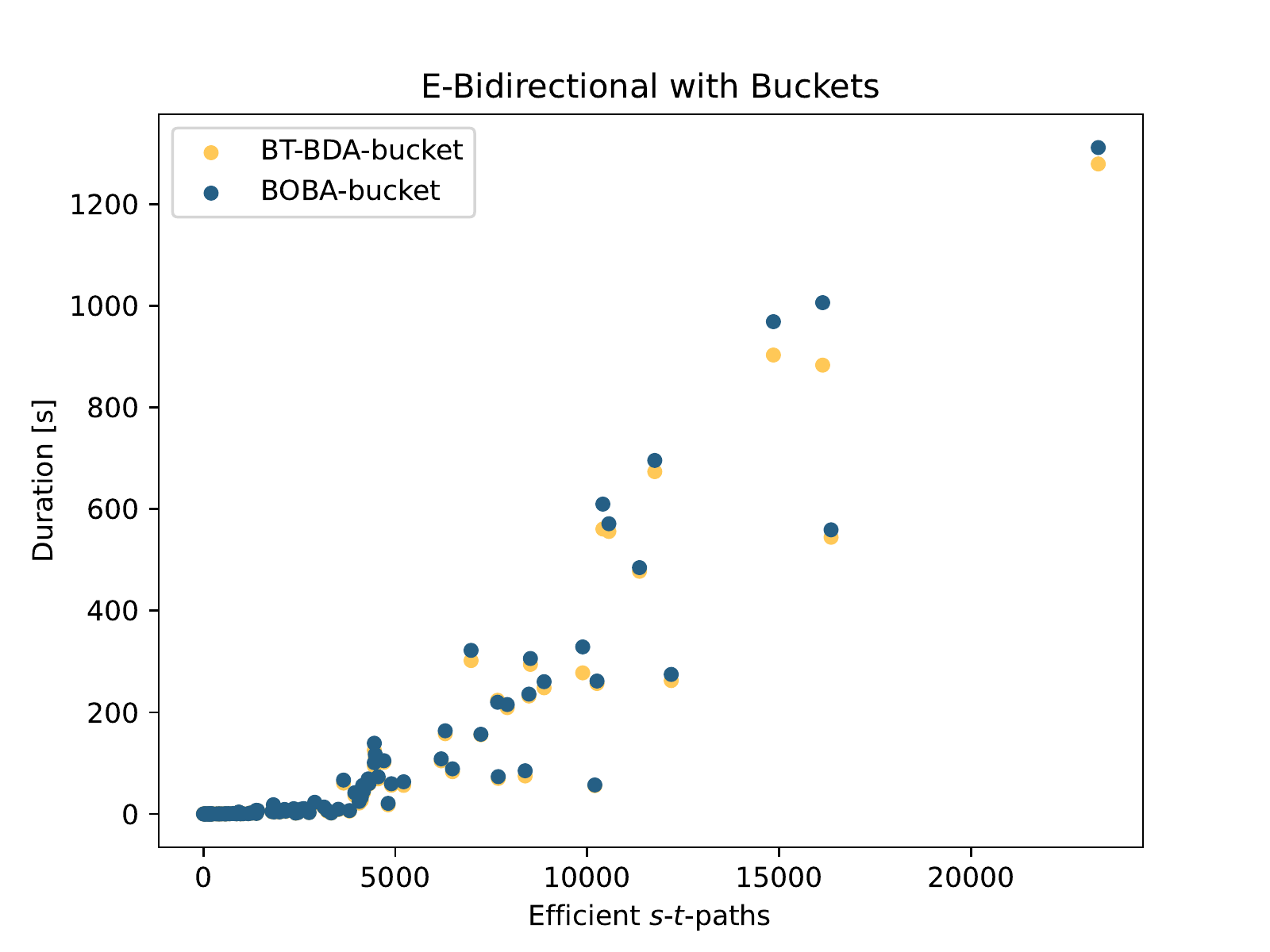}
			\captionof{figure}{}\label{fig:2d-bbda-boba-bucket-E}
		\end{minipage}
		\begin{minipage}{.48\linewidth}
			\captionsetup{type=figure}\includegraphics[width=\textwidth]{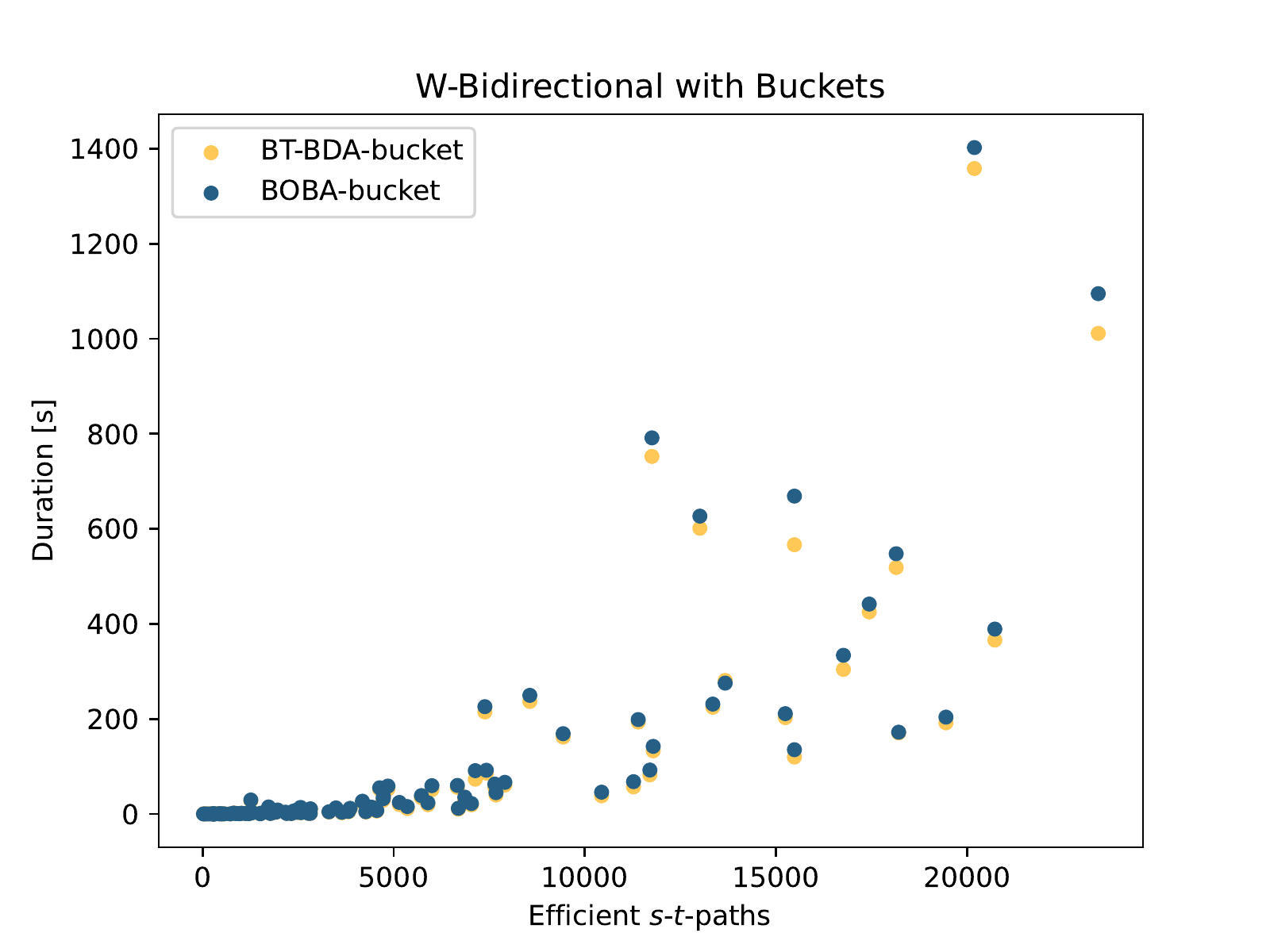}
			\captionof{figure}{}\label{fig:2d-bbda-boba-bucket-W}
		\end{minipage}
	\end{figure}
	\begin{figure}[H]
		\begin{minipage}{.48\linewidth}
			\captionsetup{type=figure}\includegraphics[width=\textwidth]{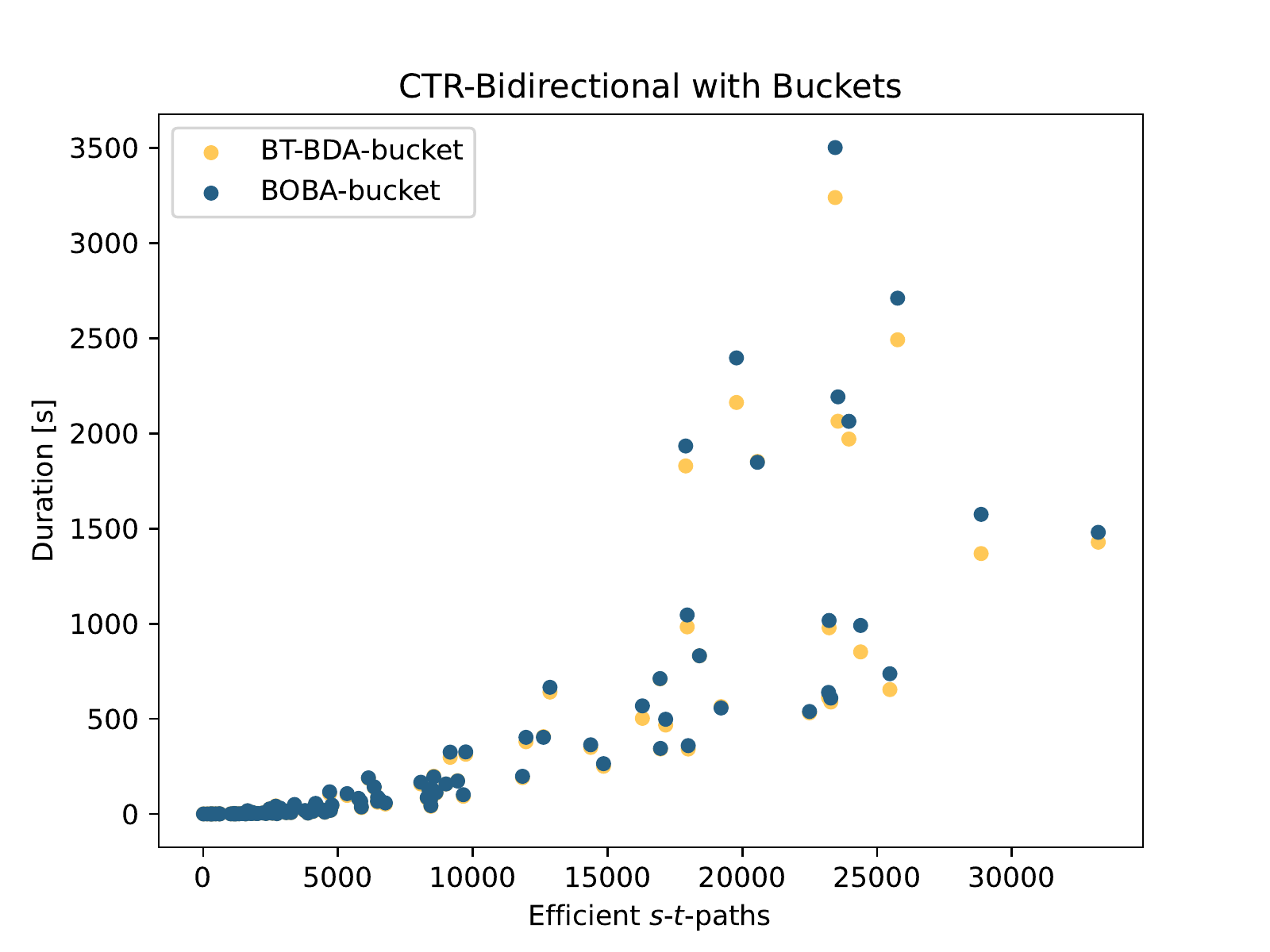}
			\captionof{figure}{}\label{fig:2d-bbda-boba-bucket-CTR}
		\end{minipage}
	\end{figure}

	\subsection{Graphics -- Bidimensional-Unidirectional-Buckets}
		\begin{figure}[H]
			\begin{minipage}{.48\linewidth}
				\captionsetup{type=figure}\includegraphics[width=\textwidth]{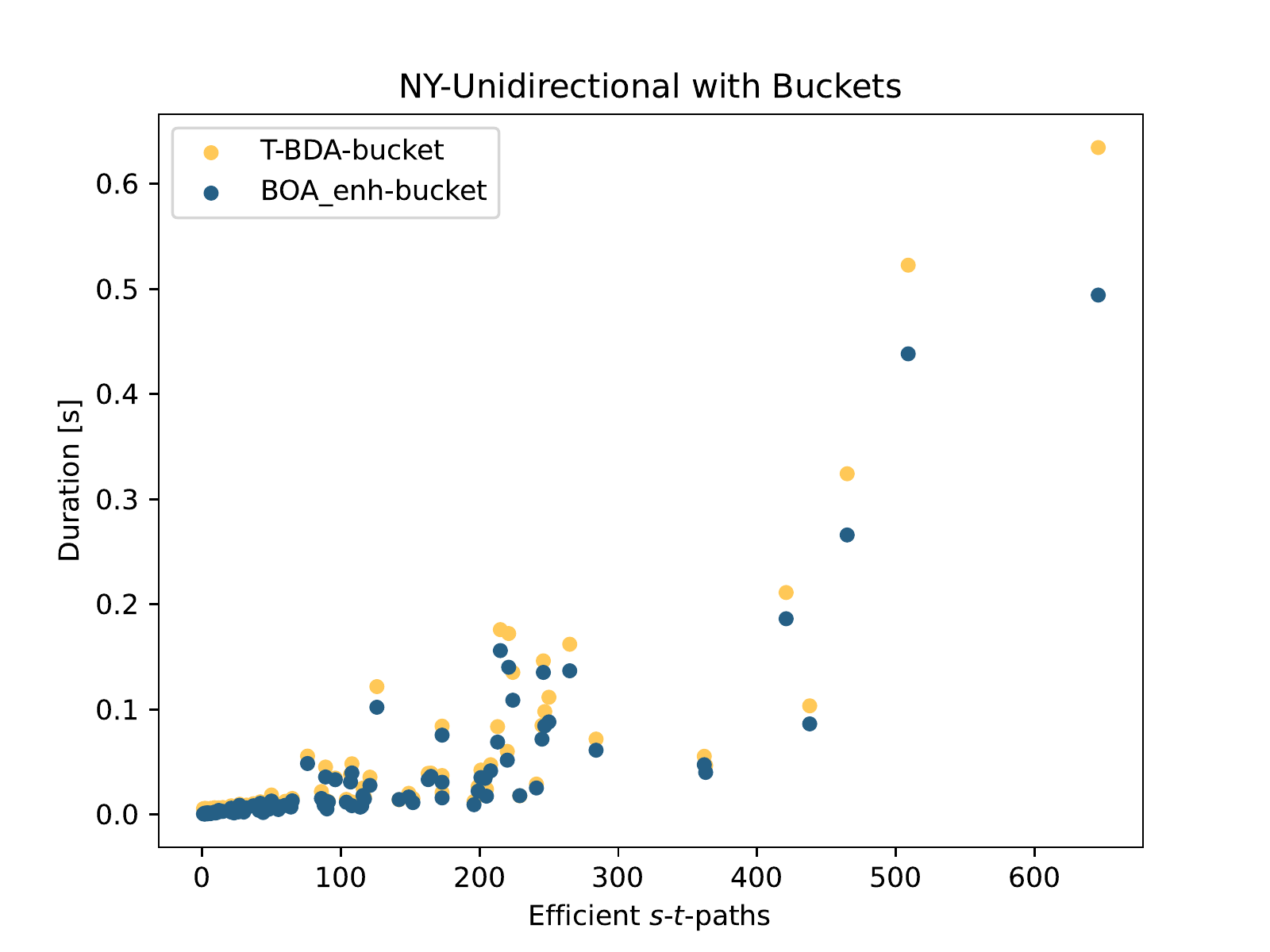}
				\captionof{figure}{}\label{fig:2d-bda-boa-bucket-NY}
			\end{minipage}
			\begin{minipage}{.48\linewidth}
				\captionsetup{type=figure}\includegraphics[width=\textwidth]{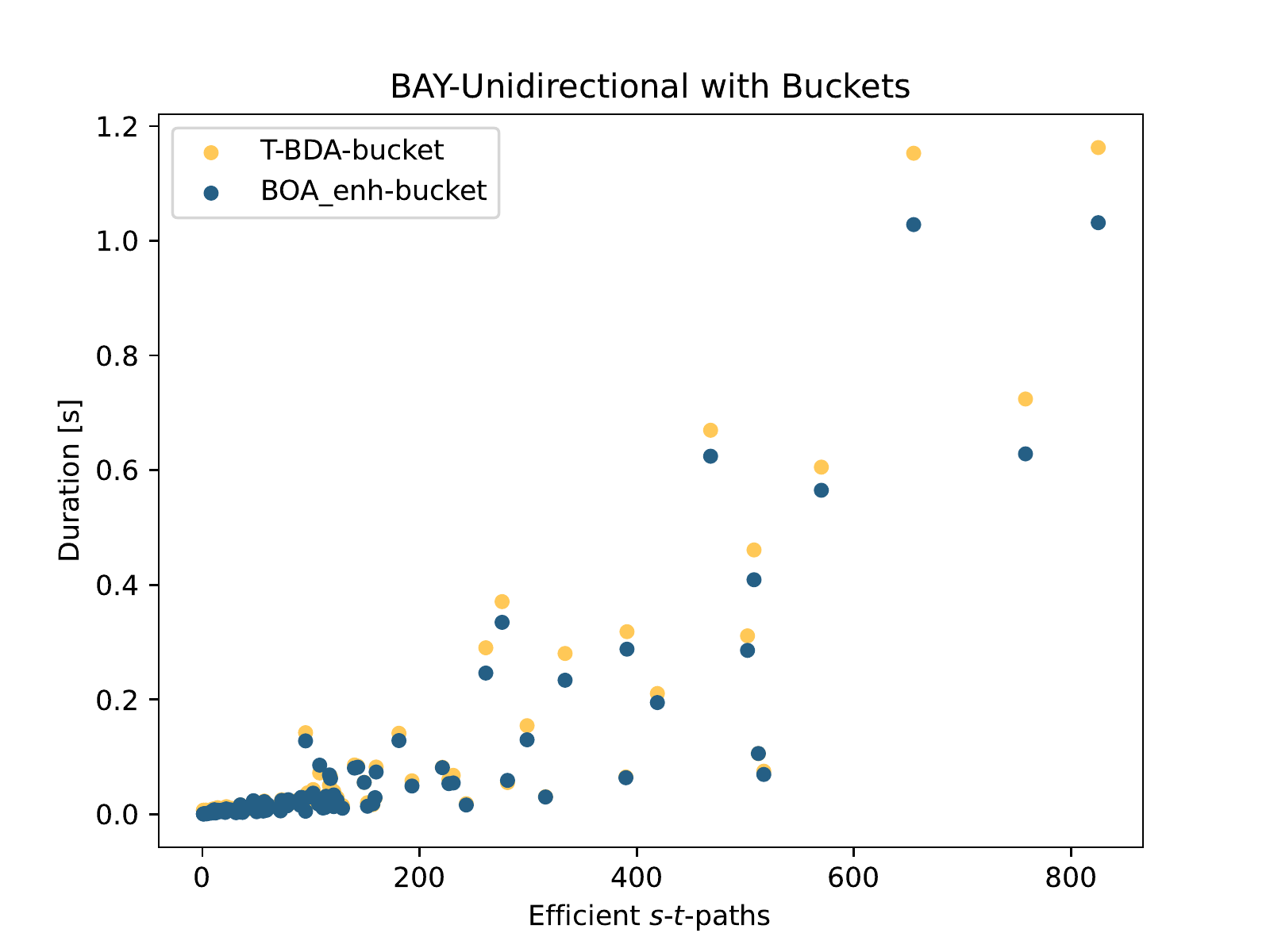}
				\captionof{figure}{}\label{fig:2d-bda-boa-bucket-BAY}
			\end{minipage}
		\end{figure}
		\begin{figure}[H]
			\begin{minipage}{.48\linewidth}
				\captionsetup{type=figure}\includegraphics[width=\textwidth]{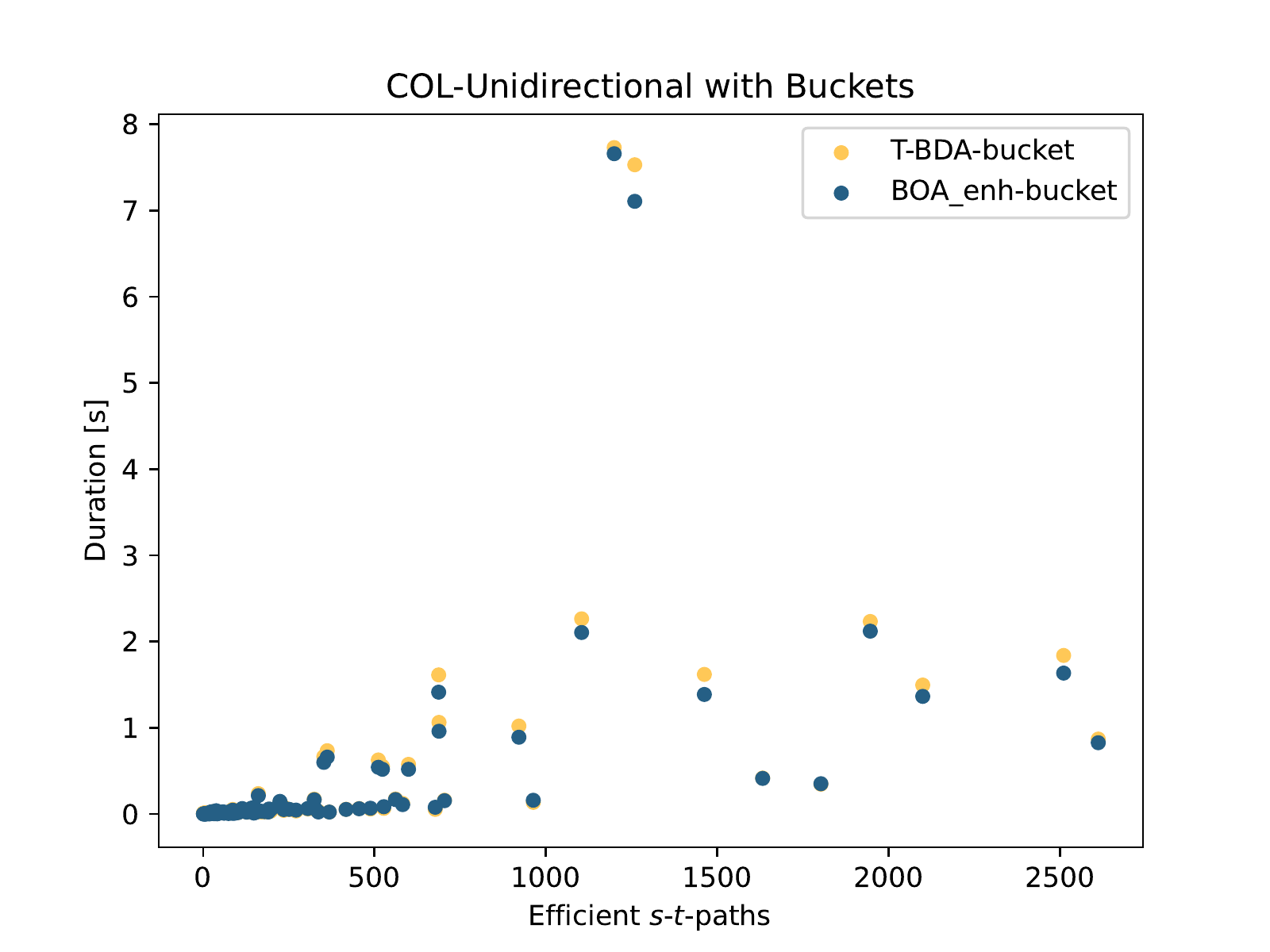}
				\captionof{figure}{}\label{fig:2d-bda-boa-bucket-COL}
			\end{minipage}
			\begin{minipage}{.48\linewidth}
				\captionsetup{type=figure}\includegraphics[width=\textwidth]{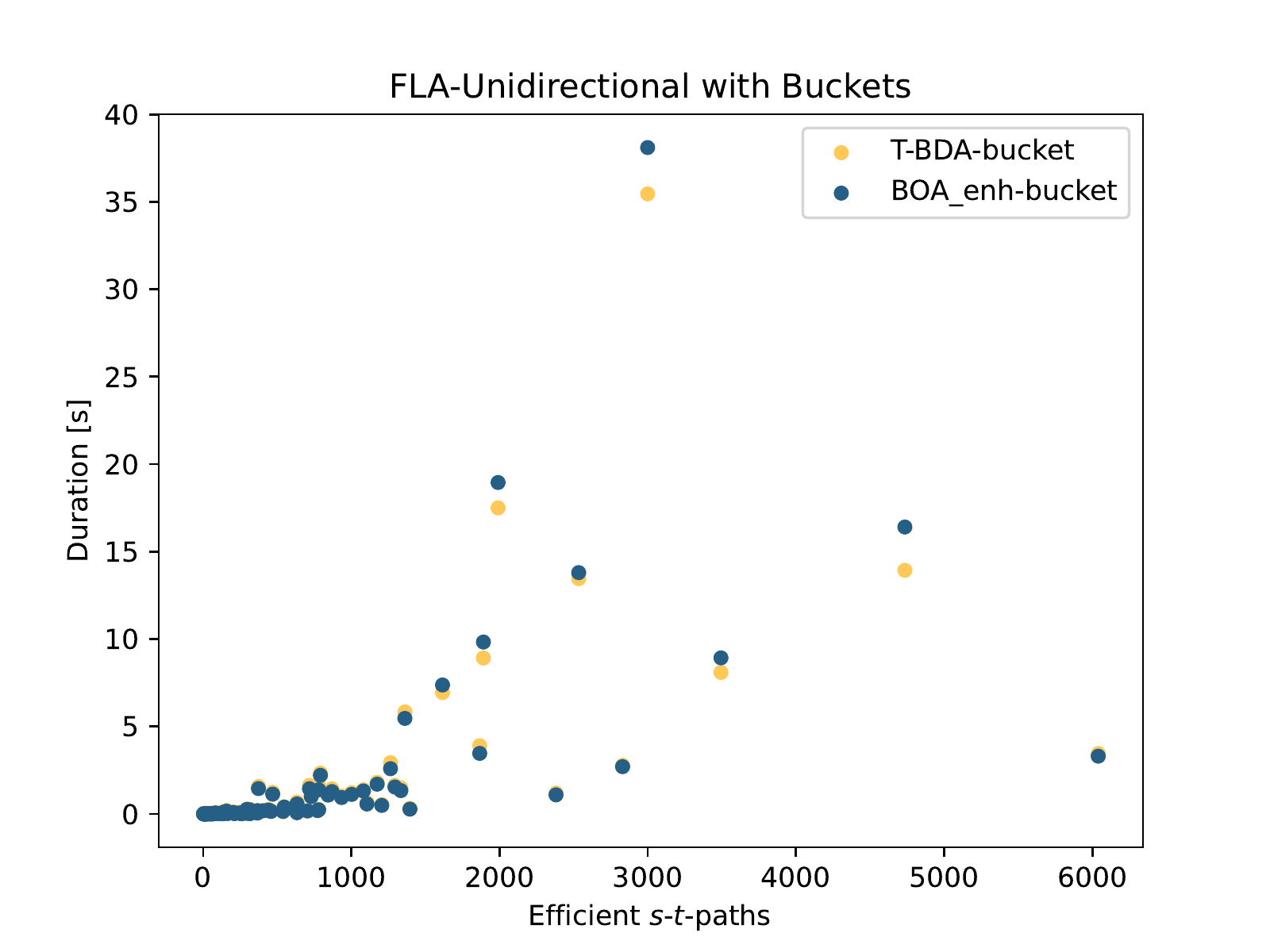}
				\captionof{figure}{}\label{fig:2d-bda-boa-bucket-FLA}
			\end{minipage}
		\end{figure}
		\begin{figure}[H]
			\begin{minipage}{.48\linewidth}
				\captionsetup{type=figure}\includegraphics[width=\textwidth]{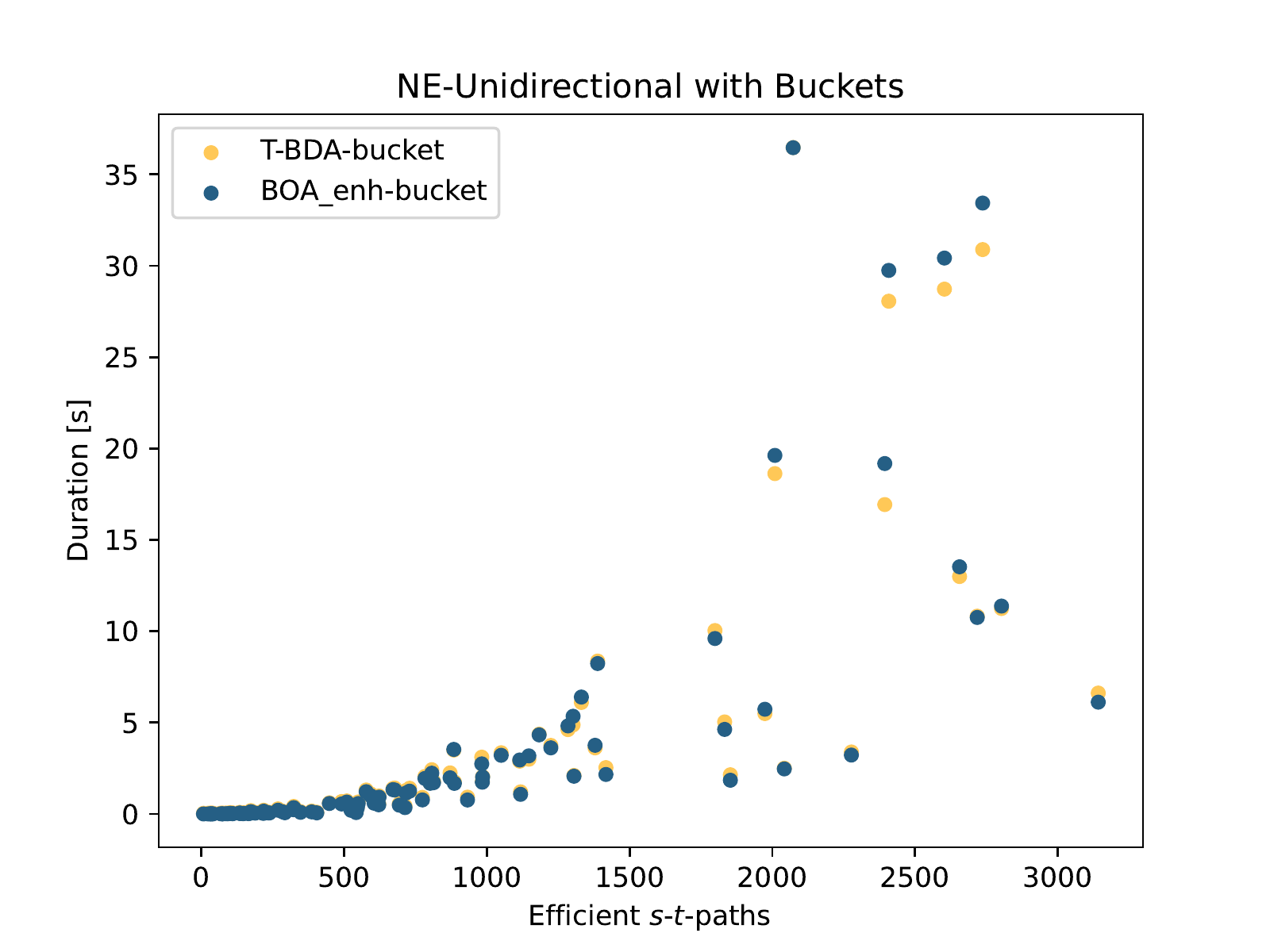}
				\captionof{figure}{}\label{fig:2d-bda-boa-bucket-NE}
			\end{minipage}
			\begin{minipage}{.48\linewidth}
				\captionsetup{type=figure}\includegraphics[width=\textwidth]{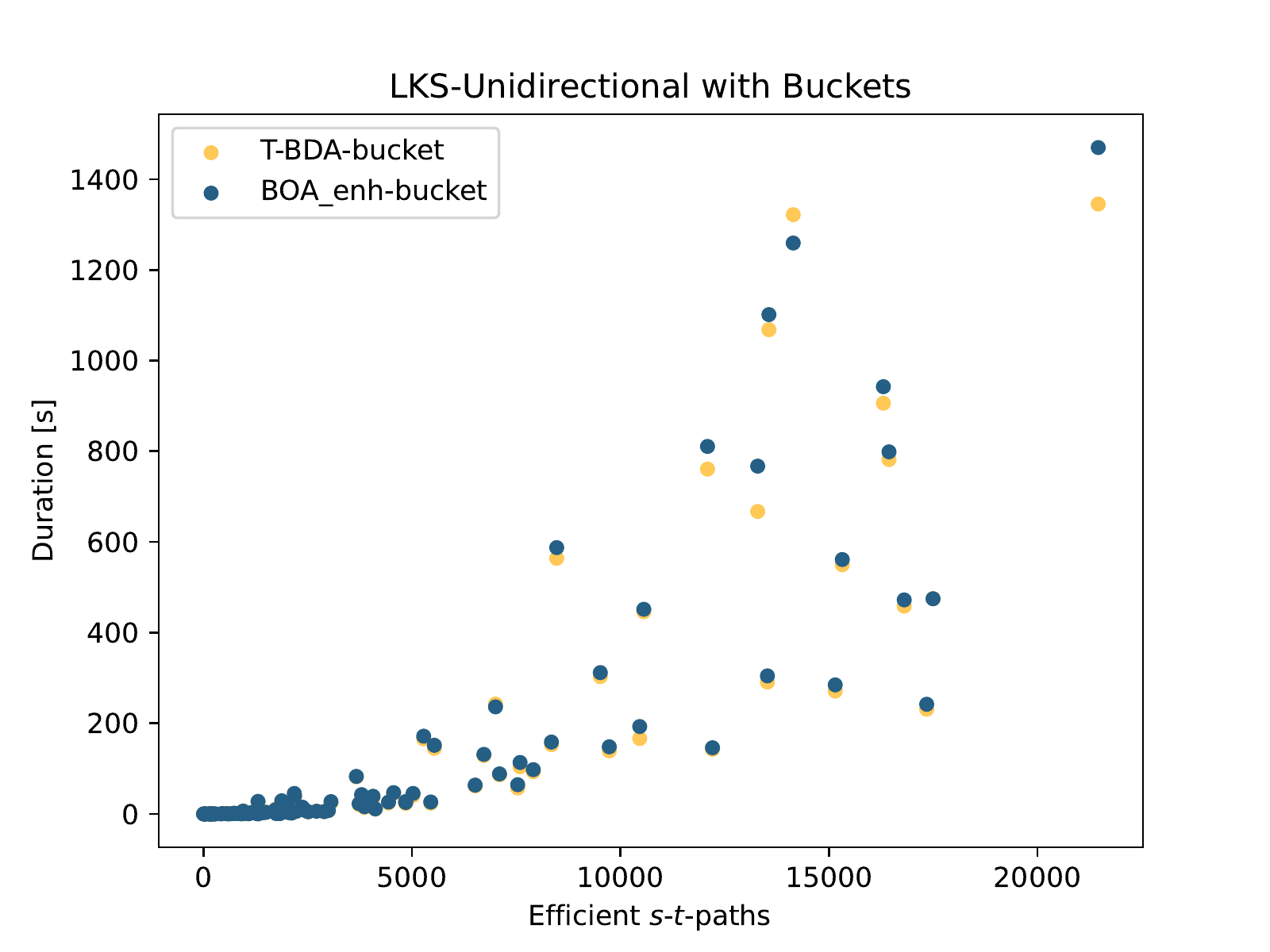}
				\captionof{figure}{}\label{fig:2d-bda-boa-bucket-LKS}
			\end{minipage}
		\end{figure}
		\begin{figure}[H]
			\begin{minipage}{.48\linewidth}
				\captionsetup{type=figure}\includegraphics[width=\textwidth]{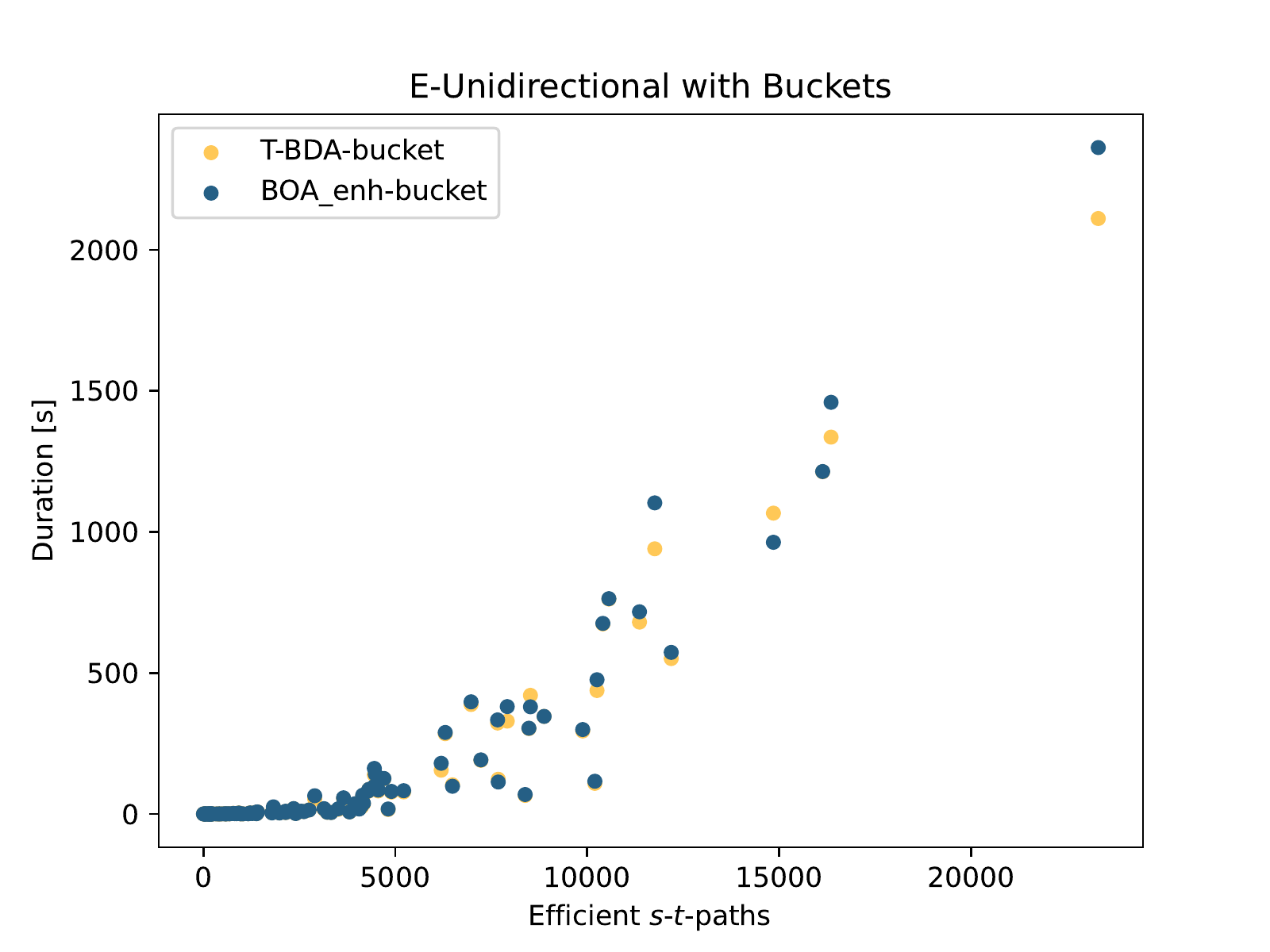}
				\captionof{figure}{}\label{fig:2d-bda-boa-bucket-E}
			\end{minipage}
			\begin{minipage}{.48\linewidth}
				\captionsetup{type=figure}\includegraphics[width=\textwidth]{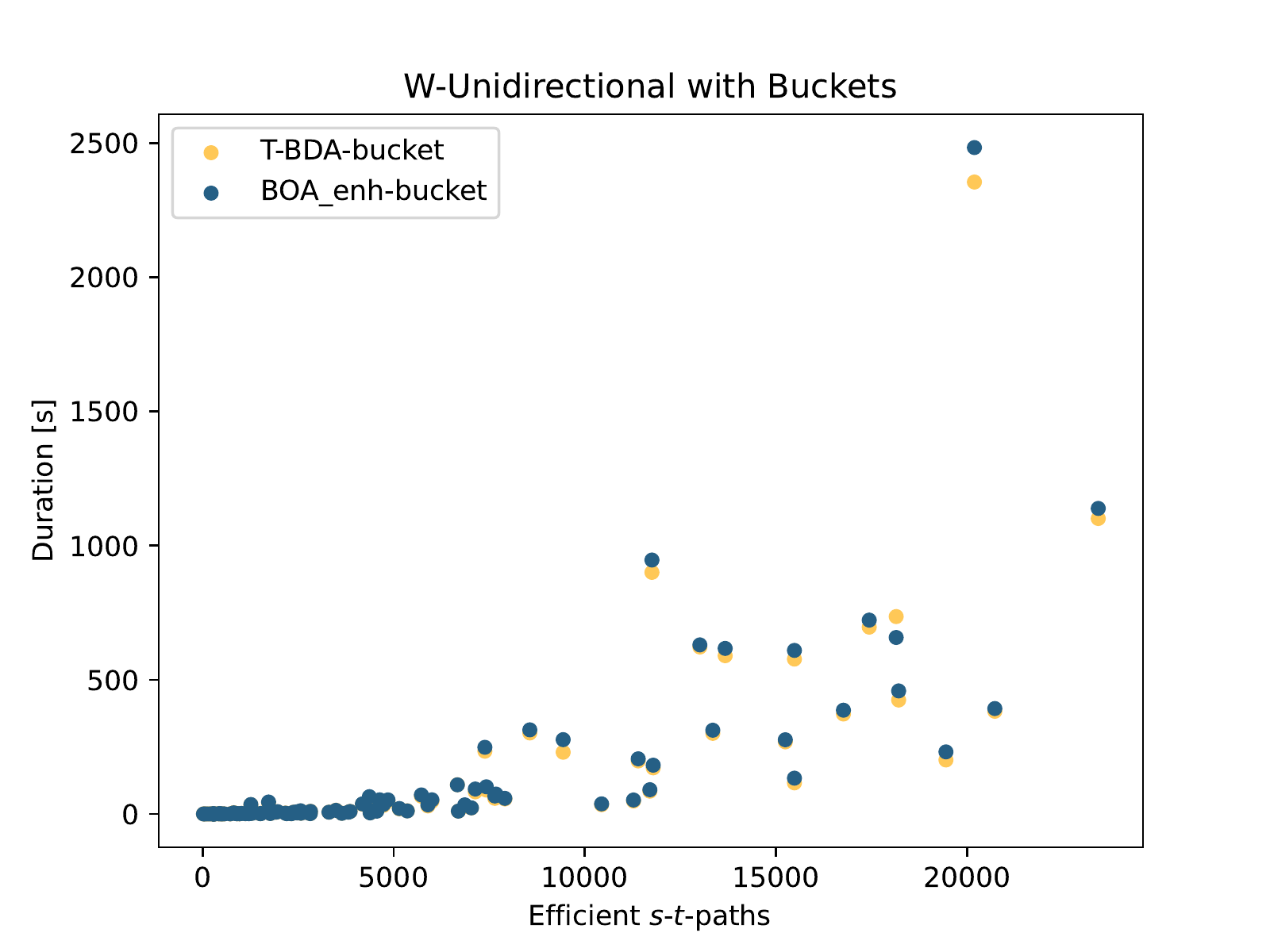}
				\captionof{figure}{}\label{fig:2d-bda-boa-bucket-W}
			\end{minipage}
		\end{figure}
		\begin{figure}[H]
			\begin{minipage}{.48\linewidth}
				\captionsetup{type=figure}\includegraphics[width=\textwidth]{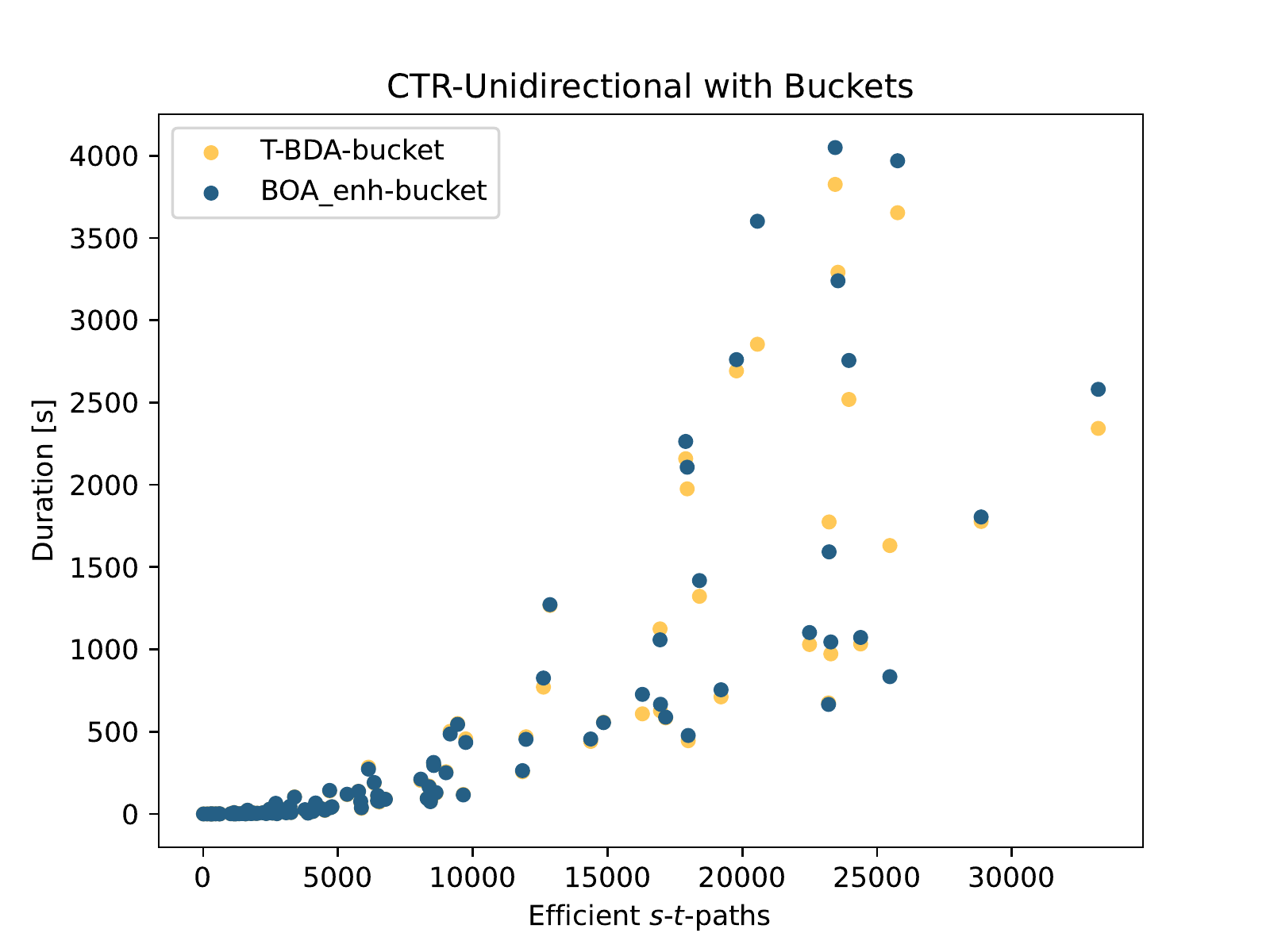}
				\captionof{figure}{}\label{fig:2d-bda-boa-bucket-CTR}
			\end{minipage}
		\end{figure}
	
		\subsection{Graphics -- Bidimensional-Bidirectional-Heaps}
	\begin{figure}[H]
		\begin{minipage}{.48\linewidth}
			\captionsetup{type=figure}\includegraphics[width=\textwidth]{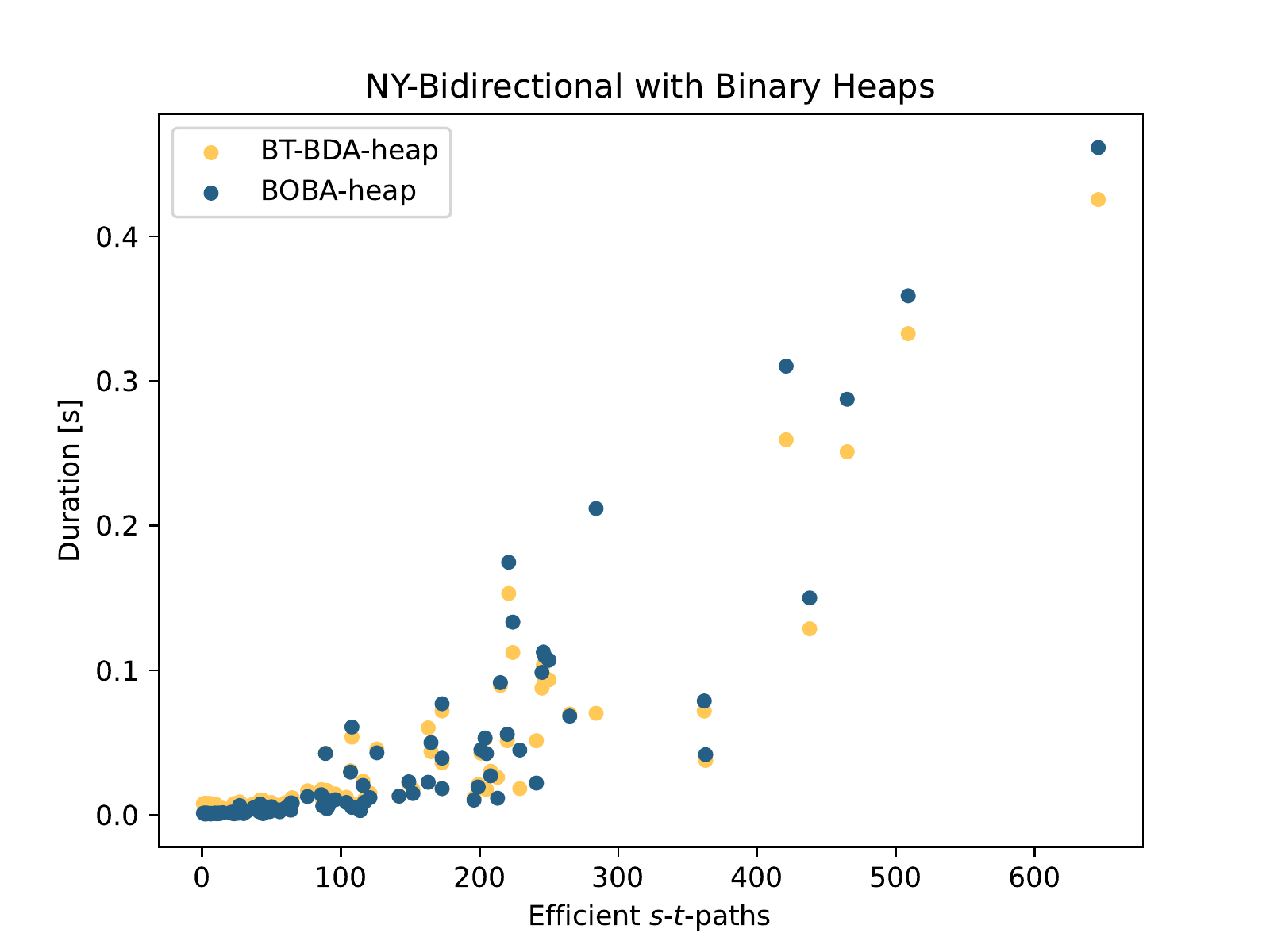}
			\captionof{figure}{}\label{fig:2d-bbda-boba-heap-NY}
		\end{minipage}
		\begin{minipage}{.48\linewidth}
			\captionsetup{type=figure}\includegraphics[width=\textwidth]{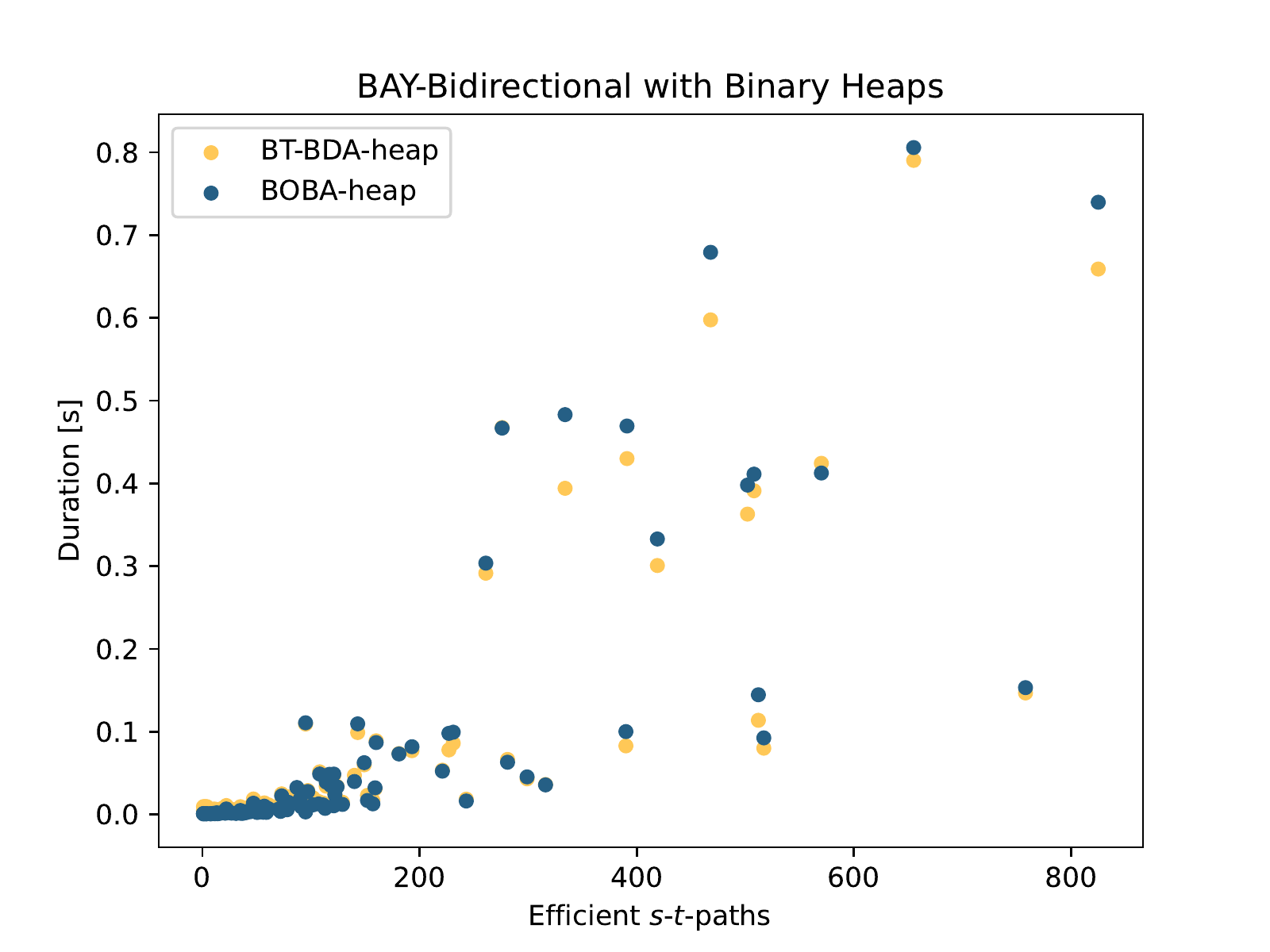}
			\captionof{figure}{}\label{fig:2d-bbda-boba-heap-BAY}
		\end{minipage}
	\end{figure}
	\begin{figure}[H]
		\begin{minipage}{.48\linewidth}
			\captionsetup{type=figure}\includegraphics[width=\textwidth]{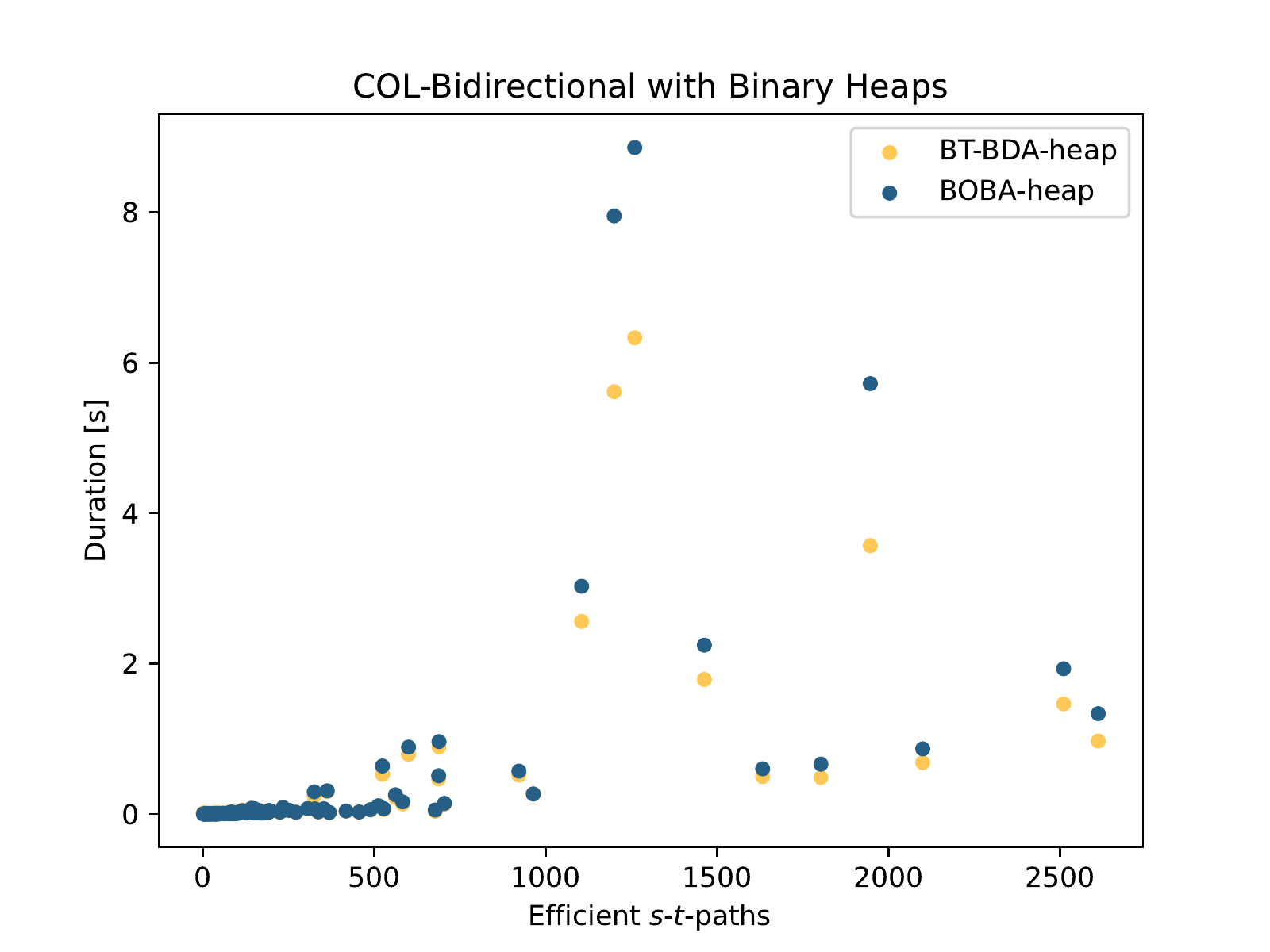}
			\captionof{figure}{}\label{fig:2d-bbda-boba-heap-COL}
		\end{minipage}
		\begin{minipage}{.48\linewidth}
			\captionsetup{type=figure}\includegraphics[width=\textwidth]{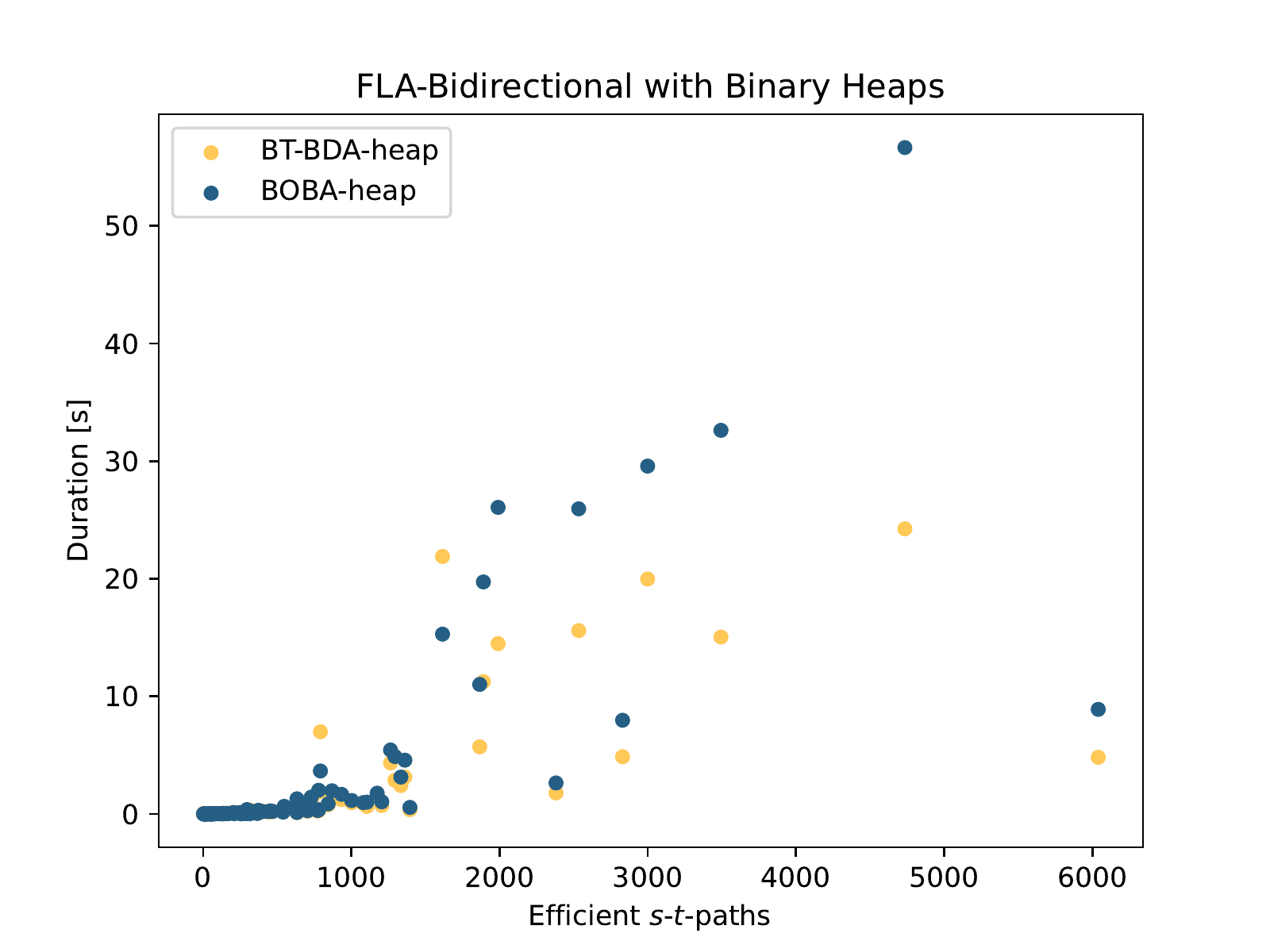}
			\captionof{figure}{}\label{fig:2d-bbda-boba-heap-FLA}
		\end{minipage}
	\end{figure}
	\begin{figure}[H]
		\begin{minipage}{.48\linewidth}
			\captionsetup{type=figure}\includegraphics[width=\textwidth]{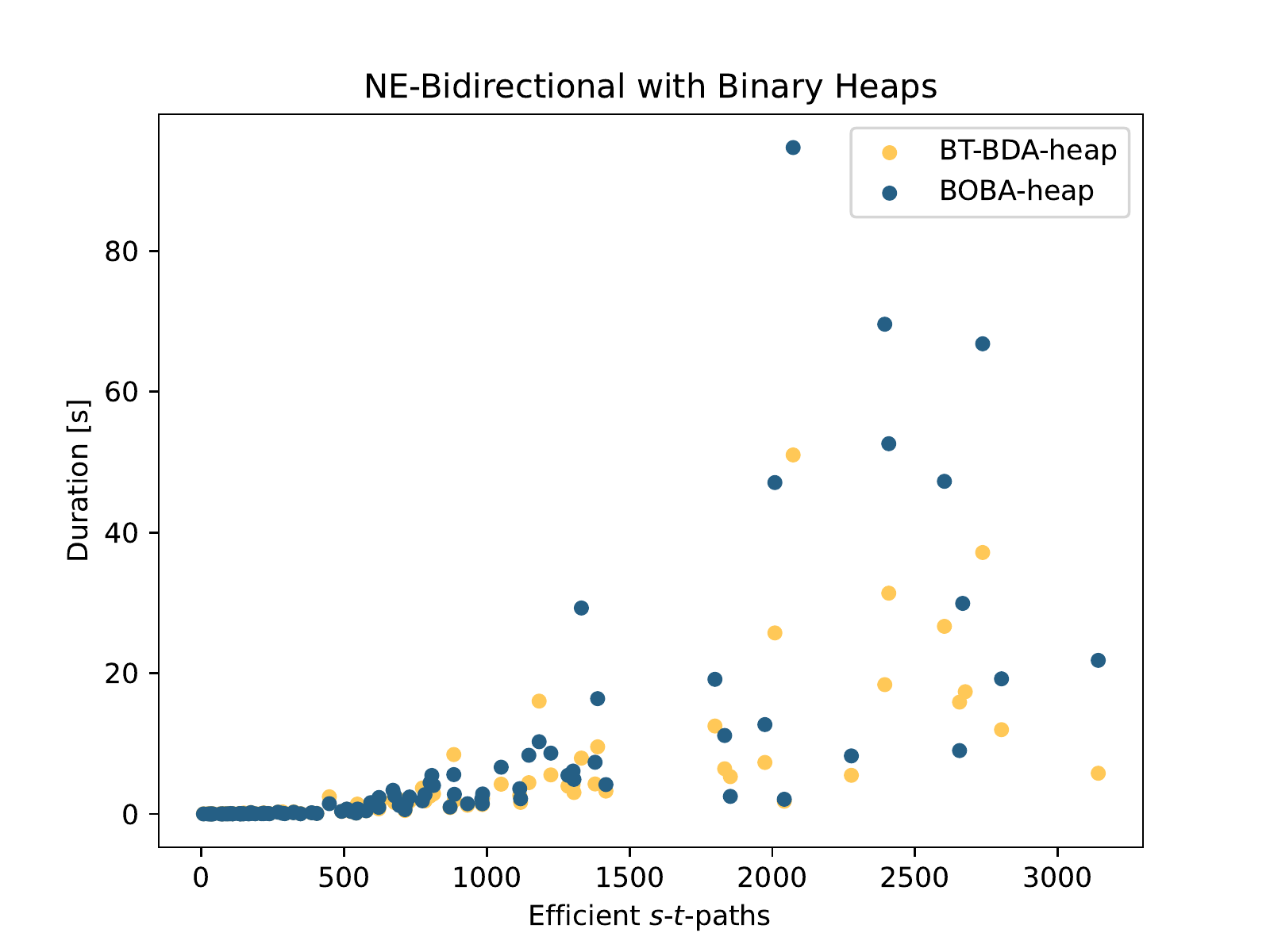}
			\captionof{figure}{}\label{fig:2d-bbda-boba-heap-NE}
		\end{minipage}
		\begin{minipage}{.48\linewidth}
			\captionsetup{type=figure}\includegraphics[width=\textwidth]{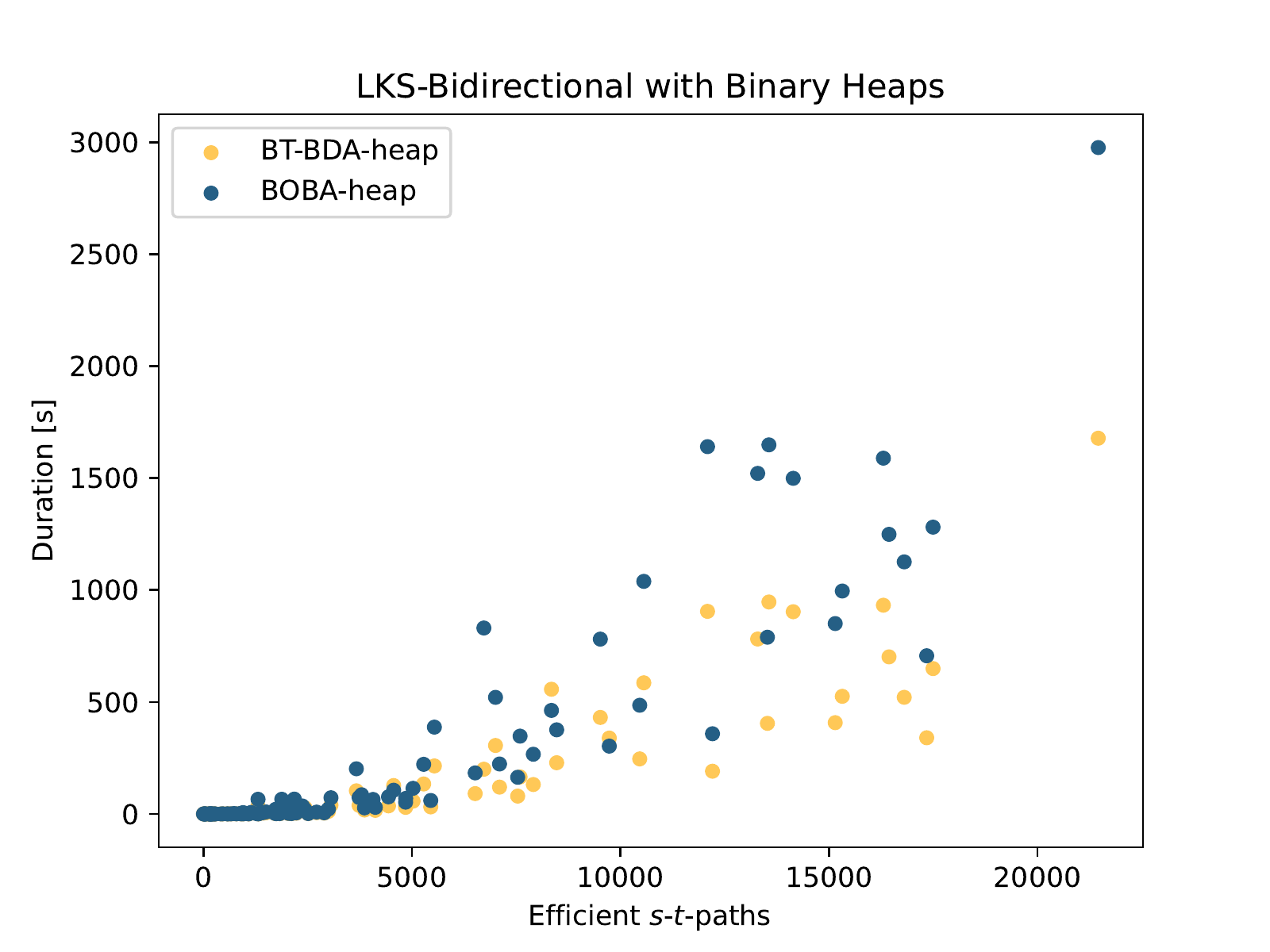}
			\captionof{figure}{}\label{fig:2d-bbda-boba-heap-LKS}
		\end{minipage}
	\end{figure}
	\begin{figure}[H]
		\begin{minipage}{.48\linewidth}
			\captionsetup{type=figure}\includegraphics[width=\textwidth]{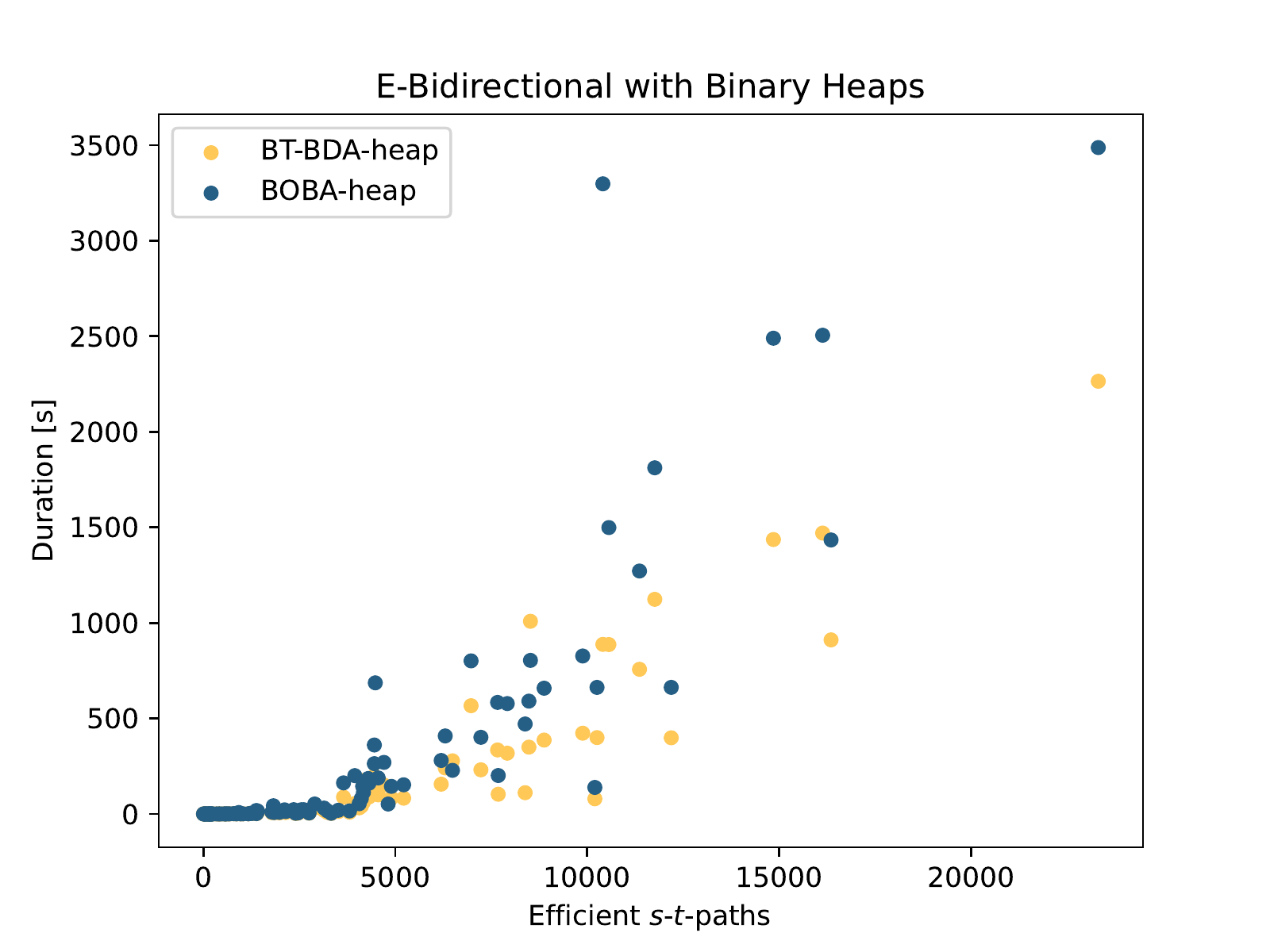}
			\captionof{figure}{}\label{fig:2d-bbda-boba-heap-E}
		\end{minipage}
		\begin{minipage}{.48\linewidth}
			\captionsetup{type=figure}\includegraphics[width=\textwidth]{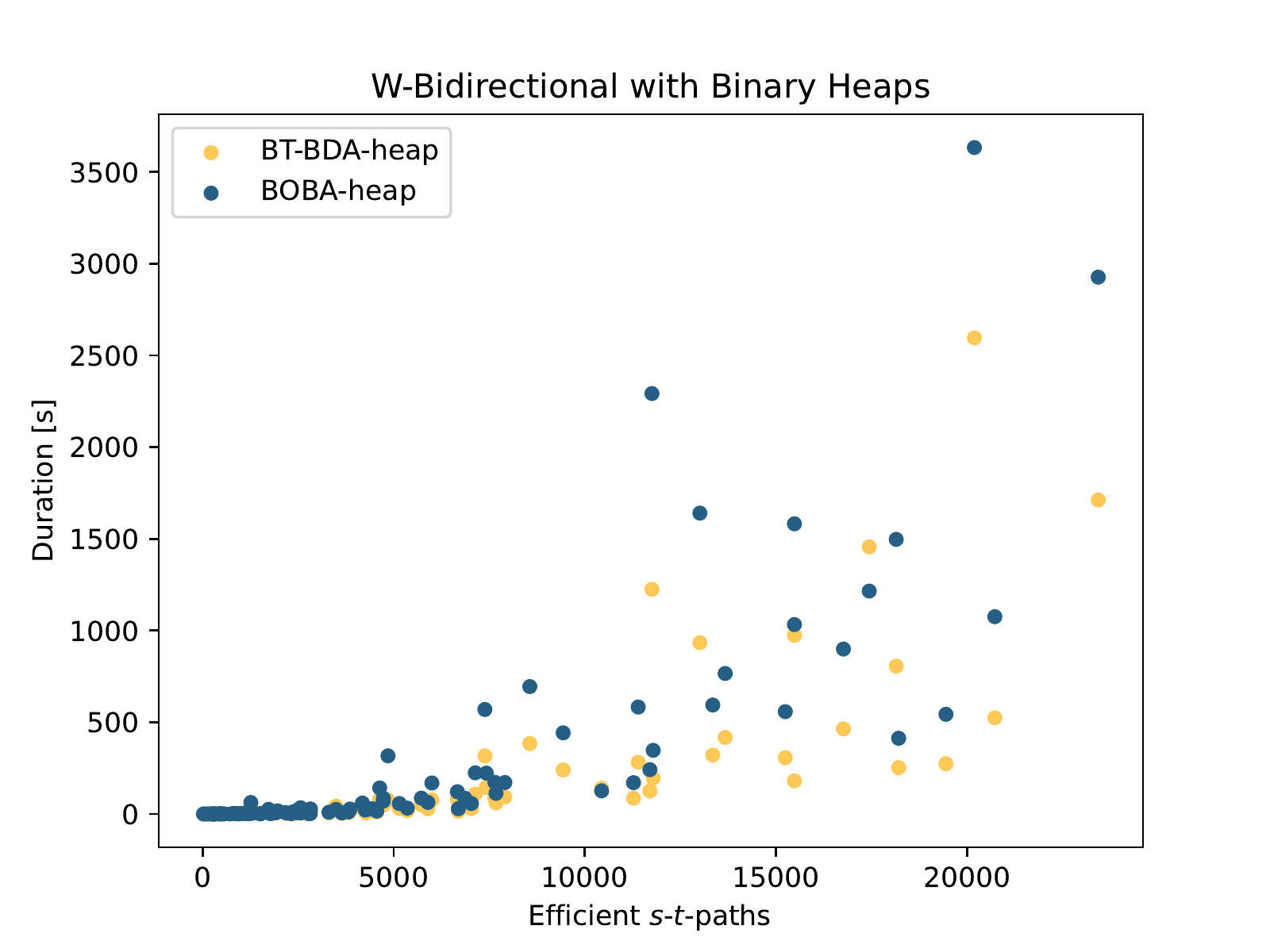}
			\captionof{figure}{}\label{fig:2d-bbda-boba-heap-W}
		\end{minipage}
	\end{figure}
	\begin{figure}[H]
		\begin{minipage}{.48\linewidth}
			\captionsetup{type=figure}\includegraphics[width=\textwidth]{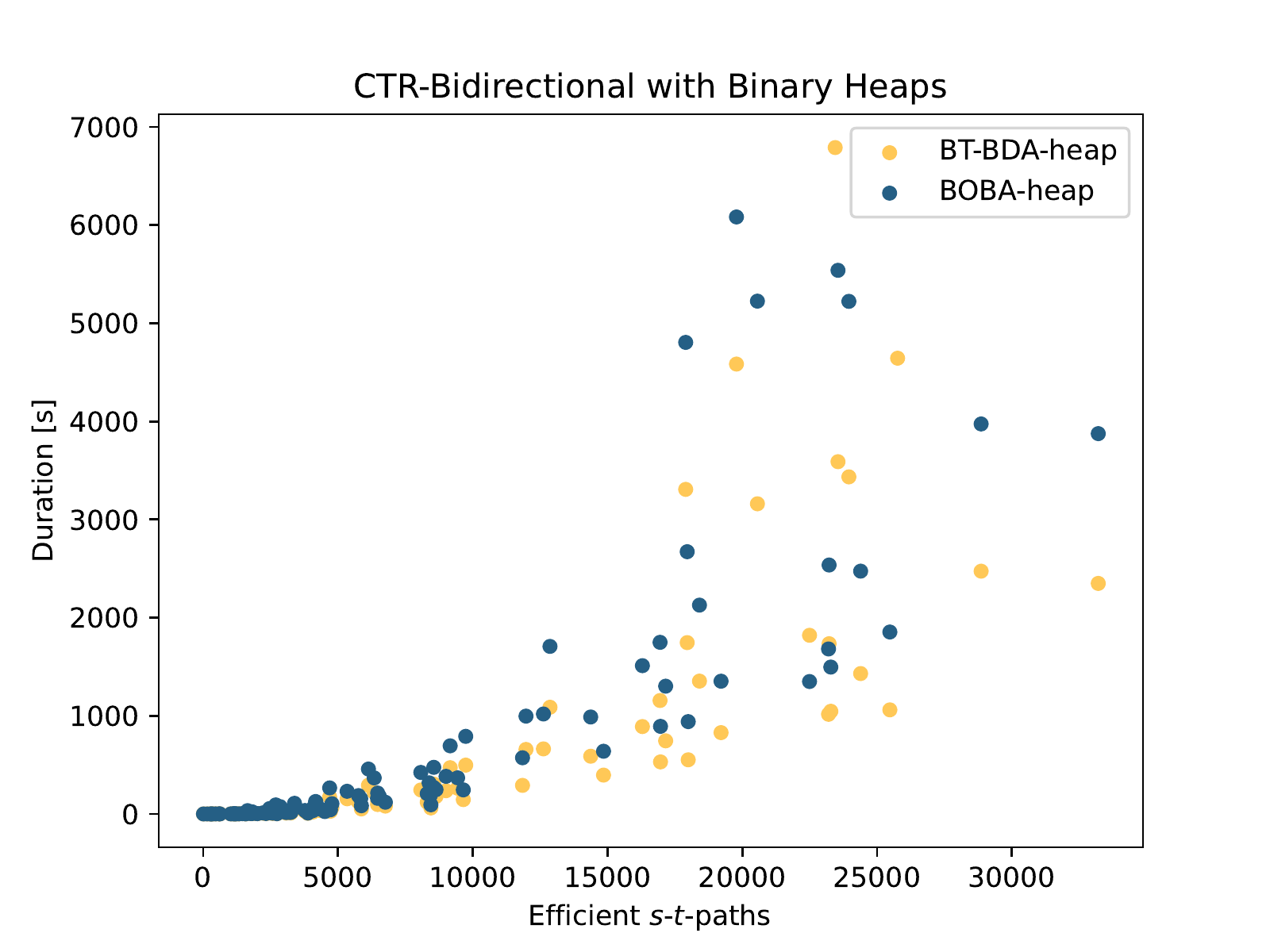}
			\captionof{figure}{}\label{fig:2d-bbda-boba-heap-CTR}
		\end{minipage}
	\end{figure}

		\subsection{Graphics -- Bidimensional-Unidirectional-Heaps}
\begin{figure}[H]
	\begin{minipage}{.48\linewidth}
		\captionsetup{type=figure}\includegraphics[width=\textwidth]{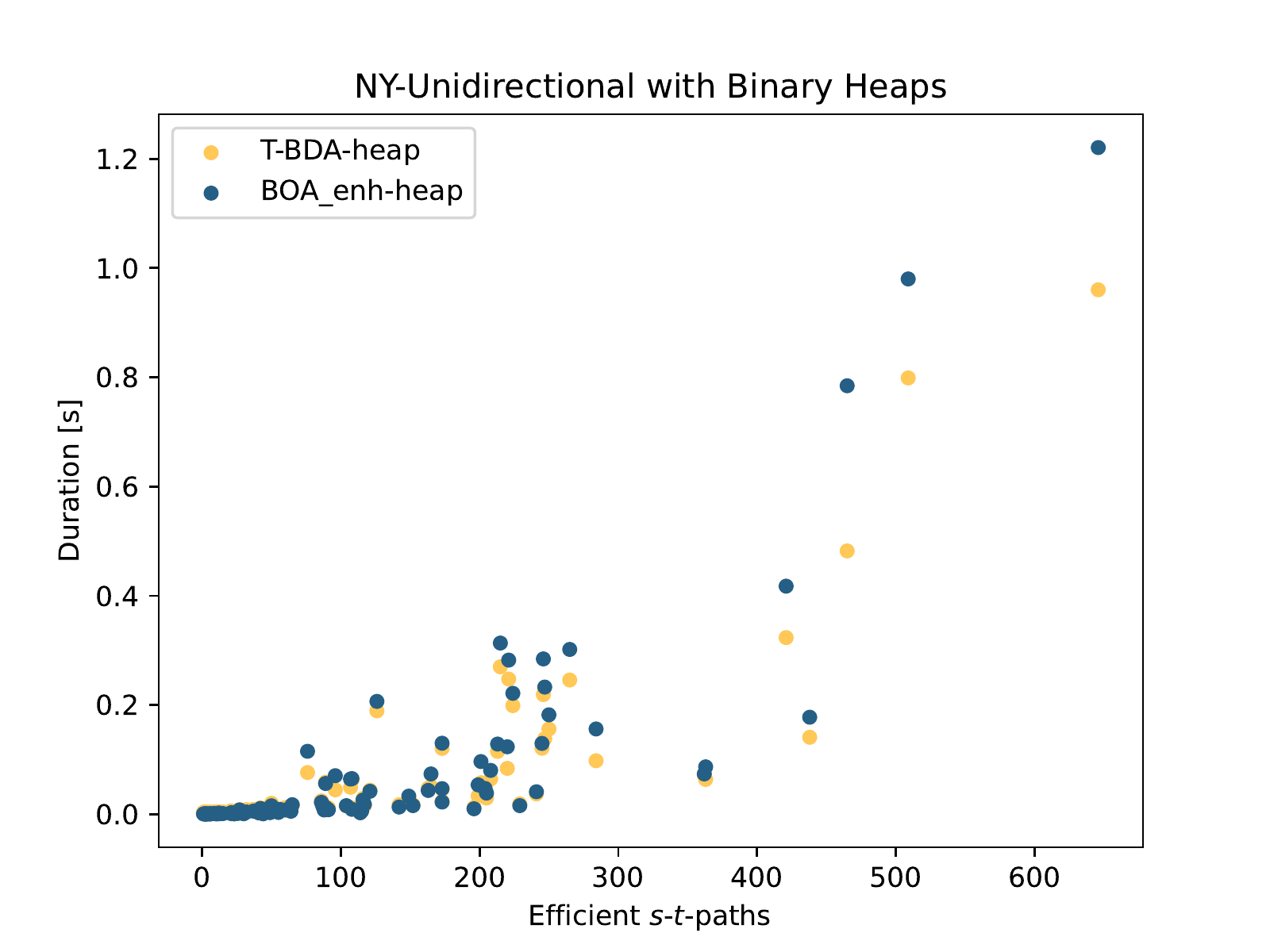}
		\captionof{figure}{}\label{fig:2d-bda-boa-heap-NY}
	\end{minipage}
	\begin{minipage}{.48\linewidth}
		\captionsetup{type=figure}\includegraphics[width=\textwidth]{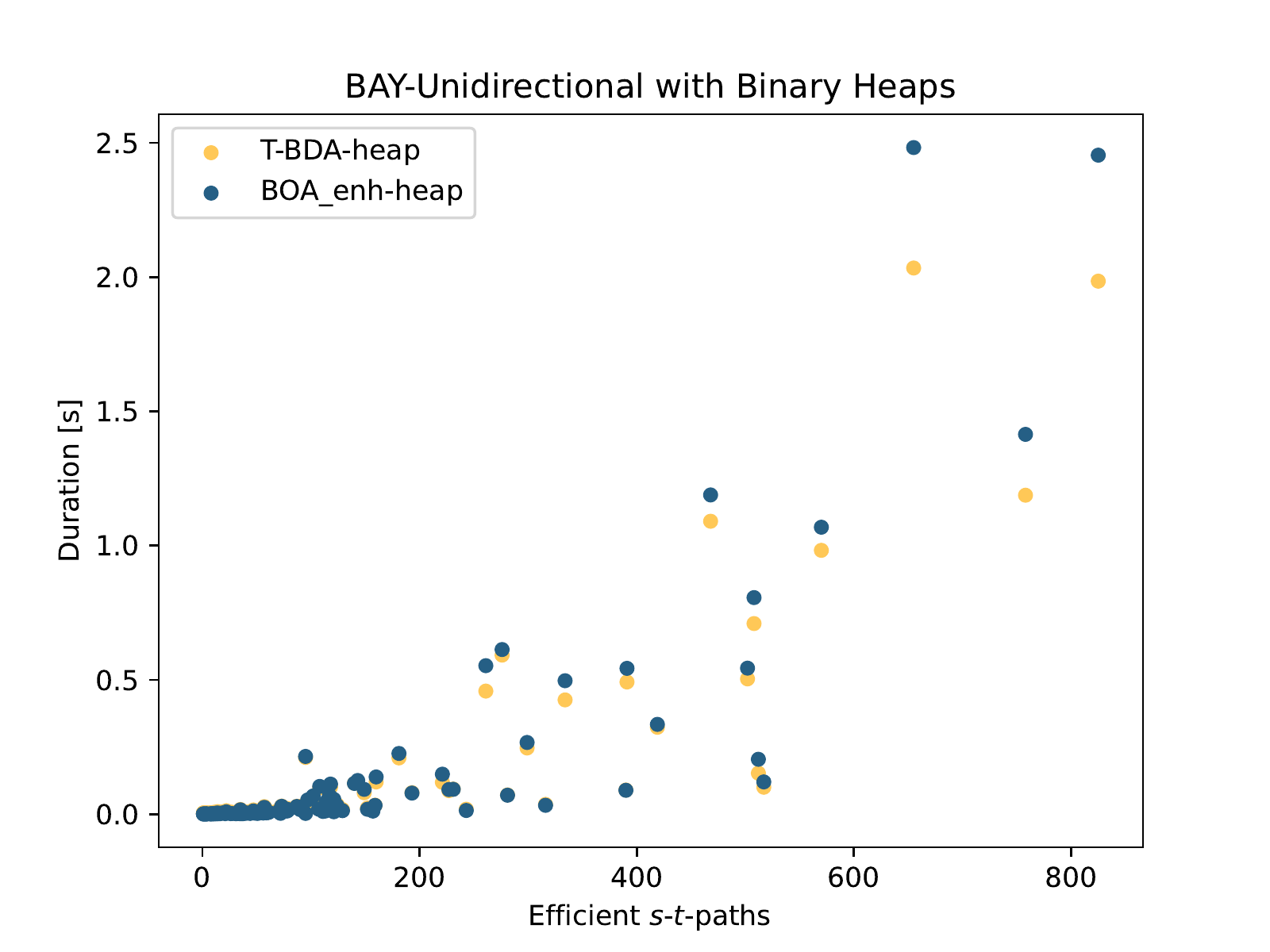}
		\captionof{figure}{}\label{fig:2d-bda-boa-heap-BAY}
	\end{minipage}
\end{figure}
\begin{figure}[H]
	\begin{minipage}{.48\linewidth}
		\captionsetup{type=figure}\includegraphics[width=\textwidth]{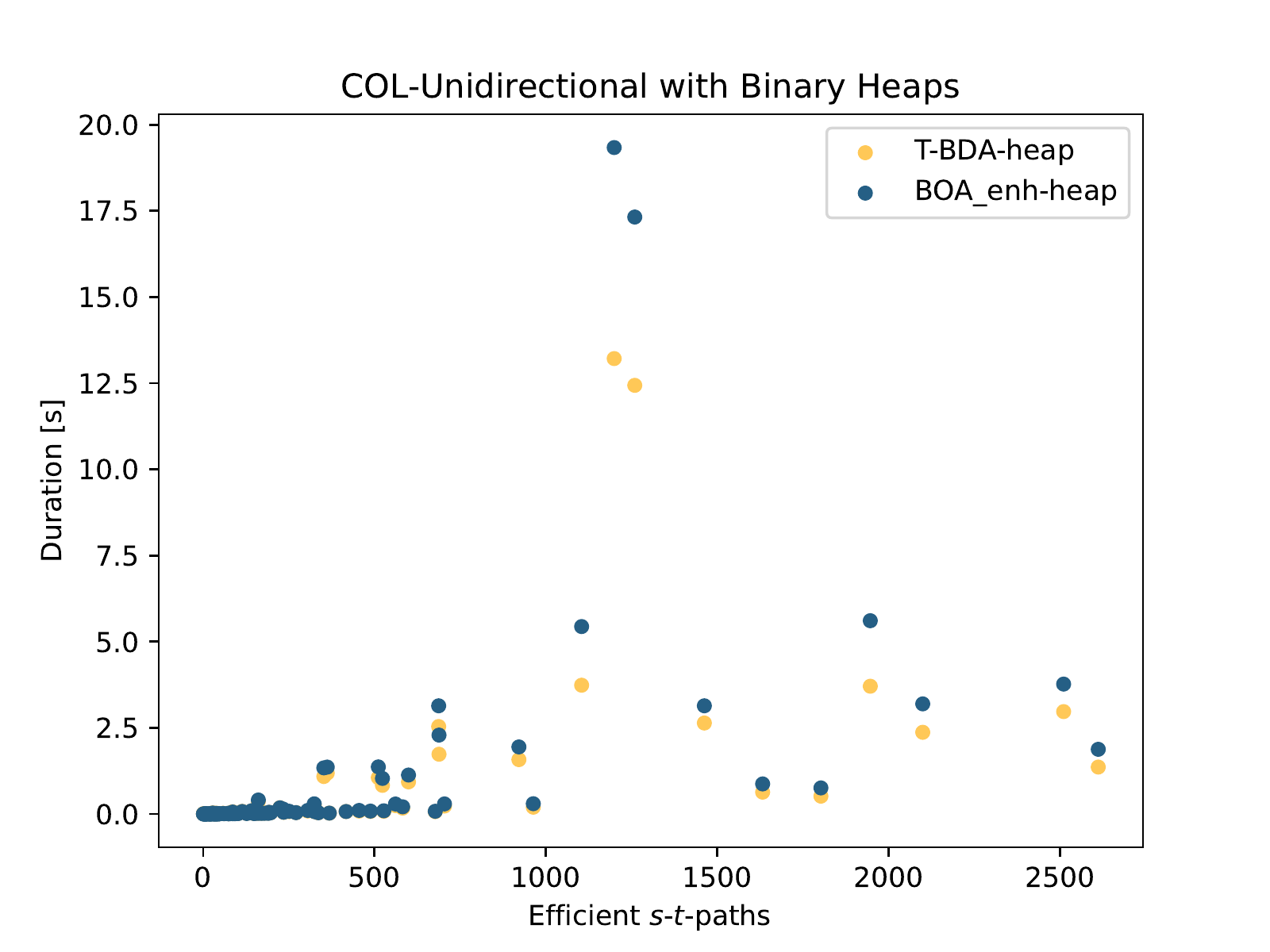}
		\captionof{figure}{}\label{fig:2d-bda-boa-heap-COL}
	\end{minipage}
	\begin{minipage}{.48\linewidth}
		\captionsetup{type=figure}\includegraphics[width=\textwidth]{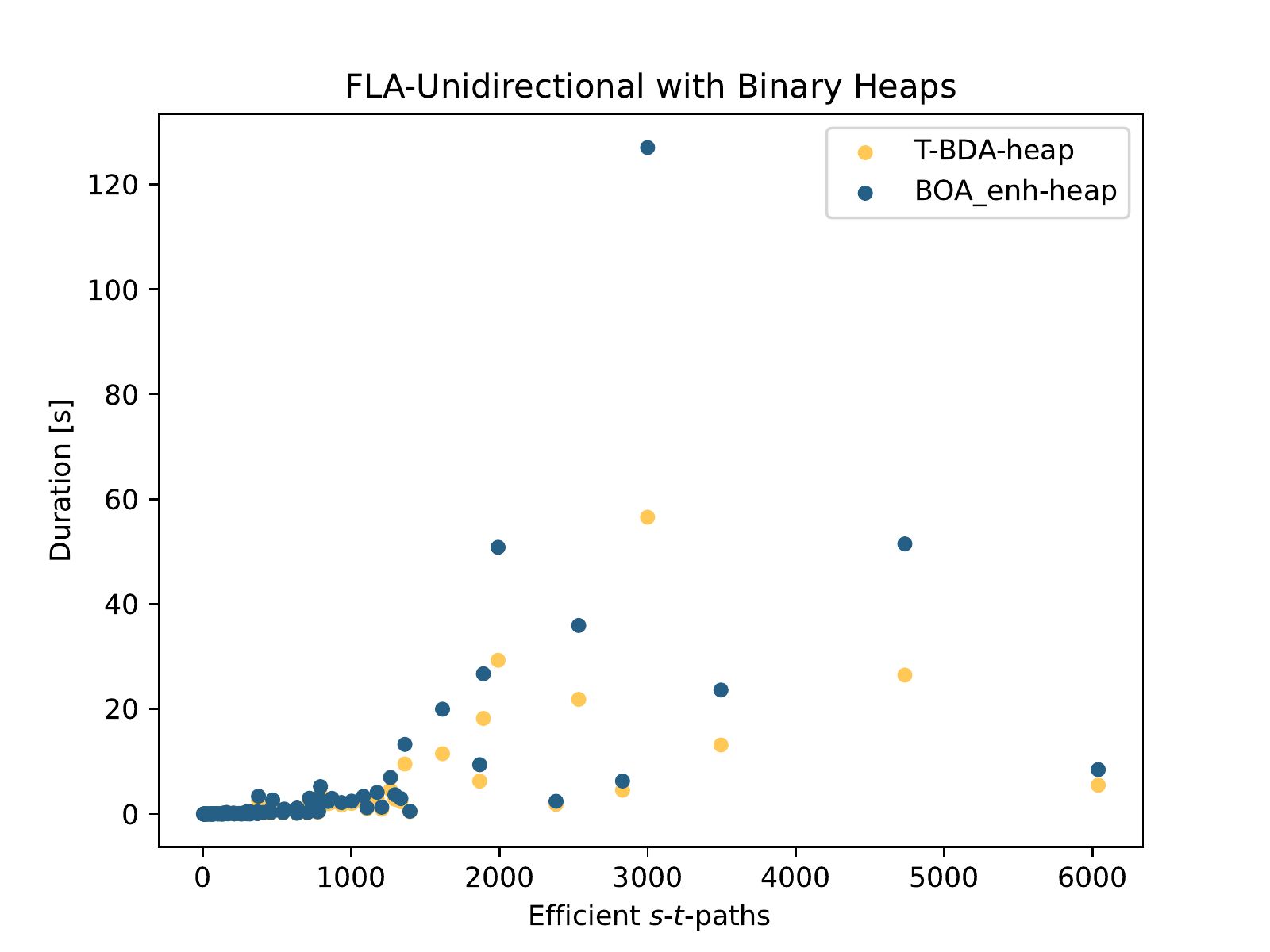}
		\captionof{figure}{}\label{fig:2d-bda-boa-heap-FLA}
	\end{minipage}
\end{figure}
\begin{figure}[H]
	\begin{minipage}{.48\linewidth}
		\captionsetup{type=figure}\includegraphics[width=\textwidth]{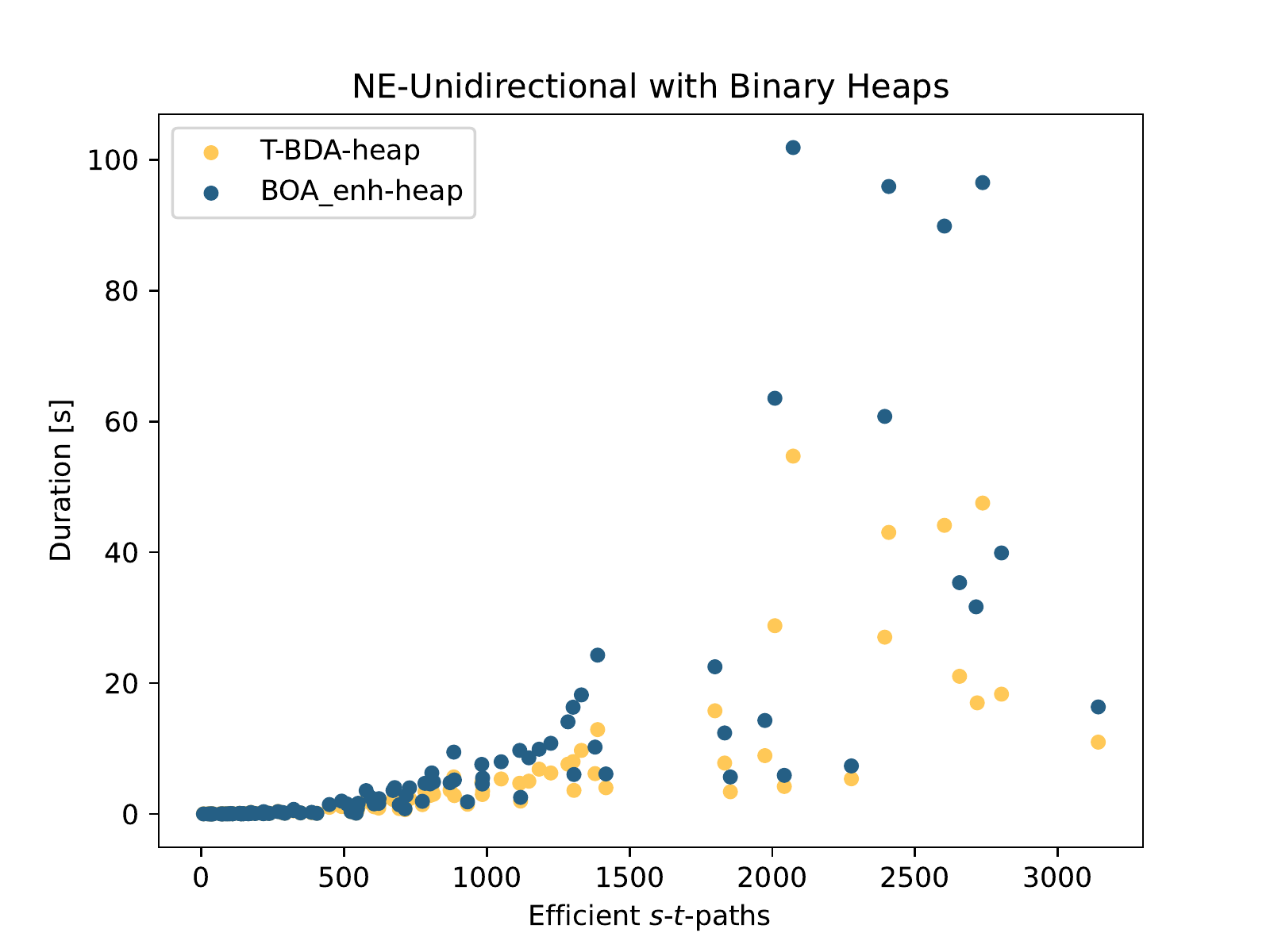}
		\captionof{figure}{}\label{fig:2d-bda-boa-heap-NE}
	\end{minipage}
	\begin{minipage}{.48\linewidth}
		\captionsetup{type=figure}\includegraphics[width=\textwidth]{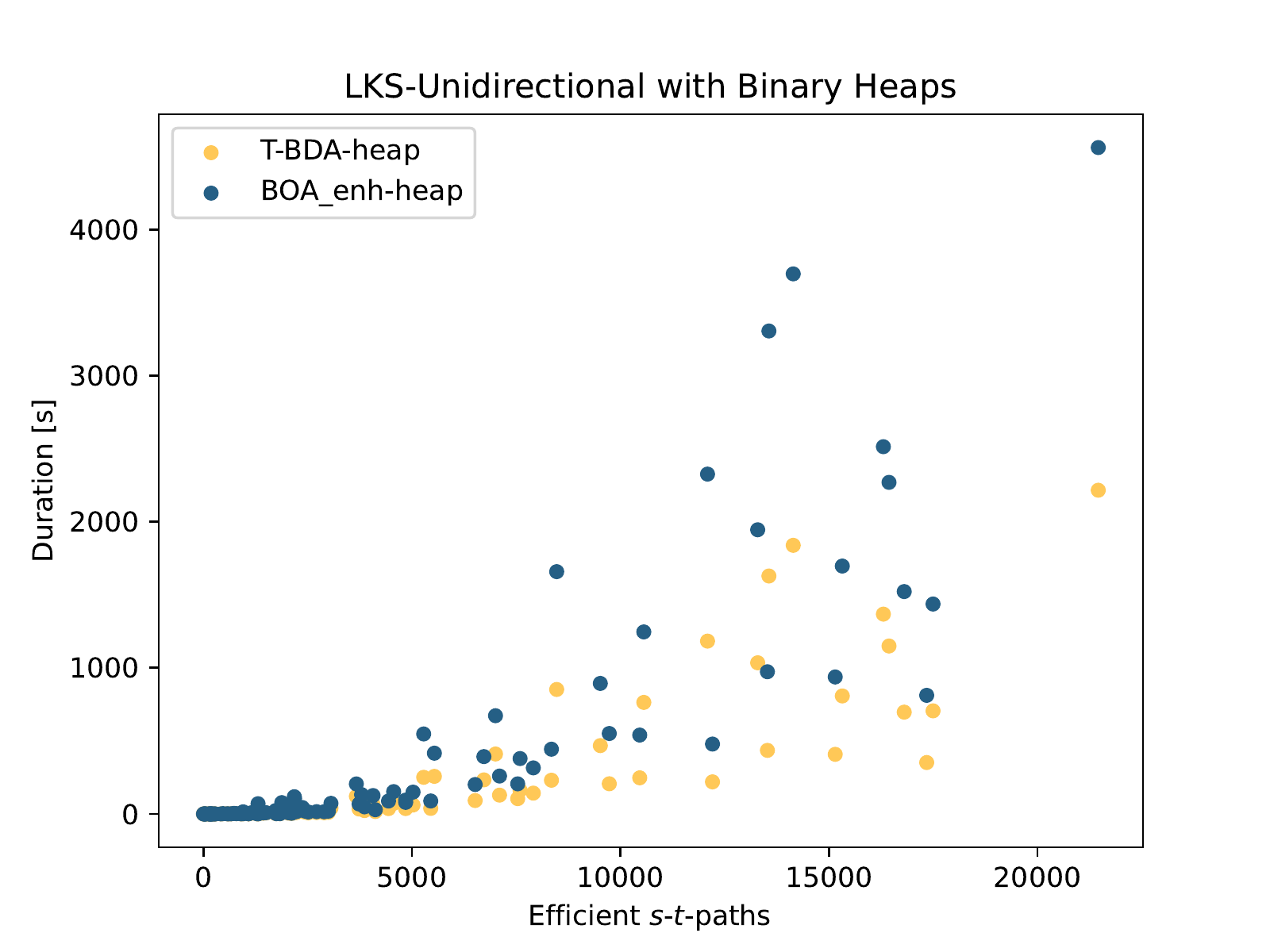}
		\captionof{figure}{}\label{fig:2d-bda-boa-heap-LKS}
	\end{minipage}
\end{figure}
\begin{figure}[H]
	\begin{minipage}{.48\linewidth}
		\captionsetup{type=figure}\includegraphics[width=\textwidth]{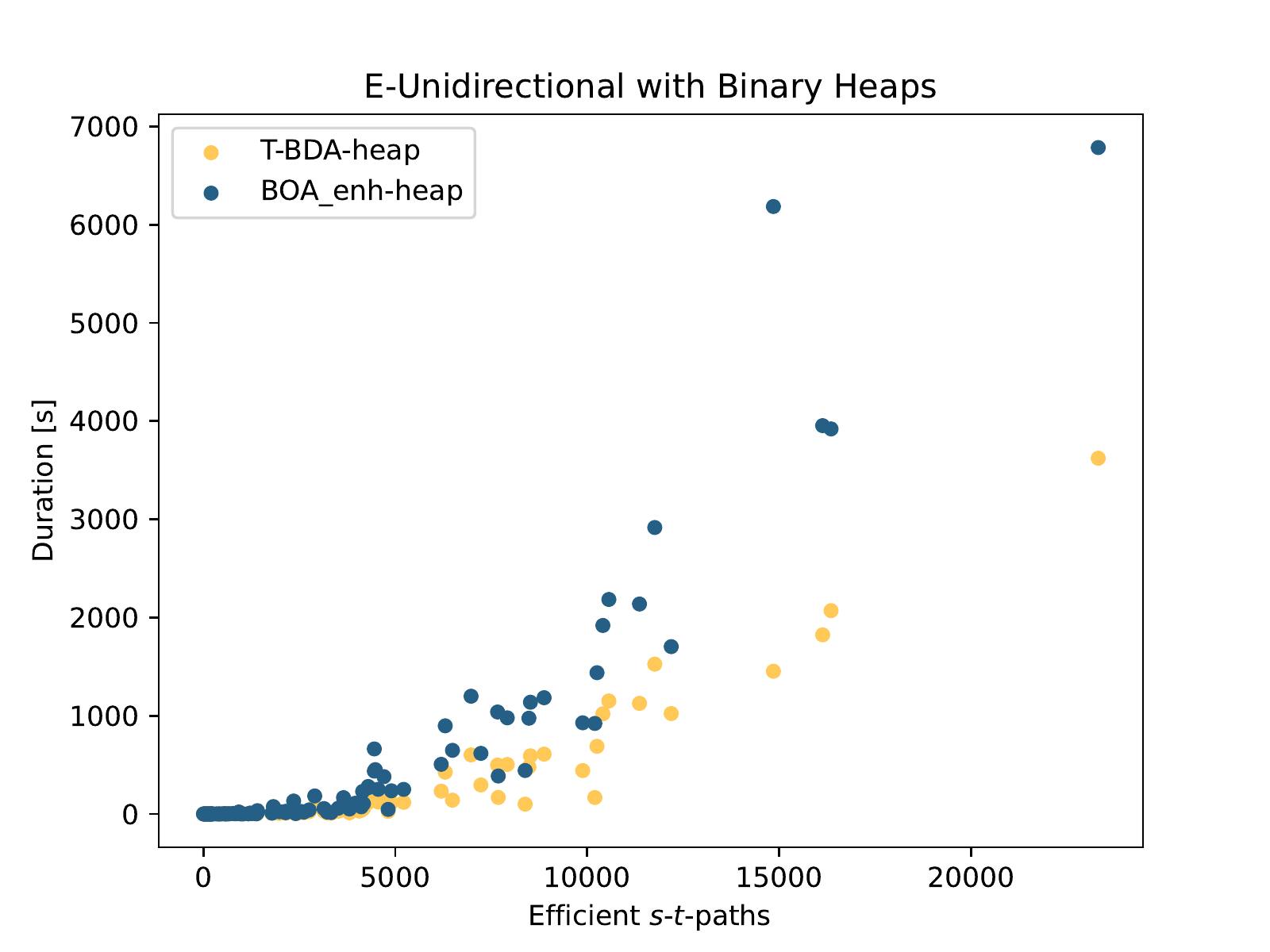}
		\captionof{figure}{}\label{fig:2d-bda-boa-heap-E}
	\end{minipage}
	\begin{minipage}{.48\linewidth}
		\captionsetup{type=figure}\includegraphics[width=\textwidth]{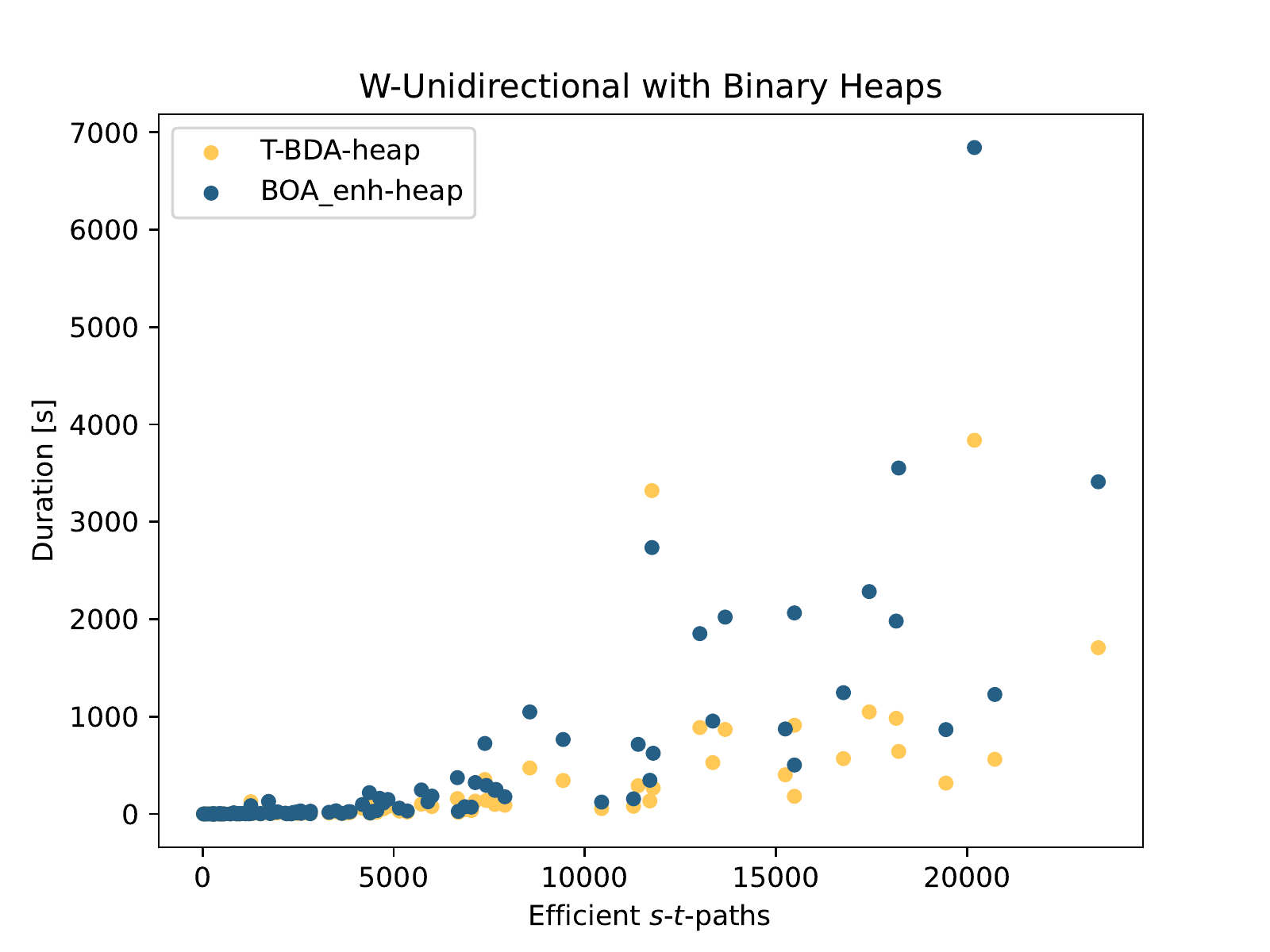}
		\captionof{figure}{}\label{fig:2d-bda-boa-heap-W}
	\end{minipage}
\end{figure}
\begin{figure}[H]
	\begin{minipage}{.48\linewidth}
		\captionsetup{type=figure}\includegraphics[width=\textwidth]{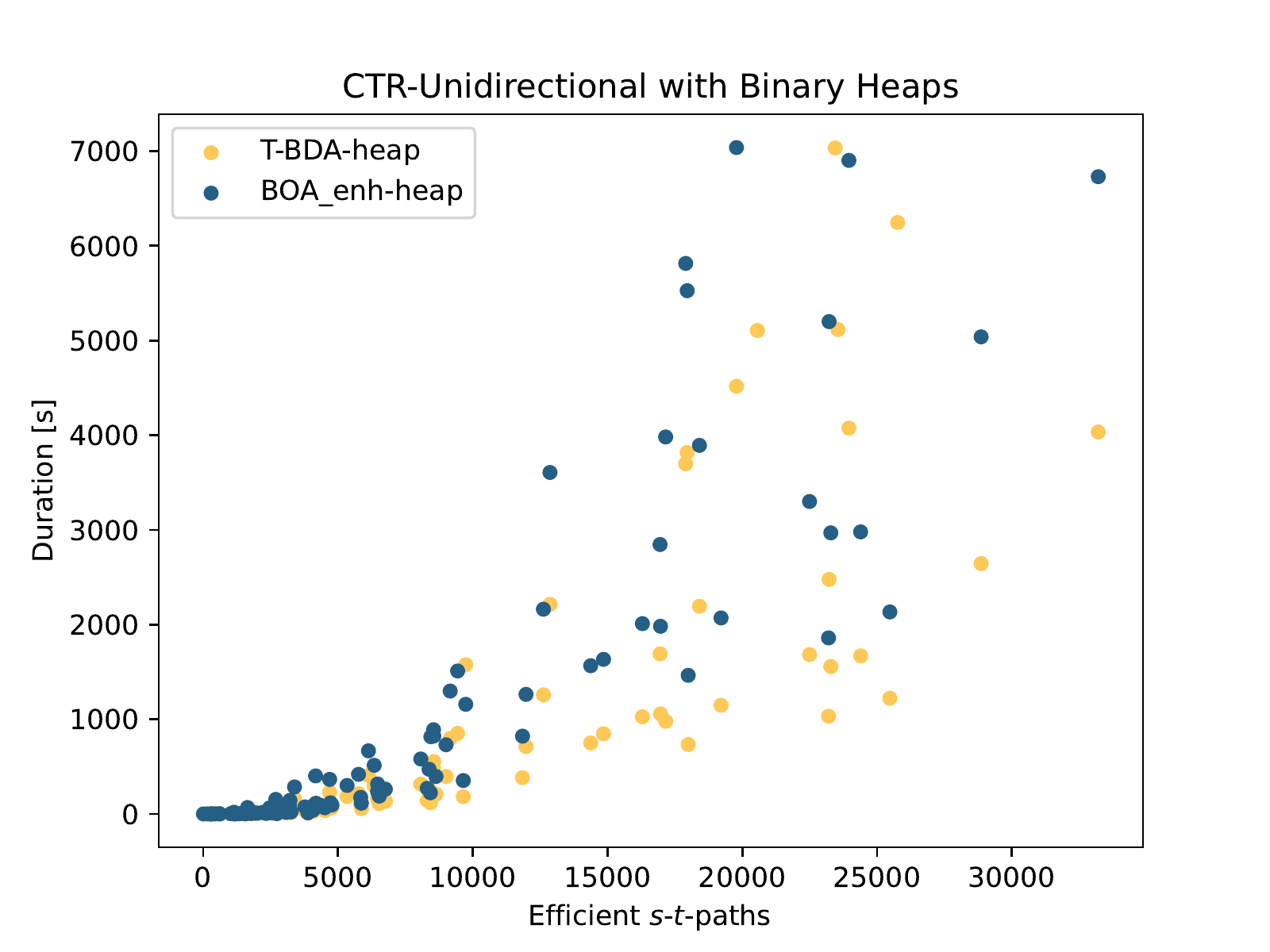}
		\captionof{figure}{}\label{fig:2d-bda-boa-heap-CTR}
	\end{minipage}
\end{figure}
\end{document}